\documentclass[11pt]{article}

\usepackage[T1]{fontenc}
\usepackage[margin=1in]{geometry}
\usepackage{color}
\usepackage{setspace}
\usepackage{amsmath,amsfonts,graphicx,amssymb,bm,amsthm}
\usepackage{appendix}
\usepackage{hyperref}
\usepackage{color}
\usepackage{xcolor}
\usepackage{multirow}
\usepackage{makecell}
\usepackage{placeins}
\usepackage[ruled, vlined]{algorithm2e}
\usepackage[numbers]{natbib}
\usepackage{mathtools}
\usepackage{cleveref}
\usepackage{titlesec}
\crefname{assumption}{Assumption}{Assumption}
\crefname{algorithm}{Algorithm}{Algorithm}

\newcounter{setting}
\crefname{setting}{Setting}{Settings}
\SetKwInput{KwParams}{Parameters}
\makeatletter
\newenvironment{setting}[1][t]{%
  \edef\setting@savedalgocf{\the\value{algocf}}%
  \setcounter{algocf}{\value{setting}}%
  \let\setting@hyrefstep\hyper@refstepcounter
  \def\hyper@refstepcounter##1{\def\theHalgocf{setting.\arabic{algocf}}\setting@hyrefstep{##1}}%
  \crefalias{algocf}{setting}%
  \begin{algorithm}[#1]%
}{%
  \end{algorithm}%
  \setcounter{setting}{\value{algocf}}%
  \setcounter{algocf}{\setting@savedalgocf}%
}
\makeatother

\newcommand{\E}{\mathbb{E}}
\newcommand{\indi}[1]{{\mathbb{I}\left[#1\right]}}

\newcommand{\sumi}{\sum_{i=1}^n}
\newcommand{\sumt}{\sum_{t=1}^T}
\newcommand{\opt}{\mathrm{OPT}}

\newcommand{\vmax}{{\bar{v}}}
\newcommand{\alphalower}{{\underline{\alpha}}}

\newcommand{\muupper}{\bar \mu}
\newcommand{\valuespace}{{[0, \vmax]^n}}
\newcommand{\ppad}{$\mathsf{PPAD}$}

\newcommand{\supp}{{\operatorname{supp}}}

\newcommand{\dd}{\mathrm{d}}

\newcommand{\regret}{\mathrm{Regret}}
\theoremstyle{plain}

\crefname{budgettheorem}{Theorem}{Theorems}
\Crefname{budgettheorem}{Theorem}{Theorems}

\newtheorem{lemma}{Lemma}[section]
\newtheorem{definition}{Definition}[section]
\newtheorem{claim}{Claim}[section]
\crefname{claim}{claim}{claims}
\Crefname{claim}{Claim}{Claims}

\newcommand{\informalname}{}
\newtheorem*{informal*}{\informalname}
\newenvironment{informal}[1]
  {\renewcommand{\informalname}{\Cref{#1}}\begin{informal*}[Informal]}
  {\end{informal*}}

\usepackage{thmtools}

\newif\ifsppefullstatement

\titleformat{\paragraph}[runin]
  {\normalfont\bfseries}  
  {\theparagraph}         
  {1em}                   
  {}                       
  
\title{Impossibility of No-Regret in Multi-Agent Bidding}
\author{Zongjun Yang\thanks{Columbia University. Email: \texttt{zy2684@columbia.edu}}  \qquad Rachitesh Kumar\thanks{Carnegie Mellon University. Email: \texttt{rachitesh@cmu.edu}} \qquad Christian Kroer\thanks{Columbia University. Email: \texttt{ck2945@columbia.edu}}  }
\date{}
\begin{document}

\maketitle
\onehalfspacing
\begin{abstract}
    We consider $T$ repeated second-price auctions across $n$ pacing bidders with \textit{i.i.d.}\ item arrivals, and show that it is computationally impossible for all bidders to achieve no-regret simultaneously while satisfying their respective budget or return-on-investment constraints. Therefore, in the prominent internet auction systems where buyers simultaneously deploy pacing algorithms or automated bidding agents, we cannot hope to achieve joint no-regret guarantees in general. 
    
    We establish a deterministic relationship that high variance in online bidding systems implies high regret. By showing that one can extract a pacing equilibrium from a low-variance system, we prove \ppad-hardness of attaining joint no-regret. For several widely used bidding algorithms, we also show evidence from dynamical systems and numerical simulations that they can indeed display high-variance behavior through cycling, even on simple market instances.
\end{abstract}
 
\section{Introduction}

Online auctions are the dominant mechanism for internet monetization, and play a central role in allocating items in various markets, such as internet advertising, sponsored search, and content recommendation~\citep{edelman2007internet,varian2007position}. In practice, it is increasingly common for auction participants to adopt bidding algorithms, software, or automated bidding agents, which optimize the bids in a sequence of auctions on behalf of the participants under their long-term objectives or constraints, such as budgets or return-on-investment (ROI) targets~\citep{aggarwal2019autobidding,balseiro2021landscape,aggarwal2024auto}.

Online bidding algorithms are well understood in the literature 
from the standpoint of a single agent, where the competing bidders are treated as a part of an exogenous environment. 
In this line of research, a single bidder faces an online learning problem where the competition environment is generated by nature, 
and the algorithm's performance is measured by its \textit{regret}, i.e., the gap to the best outcome the bidder could have obtained in hindsight. 
In the stationary paradigm where items and competing bids are generated in an ideal \textit{i.i.d.}\ manner, 
simple and efficient algorithms are known to achieve regret guarantees that are sublinear in horizon length, for various constraint types and auction formats~\citep{badanidiyuru2018bandits,balseiro2019learning,balseiro2023best,wang2023learning,feng2023online,castiglioni2023online}.
In the worst-case paradigm with adversarial arrivals, there are also results on constant-factor approximation~\citep{balseiro2019learning,castiglioni2022online,balseiro2023best}, despite the fact that any bidding algorithm must suffer from linear regret in general.
In both input paradigms, the competitors' behavior is simulated by the environment in an exogenous manner, and does not respond to the bidder's algorithmic decisions.

In modern internet advertising, a large and growing fraction of buyers delegate their bidding to automated agents, and many of these agents are provided by the auctioneer platform itself~\citep{aggarwal2019autobidding,aggarwal2024auto}. As a result, the competition faced by each bidder is simultaneously affected by the learning algorithms other bidders are using, which adapt to the past auction outcomes and, in turn, depend on the bidder's own decisions.
This multi-agent environment cannot be captured by either the stationary or the adversarial paradigm, and is far less understood. It is unknown what types of guarantees the platform could still promise to the individual participants when offering the bidding agents or software as part of their service, and if centralized coordination between the agents could improve these guarantees.

In this paper, we focus on a scenario with $T$ sequential second-price auctions and $n$ bidders, each with budget or ROI constraints. Here, one dominant practice is to scale the bids by pacing multipliers associated with the constraints, and adjust the scale over time~\citep{balseiro2019learning,conitzer2022multiplicative,balseiro2023best,balseiro2023field,zhang2016feedback,yang2019bid,tashman2020dynamic}.
However, it is known to be computationally difficult, specifically \ppad-hard, to compute a pacing equilibrium where agents use fixed multipliers to scale their bids~\citep{chen2024complexity,chen2025constant,liu2026tight,li2024vulnerabilities}, so it is unlikely for efficient dynamics to converge to such an equilibrium. On the other hand, it is unknown whether non-convergence actually harms each individual bidder, i.e., whether or not the computational hardness of pacing equilibria alone rules out the possibility of achieving \textit{no-regret} guarantees under joint multiplier learning. Indeed, in standard normal-form games, strong no-regret guarantees are known when players simultaneously deploy learning algorithms, despite the \ppad-hardness of computing a Nash equilibrium~\citep{chen2009settling,daskalakis2009complexity,daskalakis2011near,rakhlin2013optimization,syrgkanis2015fast,bailey2019fast,chen2020hedging,daskalakis2021near,anagnostides2022uncoupled,farina2022near}. It is natural to ask if similar results can be found in the online auction setting, and this leads to our central question: 
\begin{quote}
    \textit{Does there exist a pacing algorithm which leads to sub-linear regret when simultaneously deployed by every agent?}
\end{quote}

\subsection{Our Contributions}

\paragraph{High variance implies high regret.}
We first establish the relationship between each bidder's regret and the extent of instability of their own pacing multipliers over time. We show that when an individual bidder's payoff satisfies a smoothness condition, their regret can be lower-bounded by a \textit{variance} term of their own pacing multiplier sequence (\Cref{thm: strongly-convex-regret-decomposition}), plus a slack term and negligible errors. Therefore, high-variance behavior of a bidder's pacing strategy, e.g. cycling and oscillation, implies high regret. This result naturally extends to general online resource allocation systems with multiple agents running primal-dual algorithms. 

Intuitively, we show that non-convergence is indeed costly for individual bidders. This suggests that platforms can often monitor the fluctuation of pacing multipliers as a simple and directly observable indicator of the regret. If the platform aims to achieve low regret for all bidders, the multi-agent system as a whole must behave in a stabilized manner. 

\paragraph{Computational impossibility of joint no-regret.}

Building on the connection between high multiplier variance and high regret, we then establish our main computational hardness result. We show that, unless \ppad-hard problems can be solved by polynomial-time randomized algorithms, it is impossible to guarantee sublinear regret for all bidders while satisfying their budget or ROI constraints using efficient dynamics.

\begin{informal}{thm: fixed-polynomial-regret,thm: fixed-polynomial-ROI-regret}
Unless \ppad-hard problems can be solved by polynomial-time randomized algorithms, no polynomial-time multi-agent bidding dynamic can guarantee $O(n^{-\varepsilon}T)$ regret and constraint violation for every bidder, for any fixed $\varepsilon>0$.
\end{informal}

In other words, the regret or constraint violation of some bidder must grow linearly in the number of auctions $T$, which is asymptotically no better than a single bidder facing an adversarially chosen environment.\footnote{The result applies to any horizon length $T$ that is polynomial in $n$. With exponentially many auctions, the bidders could afford to compute an equilibrium by brute force; see the discussion in \Cref{subsection: budget-1}.} Hence, the sublinear regret guarantees known for a single bidder in a stationary \textit{i.i.d.} environment cannot be extended to the competitive multi-agent setting in general. Moreover, since $\varepsilon$ can be arbitrarily small, we cannot expect the regret to decrease substantially when the number of participants in the market becomes large.

From a practical perspective, our result shows that a platform offering autobidding services cannot, in general, promise every advertiser an outcome close to the best one they could have obtained in hindsight. Our results also hold for centralized dynamics, where the platform could potentially coordinate the pacing multipliers across bidders within a round. Theoretical guarantees for such pacing systems thus need to rely on additional market assumptions, or target alternative aggregate objectives.

\paragraph{Dynamic system evidence of high variance.}
In addition to our theoretical results, we provide numerical evidence of high variance on a simple and intuitive instance with only $n=3$ bidders, using the construction of \citet{paes2024complex}. We simulate both utility maximizers with budgets and value maximizers with ROI targets, under two families of widely used algorithms, namely PID controllers~\citep{balseiro2023field,tashman2020dynamic,yang2019bid,zhang2016feedback,smirnov2016online} and online dual mirror descent~\citep{balseiro2023best}. In all cases, the bid multipliers cycle, causing both the variance and the regret of every bidder to grow linearly over time. This is consistent with our observation in \Cref{thm: strongly-convex-regret-decomposition}.

\subsection{Related Work}
\paragraph{Computational complexity of pacing equilibria.}
Our computational hardness results build on a line of work on the complexity of pacing equilibria in second-price auctions. For budget-constrained, utility-maximizing bidders, \citet{conitzer2022multiplicative} introduce second-price pacing equilibrium (SPPE), a steady state in which every bidder scales their values by a multiplier that best responds to the other bidders under their budget constraint. They show that an SPPE always exists, while computing one that maximizes revenue or welfare is NP-hard. \citet{chen2024complexity} show that computing an SPPE is $\mathsf{PPAD}$-complete, and the inapproximability is later strengthened to any constant factor~\citep{chen2025constant,liu2026tight}. For ROI-constrained, value-maximizing bidders, \citet{li2024vulnerabilities} consider the analogous notion called autobidding equilibrium, and show its existence and $\mathsf{PPAD}$-hardness of approximation. In our online setting, their \ppad-hardness results imply that no efficient online algorithm can converge to SPPE (or autobidding equilibrium), but cannot rule out the possibility of converging to other low-regret states or non-convergence result with regret guarantees. Recent works also find exact algorithms for markets when there are only a few buyers or goods~\citep{yan2026pacing,huang2026pacing}. More broadly, the notion of pacing equilibrium is also studied in throttling~\citep{chen2021throttling,liu2026tight}, first-price auctions~\citep{conitzer2022pacing,yang2024online}, and more general mechanisms~\citep{afiouni2026pacing}. 

\paragraph{Dynamical systems for autobidding.}
Recent works study automated bidding systems as continuous-time dynamical systems. \citet{paes2024complex} consider ROI-constrained value maximizers whose multipliers follow PID-controller dynamics, and show that stable cycles can occur with three or more bidders. We extend their findings to the discrete-time setting for dual mirror descent algorithms, and to budget-constrained bidders. We also show in our experiments that the cycles indeed incur high regret for every bidder under both constraint types. \citet{anagnostides2026chaos} further show in the ROI setting that the same dynamics from \citep{paes2024complex} can simulate a broad class of nonlinear systems, including chaotic ones, and that their discrete-time analogues can indeed be chaotic. Although these works do not establish theoretical results for all online algorithms, they provide evidence of high variance in common multi-agent bidding systems, which by our \Cref{thm: strongly-convex-regret-decomposition} translate into high regret.

\paragraph{Learning in games.}
In repeated games without constraints, no-regret guarantees are known in the multi-agent learning setting. When all players run optimistic learning dynamics, the joint play changes slowly and becomes predictable, so each player's regret can be much smaller than in the adversarial case. Such guarantees were first established for two-player zero-sum games~\citep{daskalakis2011near,rakhlin2013optimization}, and later extended to general-sum games~\citep{syrgkanis2015fast,chen2020hedging,daskalakis2021near,soleymani2025faster} and convex games~\citep{farina2022near}, as well as to the stronger notion of swap regret~\citep{anagnostides2022near,anagnostides2022uncoupled}. In auctions without constraints, \citet{feng2021convergence} show that mean-based no-regret bidders converge to the equilibrium in repeated second-price and first-price auctions under certain conditions, while \citet{kolumbus2022auctions} and \citet{deng2022nash} show that convergence may fail without these conditions. 
Different from these works, in our setting the players have long-term constraints on their actions, which are budgets or ROI targets. With these constraints, no algorithm can guarantee sublinear regret against adversarial competing bids~\citep{balseiro2019learning,immorlica2022adversarial,balseiro2023best}. Hence, unlike in unconstrained games, joint no-regret cannot follow from individual worst-case guarantees.

\paragraph{Positive results on multi-agent bidding.}
When all bidders adopt pacing algorithms simultaneously, positive results are known under additional market structure assumptions or for objectives other than individual regret. \citet{balseiro2019learning} consider budget-constrained bidders whose valuations are independent, and show that adaptive pacing can converge under a strong monotonicity assumption, a joint condition across all bidders that can fail in many instances. They also show that adaptive pacing constitutes an approximate Nash equilibrium in large markets. Without requiring convergence, \citet{gaitonde2022budget} show that if all bidders use gradient-based pacing, the liquid welfare, a group-level welfare metric, is at least half of the optimum for a broad class of auctions. Our impossibility result directly answers the question left open by \citet{gaitonde2022budget} on whether all agents can achieve vanishing regret at once. For individual bidders, \citet{gaitonde2022budget} also measure a different notion of regret against an alternative benchmark, defined as the sequence of multipliers under which the expected spending in each round equals the per-round budget, which need not be optimal in hindsight. 
\citet{lucier2023autobidders} extend the liquid welfare guarantee to bidders with both budget and ROI constraints. \citet{fikioris2023liquid} obtain liquid welfare guarantees for budget-constrained bidders in first-price auctions, and \citet{feng2024strategic} bound the liquid welfare when advertisers strategically choose the budgets of their autobidders. For a different objective, \citet{fikioris2026robust} design a budget-pacing algorithm under which, when adopted by all bidders, every bidder wins a number of auctions proportional to their budget. 

\section{Preliminaries}
\label{section: preliminaries}

\paragraph{Multi-agent bidding in repeated auctions.}

We consider the market where $n$ bidders compete for $T$ sequential items via repeated second-price auctions. 
At each round $t\in[T]$, an item arrives, and each bidder $i\in[n]$ observes their private value $v_{i, t}\geq 0$ for the current item, and submits a bid $b_{i,t}\geq 0$. We let $\vec v_t =(v_{1,t}, \cdots, v_{n,t})$ denote the value vector of agents, and assume that all values are bounded by a constant $\vmax >0$. The item is then allocated to the highest bidder; let $x_{i,t}$ be the allocation bidder $i$ receives. In the case of a tie, the item can be divided arbitrarily across the top bidders.\footnote{Our impossibility results will always hold regardless of how the decision maker breaks ties.} Let $h_t$ and $p_t$ be the highest and second highest among bids $(b_{i,t})$. In round $t$, the bidder gains value $u_{i,t} = x_{i,t}v_{i,t}$ and pays $c_{i,t} = x_{i,t} p_t$. 

The bidders may seek to maximize different objectives under different constraints on their total expenditure. We consider the setting where all bidders belong to one of the following two well-studied types:
\begin{itemize}
    \item \textit{Utility maximizers with budgets.} Each bidder $i$ aims to maximize their quasi-linear utility gain $\sumt (u_{i,t} - c_{i,t}) $, subject to a budget constraint $\sumt c_{i,t} \leq T \rho_i $, which bounds the total amount of money spent in the auctions; $\rho_i$ is bidder $i$'s per-round budget. 
    \item \textit{Value maximizers with return-on-investment (ROI) targets.} Each bidder $i$ aims to maximize their total value derived from the items, $\sumt u_{i,t}$. Each bidder $i$ has an ROI target $r_i>0$, which requires the total value gained to be at least $r_i$ times the total expenditure, i.e., $r_i \sumt c_{i,t} \leq \sumt u_{i,t}$. 
\end{itemize}

We consider the stochastic setting where value vectors across rounds are independent and identically distributed (\textit{i.i.d.}), while allowing correlations across bidders within a round; see \Cref{alg: multi-agent-pacing} for details. We let $F$ be the distribution of each item value vector $\vec v_t$, and denote the entire sequence of item arrivals as $\bm v$. 

\paragraph{Bid-scaling.}
We restrict the bidders' strategy to \textit{pacing}, where bidder $i$ commits to a \textit{pacing multiplier} $\alpha_{i,t}$ at the beginning of round $t$ \textit{before} observing the private value, and bids $b_{i,t} = \alpha_{i,t} v_{i,t}$, i.e., scales their value with the multiplier $\alpha_{i,t}$. Notably, we also allow $\alpha_{i,t}$ to be chosen in a centralized manner before the valuations are revealed. We let $\vec \alpha_t = (\alpha_{1,t}, \cdots, \alpha_{n,t})$ be the collection of pacing multipliers in the $t$'th auction, and let $\vec \alpha_{-i, t}$ be the multipliers of bidders $j\neq i$. For a general value vector $\vec v$ and pacing multiplier profile $\vec \alpha$, we write $h(\vec v; \vec \alpha)$ and $p(\vec v; \vec \alpha)$ for the highest and second-highest bids among $(\alpha_i v_i)_{i\in [n]}$, respectively. 

Pacing is a dominant practice for bidding with budgets or ROI targets~\citep{balseiro2023field,aggarwal2024auto}, and is optimal up to an additive error of $\vmax$, i.e., a single item. For utility maximizers with budgets, we denote $\alphalower_i = \rho_i/\vmax$ and assume without loss of generality that all pacing multipliers from bidder $i$ are in the range $[{\alphalower}_i, 1]$, which always contains a hindsight optimal pacing multiplier. Similarly, for value maximizers with ROI constraints, we assume without loss of generality that all pacing multipliers from bidder $i$ are in the range $(\alphalower_i, +\infty)$ for $\alphalower_{i}  = 1/r_i$.

\begin{setting}[t]
\caption{Multi-Agent Bidding in Repeated Second-Price Auctions}
\label{alg: multi-agent-pacing}
\KwParams{Number of bidders $n$, total time periods $T$, per-round budgets $(\rho_i)_{i\in[n]}$ or ROI targets $(r_i)_{i\in[n]}$, value distribution $F$.}
\For{$t=1,\ldots,T$}{
    Each bidder $i\in [n]$ chooses a pacing multiplier $\alpha_{i, t}$. 

    The $t$'th value vector is generated independently, $\vec v_t = (v_{1,t}, \cdots, v_{n,t}) \sim F$. 

    Each bidder $i\in [n]$ submits the bid $b_{i, t} = \alpha_{i, t} v_{i, t}$.  

    Let $h_t$ and $p_t$ be the highest and second-highest bids. 

    Choose allocation $(x_{i,t})_{i\in[n]}$ such that a highest bidder gets the item (arbitrary tie-breaking).

    Each bidder $i$ receives allocation $x_{i,t}$, gains value $u_{i,t} = v_{i,t} x_{i,t}$, and pays $c_{i,t}= p_t  x_{i,t}$.
}
\end{setting}

\paragraph{An individual bidder's regret.} 

We characterize a single bidder's optimal reward in hindsight, when all item values and the competing bids in all rounds are fixed. Let the highest competing bid faced by bidder $i$ be $d_{i,t} = \max_{j\neq i} \alpha_{j,t }v_{j,t}$ in round $t$. Notice that conditional on the auction winner, the utility gain and the payment no longer depend on $i$'s own bid. Thus the hindsight optimization program can be simplified to a knapsack-type problem with binary decisions on whether to win each item or not. For a fixed arrival instance $\bm v$ and competing pacing multiplier sequence $\bm \alpha_{-i} = (\vec \alpha_{-i,t})_{t\in [T]}$, in the case with utility maximizers and budget constraints, the hindsight benchmark for bidder $i$ is given by
\begin{equation}
    \opt_i ( \bm v; \bm \alpha_{-i}) =  \max_{x_{i,t} \in \{0, 1\}} \left\{ \sumt (v_{i,t} -d_{i,t}) x_{i,t} : \sumt x_{i,t}  d_{i,t}\leq T\rho_i\right\}.
    \label{eq: primal}
\end{equation}
For value maximizers with ROI constraints, the hindsight optimum is given by
\begin{equation}
    \opt_i ( \bm v; \bm \alpha_{-i}) =  \max_{x_{i,t} \in \{0, 1\}} \left\{ \sumt v_{i,t} x_{i,t} : r_i\cdot \sumt x_{i,t}  d_{i,t}\leq \sumt x_{i,t} v_{i,t}\right\}.
    \label{eq: primal-roi}
\end{equation}

For notational simplicity we will denote the objective in \eqref{eq: primal} and \eqref{eq: primal-roi} as $\sumt f_{i,t}(x_{i,t})$, where $f_{i,t}(x_{i,t}) = (v_{i,t} - d_{i,t}) x_{i,t}$ for utility maximizers, and $f_{i,t}(x_{i,t}) = v_{i,t} x_{i,t}$ for value maximizers. An individual bidder's \textit{regret} is defined as the additive gap between the hindsight optimal objective and the actual utility (or value) derived from the auctions.  

\begin{definition}[Regret]
    For a fixed arrival instance $\bm v$ and a multi-agent bidding dynamic $\mathcal{A}$ for the stochastic setting, let $(\vec \alpha_t)_{t\in [T]}$ and $(x_{i,t})_{i\in[n], t\in [T]}$ be the pacing multipliers and allocations. Then the \textit{regret} of bidder $i\in[n]$ is defined as 
     $   \regret_{i, \mathcal{A}}(\bm v) = \opt_i(\bm v; \bm \alpha_{-i})  - \sumt f_{i,t}(x_{i,t}).
    $
\end{definition}

\paragraph{The fluid setting.} 

\begin{setting}[t]
\caption{Multi-Agent Bidding in Repeated Fluid-Item Second-Price Auctions}
\label{alg: multi-agent-pacing-fluid}
\KwParams{Number of bidders $n$, total time periods $T$, per-round budgets $(\rho_i)_{i\in[n]}$ or ROI targets $(r_i)_{i\in[n]}$, value distribution $F$.}
\For{$t=1,\ldots,T$}{
    Each bidder $i\in [n]$ chooses a pacing multiplier $\alpha_{i, t}$. 
    
    \For{\textnormal{each value instance $\vec v = (v_1, \cdots, v_n) \in \supp(F)$}
    }{
        Let $h( \vec v; \vec \alpha_t)$ and $p(\vec v; \vec \alpha_t)$ be the highest and second-highest bids among $(\alpha_{i,t} v_i)_{i\in [n]}$. \\
        Choose allocation $(x_{i,t}(\vec v))_{i\in [n]}$ such that a highest bidder gets the item (arbitrary tie-breaking). 
    }
    
    Each bidder $i$ gains value $u_{i, t}^F = \int v_{i}  x_{i, t}(\vec v) \dd F(\vec v)$ and pays $ c_{i, t}^F = \int p(\vec v; \vec \alpha_t) x_{i, t}(\vec v) \dd F(\vec v)$.
}
\end{setting}

A closely related setting is where auctions repeat on $T$ copies of a \textit{fluid} item, where the density (or mass) of each value vector corresponds to the density (or mass) in the distribution $F$; see \Cref{alg: multi-agent-pacing-fluid} for details. In each fluid-item auction, each bidder submits a bid $b_{i,t}(\vec v)$ for every value point $\vec v = (v_1, \cdots, v_n)\in \supp(F)$. Here, we also restrict bidders' strategy to pacing, so in auction $t$ bidder $i$ chooses a pacing multiplier $\alpha_{i,t}$ and bids $b_{i,t}(\vec v) = \alpha_{i,t} v_i$ uniformly over the fluid.  Allocation $x_{i,t}(\vec v)$ and price are determined point by point based on the highest bid $h(\vec v; \vec \alpha_t)$ and second-highest bid $p (\vec v; \vec \alpha_t)$. Each agent gains value $u_{i,t}^F = \int v_ix_{i,t}(\vec v) \dd F(\vec v)$, and pays $c_{i,t}^F = \int p(\vec v; \vec \alpha_t) x_{i,t}(\vec v) \dd F(\vec v)$. 

We also define an individual bidder's hindsight benchmark and regret in the fluid setting, given other bidders' pacing multipliers $\bm \alpha_{-i}=(\vec\alpha_{-i,t})_{t\in[T]}$. For a utility maximizer with a budget constraint, 
\begin{equation}
    \begin{aligned}
       \opt_i^F(\bm \alpha_{-i})=\max_{x_{i,t}(\vec v)\in\{0,1\}}\quad
       &\sumt\int x_{i,t}(\vec v)
       \left(v_i-\max_{j\neq i}\alpha_{j,t}v_j\right)\dd F(\vec v) \\
       \textnormal{s.t.}\quad
       &\sumt\int x_{i,t}(\vec v)
       \left(\max_{j\neq i}\alpha_{j,t}v_j\right)\dd F(\vec v)
       \leq T\rho_i.
    \end{aligned}
    \label{eq: fluid-hindsight-budget}
\end{equation}
For a value maximizer with an ROI constraint,
\begin{equation}
    \begin{aligned}
       \opt_i^F(\bm \alpha_{-i}) =  \max_{x_{t,i}(\vec v) \in \{0, 1\}} \quad & \sumt \int x_{t,i}(\vec v)  \cdot v_{i} \   \dd F(\vec v) \\
        \textnormal{s.t.} \quad & r_i \sumt \int x_{i,t}(\vec v)  \left(\max_{j\neq i} \alpha_{j, t} v_{j}\right) \dd F(\vec v) \leq \sumt  \int x_{i,t}(\vec v) \cdot v_i \ \dd F(\vec v) .
    \end{aligned}
    \label{eq: fluid-hindsight-roi}
\end{equation}
For simplicity, let $f_{i,t}^F (x_{t,i}(\vec v))$ be the per-round objective of bidder $i$. We also define a bidder's regret in the fluid setting. 

\begin{definition}[Fluid regret]
For a multi-agent bidding dynamic $\mathcal{A}$ for the fluid setting, let $\bm \alpha$ and $(x_{i,t}(\vec v))_{\vec v\in \supp(F), i\in [n], t\in[T]}$ be the pacing multipliers and allocations. Then the \textit{fluid regret} of bidder $i\in[n]$ is defined as
    $  \regret_{i, \mathcal{A}}^F = \opt_i^F(\bm \alpha_{-i})  - \sumt f^F_{i,t}(x_{i,t}(\vec v)).$
\end{definition}

\section{A General Regret Lower Bound}
\label{section: regret-lower-bound}

We establish a general deterministic lower-bound decomposition of each individual bidder's regret. We frame the individual bidder's hindsight benchmark in a general online optimization framework: 
\begin{equation}
\label{eq: hindsight-resource-allocation}
        \mathrm{OPT}
        = \max_{x_t\in\mathcal X}\left\{
         \sumt f_t(x_t) 
        :
         \sumt a_t(x_t)\leq T  \rho\right\}, 
\end{equation}
where $\mathcal{X}$ is a general decision space, $(f_t, a_t)$ are reward and cost functions that are revealed sequentially, and $\rho\geq 0$ is a general per-round constraint capacity. We drop all bidder subscripts $i$ for simplicity. Aligned with our auction setting, we assume the program is feasible and $0\leq f_t(\cdot )\leq \vmax$. The dual program of \eqref{eq: hindsight-resource-allocation} is as follows,
\begin{equation}
\label{eq: resource-allocation-dual}
    D^* = \inf_{\mu \geq 0} \sumt D_t(\mu), \quad \textnormal{where} \quad D_t(\mu) = \rho \mu + f_t^*(\mu),
\end{equation}
where $f_t^*(\mu)=\max_{x\in\mathcal X}
\{f_t(x)-\mu a_t(x)\}$. When the reward and cost spaces are compact, the duality gap $\Delta = D^* - \mathrm{OPT}$ is at most $\vmax$. 

\begin{restatable}[Bounded duality gap]{proposition}{dualitygapbound}
\label{prop: one-item-duality-gap}
When $\{(f_t(x), a_t(x)):x\in \mathcal{X}\}$ is compact for all $t \in[T]$, the duality gap satisfies $ 0 \leq \Delta \leq \vmax$.
\end{restatable}

We derive the dual functions for our specific bidding setting in \Cref{app: general-proofs} and show that they indeed satisfy the above \Cref{prop: one-item-duality-gap}, so the duality gap can be essentially ignored.  

\paragraph{Primal-dual algorithms.}
We interpret pacing algorithms as primal-dual algorithms for online optimization. A primal-dual algorithm for optimizing \eqref{eq: hindsight-resource-allocation} is described in \Cref{alg: primal-dual}. The algorithm maintains a dynamic dual iterate $\mu_t$. At each step $t\in [T]$, the algorithm chooses an allocation $x_t$ that corresponds to the maximizer in $f_t^*(\mu)$ and then updates the dual iterate.
\begin{restatable}{proposition}{multiplierduals}
\label{prop: multipliers}
For each bidder $i$, bidding $b_{i,t}=\alpha_{i,t}v_{i,t}$ is equivalent to running \Cref{alg: primal-dual} with
\begin{equation*}
    \begin{cases}
        \mu_{i,t}=1/\alpha_{i,t}-1  & \textnormal{(budget constraint, utility maximizer)}\\
        \mu_{i, t} = 1/({\alpha_{i,t} r_{i} - 1}) & \textnormal{(ROI constraint, value maximizer)}
    \end{cases}.
\end{equation*}
\end{restatable}

\begin{algorithm}[t]
\caption{Primal-Dual Algorithm for Online Optimization}
\label{alg: primal-dual}
\KwIn{Total time periods $T$, average constraint capacity $\rho$, initial dual iterate $\mu_1$.}
\For{$t=1,\ldots,T$}{
    Choose allocation $x_t \gets \arg \max_{x \in \mathcal{X}_t} \left\{ f_t(x) - \mu_t a_t(x)\right\}.$

    Compute the next dual iterate $\mu_{t+1}$.
}
\end{algorithm}

It is useful to consider the dual iterates $(\mu_1, \cdots, \mu_T)$ as decisions for a dual-only online learning problem, with loss functions being $D_1, \cdots, D_t$. 
We define the \textit{dual regret} \textit{w.r.t.} $\mu \in \mathbb{R}_+$ as follows, 
\begin{equation*}
    R_D (\mu) = \sumt D_t(\mu_t) - \sumt D_t(\mu).
\end{equation*}

Notice that a subgradient of $D_t$ at $\mu_t$ is $(\rho - a_t(x_t))$. Accordingly, we can define the linearized dual loss as $\widetilde{D}(\mu) = (\rho - a_t(x_t)) \cdot \mu$. The \textit{linearized dual regret} is then given by
\begin{equation*}
    \widetilde R_D (\mu) = \sumt \widetilde D_t(\mu_t) - \sumt \widetilde D_t(\mu).
\end{equation*}

The convexity of the Lagrangian dual ensures that $D_t$ is convex, so the dual regret is always no larger than the linearized dual regret; see the following \Cref{prop: linearization-gap-nonnegativity}. 
\begin{restatable}{proposition}{linearizationgap}
For any $(f_t, a_t)_{t\in [T]}$ and  $\mu \geq 0$, we have $R_D (\mu) \leq \widetilde R_D (\mu)$.
\label{prop: linearization-gap-nonnegativity}
\end{restatable}

\paragraph{The main regret decomposition.}
We show that the primal regret of any primal-dual algorithm can be decomposed into three terms as shown in \Cref{lem: regret-to-dual-regret}: 1) the linearization gap of the dual regret, i.e., the non-negative gap between $\widetilde R_t$ and $R_t$ \textit{w.r.t.} the hindsight optimal dual, 2) the complementary slack, and 3) the duality gap. \Cref{lem: regret-to-dual-regret} is a deterministic decomposition that holds broadly for any primal-dual decision rules for general online resource allocation problems.   

\begin{restatable}{lemma}{regretdecomposition}
\label{lem: regret-to-dual-regret}
Let $\mu^*$ be the optimal solution to \eqref{eq: resource-allocation-dual}. Then for any primal-dual algorithm on any instance $(f_t, a_t)_{t\in[T]}$, we have
\begin{equation*}
    \mathrm{OPT}-\sumt f_t(x_t)
    =
    \underbrace{\widetilde R_D(\mu^*)-R_D(\mu^*)}_{\textnormal{regret linearization gap}}
    +\underbrace{ \mu^* \left(\rho T-\sumt a_t(x_t)\right)}_{\textnormal{complementary slack}}-\underbrace{\Delta}_{\textnormal{duality gap}}.
\end{equation*}
\end{restatable}

When the dual losses are strongly convex, the regret linearization gap can be further lower-bounded by the \textit{variance} of the dual; see the following \Cref{thm: strongly-convex-regret-decomposition}.

\begin{restatable}{theorem}{stronglyconvexdecomposition}
\label{thm: strongly-convex-regret-decomposition}
When the dual losses $D_t$ are $\sigma$-strongly convex on an interval containing $\mu_1,\ldots,\mu_T$ and $\mu^*$, it holds that
\begin{equation*}
    \mathrm{OPT}-\sumt f_t(x_t)
    \geq \underbrace{\frac{\sigma}{2}\sumt \left(\mu_t - \frac{1}{T}\sum_{\tau=1}^T \mu_\tau\right)^2}_{\textnormal{variance}}
    +\underbrace{\mu^* \left( \rho T - \sumt a_t(x_t) \right)}_{\textnormal{complementary slack}}
    -\underbrace{\Delta}_{\textnormal{duality gap}}.
\end{equation*}
\end{restatable}

Since the dual variables can be mapped to the pacing multipliers $(\alpha_{i,t})_{t\in [T]}$, \Cref{thm: strongly-convex-regret-decomposition} directly relates the regret of a bidder to the variance in their sequence of decision variables. 
The strong convexity of the dual losses can be read as a smoothness condition on the primal problem. By convex duality, $D_t$ is $\sigma$-strongly convex when $\max_{x\in\mathcal X}\{f_t(x): a_t(x)\leq c\}$ is concave and $(1/\sigma)$-smooth in $c$, i.e., the bidder's reward is smooth in their spending. We will construct fluid settings where the dual is strongly convex for all $t\in [T]$. 

Since a constraint-feasible bidder always has a non-negative slack, by \Cref{thm: strongly-convex-regret-decomposition}, high variance in the pacing multiplier sequence will then imply high regret. This message will be useful in two ways. First, it enables our result on the computational difficulty of achieving low variance and slack, and thus joint no-regret. Meanwhile, we provide concrete numerical evidence where common bidding algorithms display cycling behavior under simple input constructions. Since the variance grows linearly over time as the state cycles, by \Cref{thm: strongly-convex-regret-decomposition} this gives further evidence that the regret can be large for a broad family of instances. 

\section{Impossibility for Utility Maximizers with Budgets}
\label{section: main-section-on-SPPE}

We first focus on the case where agents are utility maximizers with budgets. In \Cref{subsection: budget-1}, we give theoretical computational complexity results of achieving joint no-regret guarantees. In \Cref{subsection: budget-2}, we provide numerical evidence that common bidding algorithms can display high variance and incur high regret. 

\subsection{Computational Complexity of Joint No-Regret}
\label{subsection: budget-1}
We establish the main computational impossibility result: unless \ppad-hard problems can be solved by polynomial-time randomized algorithms,\footnote{Specifically, for any \ppad-hard problem, no polynomial-time randomized algorithm outputs a solution with probability at least $2/3$. This is the standard randomized version of the assumption $\mathsf{PPAD}\neq\mathsf{P}$~\citep{chen2009settling,foster2023hardness}. \ppad-hard problems are widely believed to be difficult and include problems such as computing Nash equilibrium~\citep{daskalakis2009complexity,chen2009settling}, for which no efficient algorithms are known.} in both stochastic and fluid settings, it is impossible that all bidders are no-regret and budget-feasible up to any $O(T/\mathrm{poly}(n))$ error; see \Cref{thm: fixed-polynomial-regret}.

\begin{restatable}{budgettheorem}{budgetimpossibility}
\label{thm: fixed-polynomial-regret}
Under \textit{i.i.d.}\ item arrival and utility-maximizing bidders with budget constraints, for any fixed $\varepsilon\in(0,1)$, $N_0\in\mathbb{N}$, polynomial-time multi-agent bidding dynamic $\mathcal A$,
and polynomial $q(n)\in\Omega(n)$, if \ppad-hard problems cannot be solved by polynomial-time randomized algorithms, then there exists a problem instance with $n\geq N_0$ and $T=\lceil q(n)\rceil$ such that, for at least one bidder $i\in[n]$,
\begin{equation*}
\E[\regret_{i,\mathcal A}(\bm v)]>\frac{T}{Cn^{\varepsilon}}
\quad\text{or}\quad
\E\left[\left(\sumt c_{i,t}-T\rho_i\right)_+\right]
>\frac{T}{Cn^{\varepsilon}},
\end{equation*}
where $C\geq 1$ is an absolute constant. The same statement holds for the fluid setting.
\end{restatable}

We remark on the dependence of $T$ on $n$ in \Cref{thm: fixed-polynomial-regret}, which applies for any sufficiently large horizon length $T$ that is polynomial in $n$. Therefore, it also rules out the possibility of $O(T^{1-\varepsilon})$ regret, i.e., any regret guarantee that diminishes at a polynomial rate. Our hardness result does not apply to the case where the number of auctions is exponentially large in $n$, since bidders could then use exponentially many rounds to agree on a pacing equilibrium~\citep{conitzer2022multiplicative} by brute force and commit to it afterwards. In practice, it is unlikely that markets scale exponentially with the number of participants. 

In the rest of this subsection, we will show the major steps of the hardness analysis in the proof of \Cref{thm: fixed-polynomial-regret} and highlight our techniques. The missing proofs in this section can be found in \Cref{app: SPPE-proofs}.

\paragraph{Strongly Convex Fluid Dual.}
\label{section: strongly-convex-dual-SPPE}
For utility-maximizing bidders with budgets, the fluid dual losses are given by
\begin{equation*}
    D_{i,t}^F(\mu;\vec\alpha_{-i,t})
=\rho_i\mu+
\int\left(v_i-(1+\mu)\max_{j\ne i}\alpha_{j,t}v_j\right)_+
\,\dd F(\vec v).
\end{equation*}
By \Cref{prop: multipliers}, the bounds on pacing multipliers $\alpha_{i,t}\in [\alphalower_i, 1]$ induce the dual variable bound $\mu_{i,t}\in[0, \muupper_i]$, where $\muupper_i=1/\alphalower_i-1=\vmax/\rho_i-1$.

To apply \Cref{thm: strongly-convex-regret-decomposition}, we first construct a family of value distributions under which the fluid dual losses are strongly convex for all bidders.
\begin{definition}
\label{def: F0-distribution}
We define $F_0$ as the distribution of the value vector $\vec v = (v_1, \cdots, v_n)$ generated as follows,
\begin{enumerate}
    \item Choose $\{j_1, j_2\}$ uniformly from all $\binom{n}{2}$ size-$2$ subsets of $[n]$.
    \item Generate $v_{j_1}, v_{j_2}$ independently from $\mathrm{Uniform}[0, \vmax]$.
    \item For $k\notin \{j_1, j_2\}$, set $v_{k} = 0$.
\end{enumerate}
\end{definition}

Intuitively, $F_0$ spreads values continuously, so shading a bid further always strictly reduces the bidder's expected payment, which gives the dual loss curvature everywhere.
\begin{restatable}{lemma}{strongconvexityFzero}
    \label{lem: strong-convexity-F0} 
    For any bidder $i\in [n]$ and any competing strategies $\vec \alpha_{-i,t}$, the function $D_{i, t}^{F_0}(\cdot; \vec \alpha_{-i, t})$ is $\sigma_0$-strongly convex on $[0, \bar\mu_i]$, with
    \begin{equation*}
        \sigma_0 = \frac{2\vmax}{3n (1+ \max_{j\in[n]} \bar \mu_j)^3}.
    \end{equation*}
\end{restatable}

We will use $F_0$ to regularize a general atomic distribution to get a distribution $F$ with strongly convex fluid dual losses. Specifically, in our subsequent proof, we will consider $F = (1-\eta )F_1 +\eta F_0$ where $F_1$ is an atomic distribution and $\eta \in [0, 1]$. By a simple decomposition we know that $D_{i,t}^F$ will be $\eta\sigma_0$-strongly convex.

\paragraph{Low Regret Implies Approximate SPPE.}
\label{section: reduction-SPPE}

If all bidders achieve low fluid regret and are approximately budget-feasible, we show that one could then extract an \textit{offline} approximate SPPE from the \textit{online} auction history. The following \Cref{def: approximate-fluid-SPPE} extends the definition of SPPE to fluid auctions.

\begin{definition}
\label{def: approximate-fluid-SPPE}
For $\gamma\in[0,1)$, a $\gamma$-approximate fluid SPPE is a tuple $(\vec \alpha, (x_i(\vec v))_{i\in [n]})$, satisfying that
\begin{enumerate}
\item[(a)] Allocation to near highest bidders: $x_i(\vec v)>0 \implies \alpha_i v_i\geq(1-\gamma)h(\vec v;\vec\alpha)$ for all $i\in [n]$ and $\vec v \in [0, \vmax]^n$.
\item[(b)] Market clearance: $\sumi x_i(\vec v) \leq 1$ and $h(\vec v;\vec\alpha)>0 \implies \sumi x_i(\vec v)=1$ for all $\vec v \in [0, \vmax]^n$.
\item[(c)] Approximate budget feasibility: $\int x_i(\vec v)p(\vec v;\vec\alpha)\,\dd F(\vec v)
\leq(1+\gamma)\rho_i$ for all $i\in [n]$.
\item[(d)] Approximate no unnecessary pacing: $\int x_i(\vec v)p(\vec v;\vec\alpha)\,\dd F(\vec v)
<(1-\gamma)\rho_i \implies \alpha_i\geq1-\gamma$ for all $i\in [n]$.
\end{enumerate}
\end{definition}

By our regret decomposition in \Cref{thm: strongly-convex-regret-decomposition}, if all budget violations and regrets are small in $T$ repeated fluid auctions, 
then the following \eqref{eq: extraction-budget-hypothesis}--\eqref{eq: extraction-slack-hypothesis} must be controlled simultaneously,
\begin{align}
\text{budget overspent:}\quad   & \frac1T\sumt (c_{i,t}^F - \rho_i)\leq \delta_1 \cdot \rho_i, \ \forall i\in[n] \label{eq: extraction-budget-hypothesis} \\
\text{dual deviation:} \quad & \frac{1}{T} \sumi\sumt(\mu_{i,t}-\mu_i)^2\leq\delta_2,\label{eq: extraction-variance-hypothesis}
\\
\text{budget slack:} \quad & \mu_i\left(\rho_i-\frac1T\sumt c_{i,t}^F\right) 
\leq\delta_3,\ \forall i\in[n]  \label{eq: extraction-slack-hypothesis}
\end{align}

The following \Cref{lem: low-regret-implies-approximate-SPPE} shows that the above conditions---low budget violation, low variance, and small slack---will be sufficient for extracting an approximate fluid SPPE.
\begin{restatable}{lemma}{sppeextraction}
\label{lem: low-regret-implies-approximate-SPPE}
\ifsppefullstatement
For any $\gamma\in(0,1)$ and $\delta_1, \delta_2, \delta_3 >0$, 
conditions \eqref{eq: extraction-budget-hypothesis}--\eqref{eq: extraction-slack-hypothesis}
imply that a $\gamma$-approximate fluid SPPE can be extracted
efficiently if $\delta_1, \delta_2, \delta_3$ satisfy the following, 
\begin{equation*}
\delta_1\leq\frac{\gamma^3}{36},\qquad
\delta_2\leq
\frac{\gamma^3}{9\left(1+\max_i\frac{\int v_i\,\dd F(\vec v)}{\rho_i}\right)},
\qquad
\delta_3\leq\frac{\gamma^3\min_i\rho_i}{36}.
\end{equation*}
\else
For any $\gamma\in(0,1)$, there exist 
$\delta_1,\delta_2,\delta_3>0$ which depend only on $\gamma$, $F$, and
$(\rho_i)$, such that
\eqref{eq: extraction-budget-hypothesis}--\eqref{eq: extraction-slack-hypothesis}
imply that a $\gamma$-approximate fluid SPPE can be extracted efficiently.
\fi
\end{restatable}

\label{section: fluid-SPPE-is-PPAD-hard}

We then show that, similar to the standard discrete SPPE, computing an approximate fluid SPPE is also \ppad-hard, even after applying our regularizer $F_0$. Intuitively, a hard discrete SPPE instance can be embedded as the atomic part of the distribution, and mixing in a small amount of $F_0$ barely changes any bidder's spending.

\begin{restatable}{lemma}{fluidhardness}
\label{lemma: fluid-hardness}
There exist constants $\eta, \gamma \in (0, 1)$ such that computing a $\gamma$-approximate 
fluid SPPE is \ppad-hard for the family of distributions $\{(1-\eta)F_1 + \eta F_0 \}$.
\end{restatable}

\paragraph{From fluid to stochastic auctions.}
\label{section: fluid-stochastic-SPPE}

After proving the computational impossibility of joint no-regret in the \textit{fluid} setting, we further transfer this to the stochastic \textit{i.i.d.}\ item setting by an additional reduction argument. The following \Cref{lem: uniform-fluid-stochastic-transfer} shows that, by using the pacing multiplier sequence from the stochastic auctions, one can obtain a fluid bidding process with almost the same regret and expenditure up to sublinear terms, with high probability. This suggests that achieving joint expected no-regret in the stochastic \textit{i.i.d.}\ setting is no easier than in the fluid setting.

\begin{restatable}{lemma}{fluidstochastictransfer}
\label{lem: uniform-fluid-stochastic-transfer}
For any stochastic multi-agent bidding dynamic $\mathcal A$, we can efficiently obtain $\mathcal A^F$
for the fluid setting by simulating $\mathcal A$, and achieve
\begin{align*}
\E\left[\left|\regret_{i,\mathcal A^F}^F
                 -\regret_{i,\mathcal A}\right|\right]
&\leq \vmax\left[
\sqrt{2T\log\left(2(\lceil\muupper_i T\rceil+1)\right)}
+\sqrt T+4\right],\\
\E\left[\left(\sumt c_{i,t}^F-T\rho_i\right)_+\right]
&\leq
\E\left[\left(\sumt c_{i,t}-T\rho_i\right)_+\right]
+\vmax\sqrt T. 
\end{align*}
\end{restatable}

\subsection{Evidence of High Variance in Bidding Systems}
\label{subsection: budget-2}

Next, we show concrete numerical examples where the state of the bidders' algorithms displays cycling behavior, when all bidders are utility-maximizers with budgets and use common bidding algorithms. We will simulate two scenarios: 1) a continuous-time dynamical system with agents using PID controllers, and 2) discrete $T$ auctions with agents using online mirror descent (OMD). As the cycle recurs, the variance term in \Cref{thm: strongly-convex-regret-decomposition} grows linearly over time.

\paragraph{Dynamical system with PID controllers.} 
We consider the dynamical system model studied by \citet{paes2024complex}, where second-price auctions occur in continuous time over the same fluid item. We denote the profile of pacing multipliers at time $t$ as $\vec \alpha(t) = (\alpha_i(t))_{i\in [n]}$. We abuse the notation and denote bidder $i$'s value increment and cost in the time-$t$ auction as $u_i^F(t)$ and $c_i^F(t)$.

Following \citet{paes2024complex}, we consider the PID controller, a practically successful method of adjusting the pacing multiplier over time~\citep{balseiro2023field,tashman2020dynamic,yang2019bid,zhang2016feedback,smirnov2016online}. In the continuous-time setting, the PID controller can be interpreted as the differential equations:
\begin{equation}
    \label{eq:pid-budget}
    \frac{\dd \alpha_i}{\dd t} = \rho_i - c_i^F(t), \quad \forall i \in [n].
\end{equation}
The above \eqref{eq:pid-budget} captures the natural practice of increasing the multiplier (and thus bidding higher) when the expenditure rate is lower than the per-auction budget, and decreasing the multiplier otherwise. 

To simulate cycling behavior, we use the Rock-Paper-Scissors instance from \citet{paes2024complex}, where the distribution of three bidders' valuations for the fluid item is as follows, 
\begin{equation}
    \label{eq:rps-distribution}
    \begin{aligned}
    F &= \text{Uniform}\left\{\begin{pmatrix}
        \text{Beta}(2c, c) \\ 0 \\ \text{Beta}(c, 2c)
    \end{pmatrix}, \begin{pmatrix}
        \text{Beta}(c, 2c) \\ \text{Beta}(2c, c) \\ 0
    \end{pmatrix}, \begin{pmatrix}
        0 \\ \text{Beta}(c, 2c) \\ \text{Beta}(2c, c)
    \end{pmatrix}\right\} \\
    &\approx \text{Uniform} \left\{\begin{pmatrix}
        2/3 \\ 0 \\ 1/3
    \end{pmatrix}, \begin{pmatrix}
        1/3 \\ 2/3 \\ 0
    \end{pmatrix}, \begin{pmatrix}
        0 \\ 1/3 \\ 2/3
    \end{pmatrix}\right\}.
    \end{aligned}
\end{equation}

We explain why this instance leads to cycling. Notice that each item is won by its high-value bidder at a price set by its low-value bidder, so if bidder~1 bids higher, bidder~2 pays more and bids lower, then bidder~3 pays less and bids higher, and finally bidder~1 pays more and bids lower, completing a loop.

Under this instance, the pacing multipliers of bidders form a stable cycle over time; see \Cref{fig: budget-pid-cycles}. We also visualize the cycle by projecting the orbit of $\vec \alpha(t)$ to its first two dimensions, with three different initializations. Cycling directly leads to high variance in the multiplier sequence, and we next show that this indeed leads to high regret for each bidder; see the growth of cumulative regret in \Cref{fig: budget-pid-regret}. This is consistent with our observations in \Cref{thm: strongly-convex-regret-decomposition}. 

\begin{figure}[!ht]
\centering
\includegraphics[width=\linewidth]{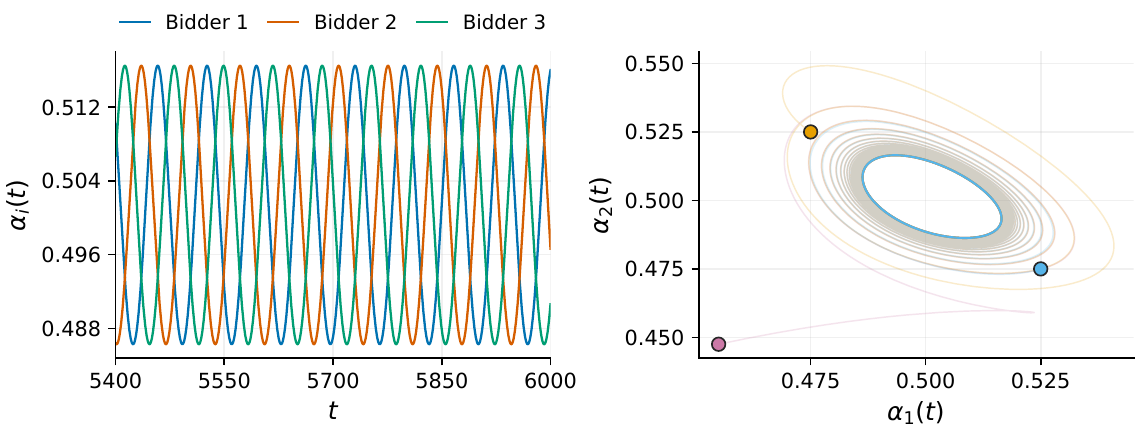}
\caption{Cycling of pacing multipliers under the PID dynamics. Left: the pacing multipliers $\alpha_i(t)$ in a time segment. Right: orbits of $(\alpha_1(t), \alpha_2(t))$ from three different initializations.}
\label{fig: budget-pid-cycles}
\end{figure}

\begin{figure}[!ht]
\centering
\includegraphics[width=0.6\linewidth]{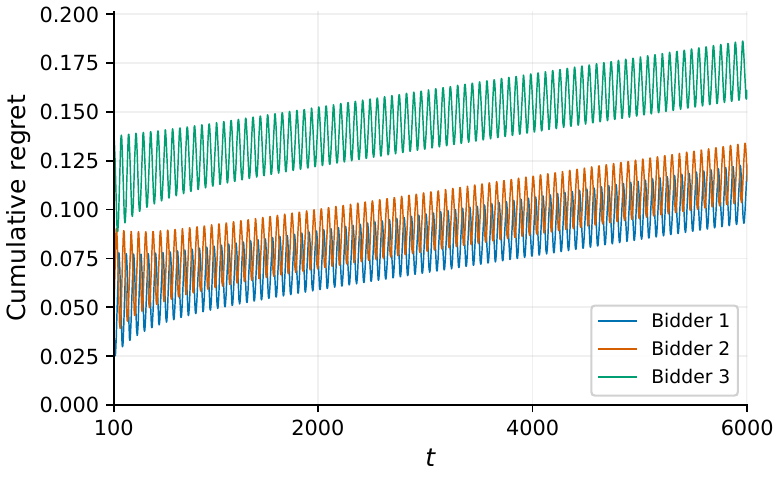}
\caption{Cumulative regret of each bidder under the PID dynamics.}
\label{fig: budget-pid-regret}
\end{figure}

\paragraph{Fluid auctions with online mirror descent.}
We next consider $T$ rounds of fluid auctions with the same value distribution \eqref{eq:rps-distribution}. Following the standard primal-dual approach to online resource allocation~\citep{balseiro2023best}, each bidder updates their dual variable using online mirror descent (OMD) on the dual losses $D_{i,t}^F$, whose gradient at $\mu_{i,t}$ is $\rho_i - c_{i,t}^F$. We consider the Euclidean and the entropic updates with learning rate $\eta>0$,
\begin{equation*}
    \mu_{i,t+1} =
    \begin{cases}
        \mu_{i,t} - \eta\left(\rho_i - c_{i,t}^F\right) & \text{(Euclidean)},\\
        \mu_{i,t} \exp\left(-\eta\left(\rho_i - c_{i,t}^F\right)\right) & \text{(entropic)},
    \end{cases}
\end{equation*}

Similar to the PID controller \eqref{eq:pid-budget}, a bidder decreases the dual variable to bid higher when the expenditure is lower than the per-auction budget, and increases the dual variable otherwise. 

We simulate different choices of learning rate $\eta$. For a constant learning rate, \Cref{fig: budget-omd-cycles} shows that the pacing multipliers keep oscillating under both updates. We also consider $\eta \sim 1/\sqrt{T}$ as in the best-of-many-worlds guarantees~\citep{balseiro2023best}, and show that the regret of each bidder still grows linearly in $T$ even when the stepsize is chosen with the knowledge of the horizon length. 

\begin{figure}[!ht]
\centering
\includegraphics[width=\linewidth]{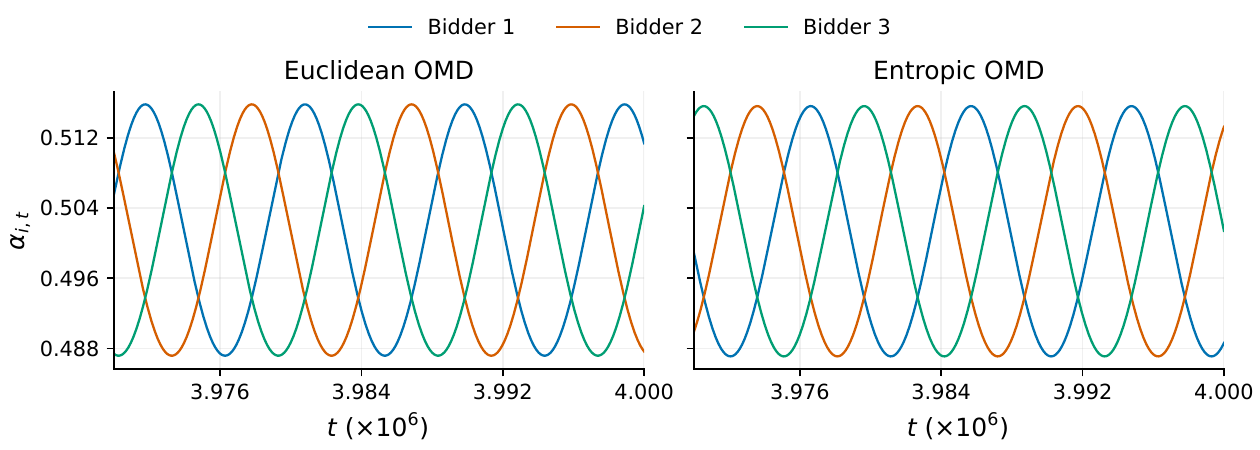}
\caption{Pacing multipliers under Euclidean OMD (left) and entropic OMD (right) with a fixed learning rate.}
\label{fig: budget-omd-cycles}
\end{figure}

\begin{figure}[!ht]
\centering
\includegraphics[width=\linewidth]{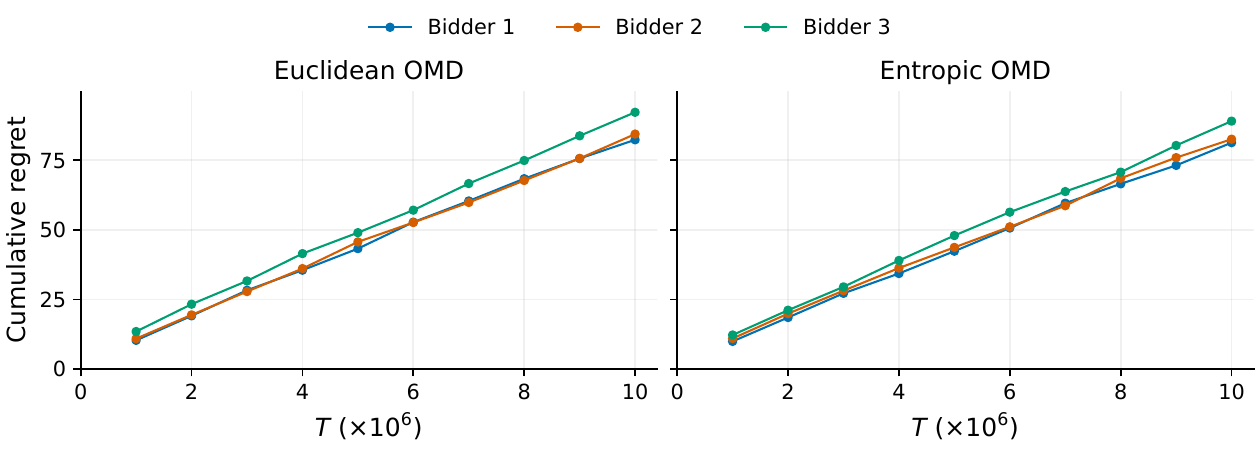}
\caption{Fluid regret of each bidder under Euclidean OMD (left) and entropic OMD (right), with $\eta$ proportional to $1/\sqrt{T}$.}
\label{fig: budget-omd-regret}
\end{figure}

\section{Impossibility for Value Maximizers with ROI Targets}
\label{section: main-section-on-ABE}

We next show the computational difficulty results and numerical simulations for value maximizers with ROI targets. 

\subsection{Computational Complexity of Joint No-Regret}
\begin{restatable}{budgettheorem}{roiimpossibility}
\label{thm: fixed-polynomial-ROI-regret}
Under \textit{i.i.d.} item arrival and value-maximizing bidders with ROI constraints, for any fixed $\varepsilon\in(0,1)$, $N_0\in\mathbb{N}$, polynomial-time multi-agent bidding dynamic $\mathcal A$,
and polynomial $q(n)=\Omega(n)$, if \ppad-hard problems cannot be solved by polynomial-time randomized algorithms, then there exists a problem instance with $n\geq N_0$ and $T=\lceil q(n)\rceil$ such that, for at least one bidder $i\in[n]$,
\begin{equation*}
\E[\regret_{i,\mathcal A}(\bm v)]>\frac{T}{Cn^{\varepsilon}}
\quad\text{or}\quad
\E\left[\left(r_i\sumt c_{i,t}-\sumt u_{i,t}\right)_+\right]
>\frac{T}{Cn^{\varepsilon}},
\end{equation*}
where $C\geq1$ is an absolute constant. The same statement holds for the fluid setting.
\end{restatable}

We will prove \Cref{thm: fixed-polynomial-ROI-regret} in an analogous way to \Cref{thm: fixed-polynomial-regret}, since both proofs follow the same primal-dual framework of analysis. We will use the simplest ROI target $r_i=1$ for all bidders. The missing lemmas and proofs in this section can be found in \Cref{app: ROI-proofs}.

\paragraph{Proof sketch of \Cref{thm: fixed-polynomial-ROI-regret}.} For value-maximizing bidders, the most important technical distinction is that the dual variables can be \textit{unbounded}. Concretely, the fluid dual losses are given by
\begin{equation*}
D_{i,t}^F(\mu;\vec\alpha_{-i,t})
=\int\left((1+\mu)v_i-\mu\max_{j\ne i}\alpha_{j,t}v_j\right)_+
\,\dd F(\vec v).
\end{equation*}
By \Cref{prop: multipliers}, the pacing multipliers $\alpha_{i,t}\in(\alphalower_i,+\infty)$ correspond to dual iterates $\mu_{i,t}\in(0,+\infty)$. To handle unboundedness, we will consider the projections of the dual iterates to a bounded interval $(0, H]$ in our analysis. This recovers several desirable properties such as strong convexity of the dual, and the case where the multipliers are outside the interval can be handled separately. 

\begin{restatable}{lemma}{roicurvature}
\label{lem: ROI-strong-convexity}
For any bidder $i\in[n]$, competing strategies $\vec\alpha_{-i,t}$, and $H>0$, $D_{i,t}^{F_0}(\cdot;\vec\alpha_{-i,t})$ is $\sigma_0$-strongly convex on $[0,H]$, with
\begin{equation*}
\sigma_0=\frac{2\vmax}{3n(n-1)(H+1)^3\min_{j\ne i}\alpha_{j,t}}.
\end{equation*}
\end{restatable}

In addition to $F_0$ in \Cref{def: F0-distribution}, in the ROI-constrained setting we consider the following additional regularizer, which randomly gives bidders values at price zero. This helps to bound dual minimizers in our instance and to control relative ROI errors.
For an atomic distribution $F_1$, we will consider $F=(1-\eta-\eta')F_1+\eta F_0+\eta'F_0'$,
where $\eta,\eta'>0$ and $\eta+\eta'<1$. 
\begin{definition}
\label{def: ROI-regularizer}
We define $F_0'$ as the distribution of the value vector $\vec v=(v_1,\ldots,v_n)$ generated as follows,
\begin{enumerate}
\item Choose $j$ uniformly from $[n]$ and set $v_j=\vmax$.
\item For $k\ne j$, set $v_k=0$.
\end{enumerate}
\end{definition}

The main step in the proof, similar to the previous budget-constrained setting, is to show that an online process with low regret and ROI feasibility across all agents will be sufficient for extracting an offline equilibrium solution that is \ppad-hard to compute. To this end, we use the following fluid analogue of \textit{autobidding equilibria}~\citep{li2024vulnerabilities}.

\begin{definition}
\label{def: approximate-fluid-ABE}
For $\gamma\in[0,1)$ and $M>1$, a $\gamma$-approximate fluid autobidding equilibrium with cap $M$ is a tuple $(\vec\alpha,(x_i(\vec v))_{i\in[n]})$, with $\vec\alpha\in[1,M]^n$, satisfying that
\begin{enumerate}
\item[(a)] Allocation to near highest bidders: $x_i(\vec v)>0\implies\alpha_i v_i\geq(1-\gamma)h(\vec v;\vec\alpha)$ for all $i\in[n]$ and $\vec v\in[0,\vmax]^n$.
\item[(b)] Market clearance: $\sumi x_i(\vec v)\leq1$ and $h(\vec v;\vec\alpha)>0\implies\sumi x_i(\vec v)=1$ for all $\vec v\in[0,\vmax]^n$.
\item[(c)] Approximate ROI feasibility: $\int x_i(\vec v)p(\vec v;\vec\alpha)\,\dd F(\vec v)\leq(1+\gamma)\int x_i(\vec v)v_i\,\dd F(\vec v)$ for all $i\in[n]$.
\item[(d)] Approximate complementarity: $\int x_i(\vec v)p(\vec v;\vec\alpha)\,\dd F(\vec v)<(1-\gamma)\int x_i(\vec v)v_i\,\dd F(\vec v)\implies\alpha_i=M$ for all $i\in[n]$.
\end{enumerate}
\end{definition}

By our decomposition in \Cref{lem: regret-to-dual-regret}, low regret and small ROI violations will imply that the following terms \eqref{eq: ROI-extraction-payment-hypothesis}--\eqref{eq: ROI-extraction-slack-hypothesis} are controlled simultaneously. Specifically, we allow $\alpha_{i,t}\in(1,M]$ and any fixed multipliers $\mu_i\geq1/(M-1)$,
\begin{align}
\text{ROI violation:}\quad
&\frac{\sumt(c_{i,t}^F-u_{i,t}^F)}{\sumt u_{i,t}^F}\leq\delta_1,
\quad\forall i\in[n],
\label{eq: ROI-extraction-payment-hypothesis}\\
\text{dual deviation:}\quad
&\frac1T\sumi\sumt(\mu_{i,t}-\mu_i)^2\leq\delta_2,
\label{eq: ROI-extraction-variance-hypothesis}\\
\text{ROI slack:}\quad
&\frac{\mu_i\sumt(u_{i,t}^F-c_{i,t}^F)}{(1+\mu_i)\sumt u_{i,t}^F}\leq\delta_3,
\quad\forall i:\ \mu_i>\frac1{M-1}.
\label{eq: ROI-extraction-slack-hypothesis}
\end{align}

\begin{figure}[t]
\centering
\includegraphics[width=\linewidth]{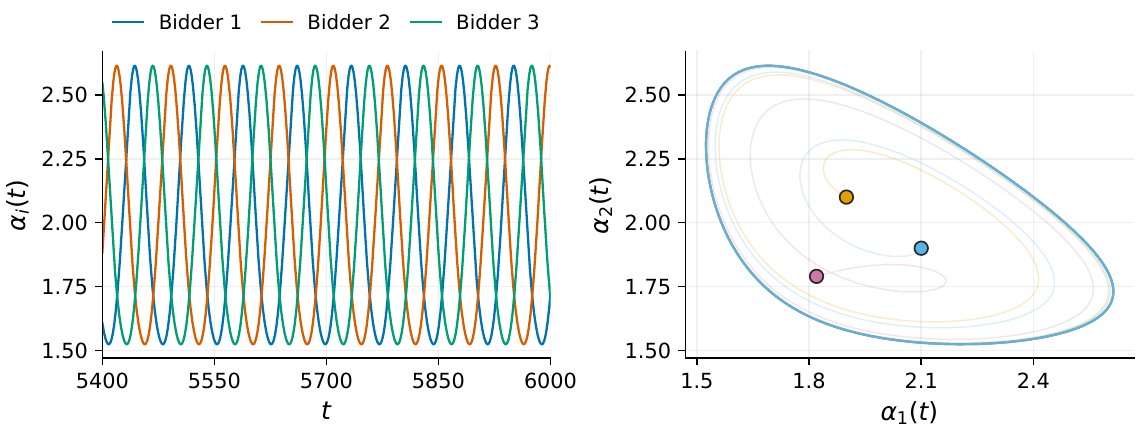}
\caption{Cycling of pacing multipliers of value maximizers under the PID dynamics. Left: the pacing multipliers $\alpha_i(t)$ in a time segment. Right: orbits of $(\alpha_1(t), \alpha_2(t))$ from three different initializations.}
\label{fig: roi-pid-cycles}
\end{figure}

Similar to \Cref{lem: low-regret-implies-approximate-SPPE}, we show that this will further imply an approximate fluid autobidding equilibrium.

\newif\ifroifullstatement
\begin{restatable}{lemma}{roiextraction}
\label{lem: low-regret-implies-approximate-ABE}
\ifroifullstatement
For any $\gamma\in(0,1)$ and $\delta_1,\delta_2,\delta_3>0$,
conditions \eqref{eq: ROI-extraction-payment-hypothesis}--\eqref{eq: ROI-extraction-slack-hypothesis}
imply that a $\gamma$-approximate fluid autobidding equilibrium with cap $M$ can be extracted
efficiently from the online history and $(\mu_i)_{i\in[n]}$ if $\delta_1, \delta_2, \delta_3$ satisfy the following,
\begin{equation*}
\delta_1\leq\frac{\gamma}{8},\qquad
\delta_2\leq
\frac{\gamma^3\sumt u_{i,t}^F}
{32M(M-1)^2T\int v_i\,\dd F(\vec v)}\ \ \text{ $\forall i\in[n]$},
\qquad
\delta_3\leq\frac{\gamma}{8M}.
\end{equation*}
\else
For any $\gamma\in(0,1)$, there exist $\delta_1,\delta_2,\delta_3>0$ such that \eqref{eq: ROI-extraction-payment-hypothesis}--\eqref{eq: ROI-extraction-slack-hypothesis} imply that a $\gamma$-approximate fluid autobidding equilibrium with cap $M$ can be extracted efficiently from the online history and $(\mu_i)_{i\in[n]}$.
\fi
\end{restatable}

To complete the proof, we will also show the \ppad-hardness of computing an approximate fluid autobidding equilibrium under regularization (\Cref{lemma: roi-fluid-hardness}), and perform a reduction from fluid to stochastic \textit{i.i.d.}\ auctions (\Cref{lem: ROI-fluid-stochastic-transfer}).

\subsection{Evidence of High Variance in Bidding Systems}
\label{subsection: roi-2}

Similar to \Cref{subsection: budget-2}, we show that common bidding algorithms display cycling behavior when all bidders are value maximizers with ROI targets $r_i=1$. We use the same value distribution \eqref{eq:rps-distribution}, and simulate the same two settings.

\paragraph{Dynamical system with PID controllers.}
For value maximizers, the PID controller adjusts the pacing multiplier according to the slack in the ROI constraint,
\begin{equation*}
    \frac{\dd \alpha_i}{\dd t} = u_i^F(t) - c_i^F(t), \quad \forall i \in [n],
\end{equation*}
which increases the multiplier when the value gained exceeds the payment, and decreases it otherwise. This is the setting originally studied by \citet{paes2024complex}. We reproduce the cycling behavior in \Cref{fig: roi-pid-cycles}, and additionally show cumulative regret growth of bidders in \Cref{fig: roi-pid-regret}.

\begin{figure}[!ht]
\centering
\includegraphics[width=0.6\linewidth]{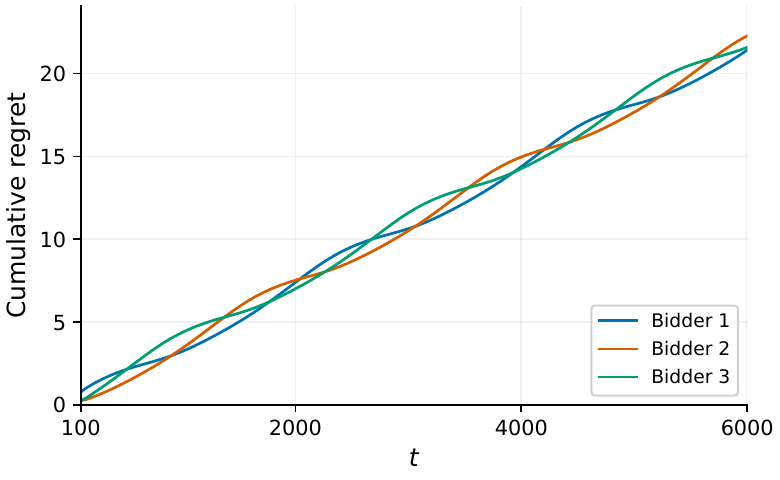}
\caption{Cumulative regret of each value maximizer under the PID dynamics.}
\label{fig: roi-pid-regret}
\end{figure}

\paragraph{Fluid auctions with online mirror descent.}
We next consider $T$ rounds of fluid auctions with the same value distribution, where each bidder runs OMD with Euclidean or entropic updates. With a fixed learning rate, \Cref{fig: roi-omd-cycles} shows that the multipliers keep oscillating under both updates. With $\eta \sim 1/\sqrt{T}$, \Cref{fig: roi-omd-regret} shows that the fluid regret of each bidder still grows linearly in $T$.

\begin{figure}[!ht]
\centering
\includegraphics[width=\linewidth]{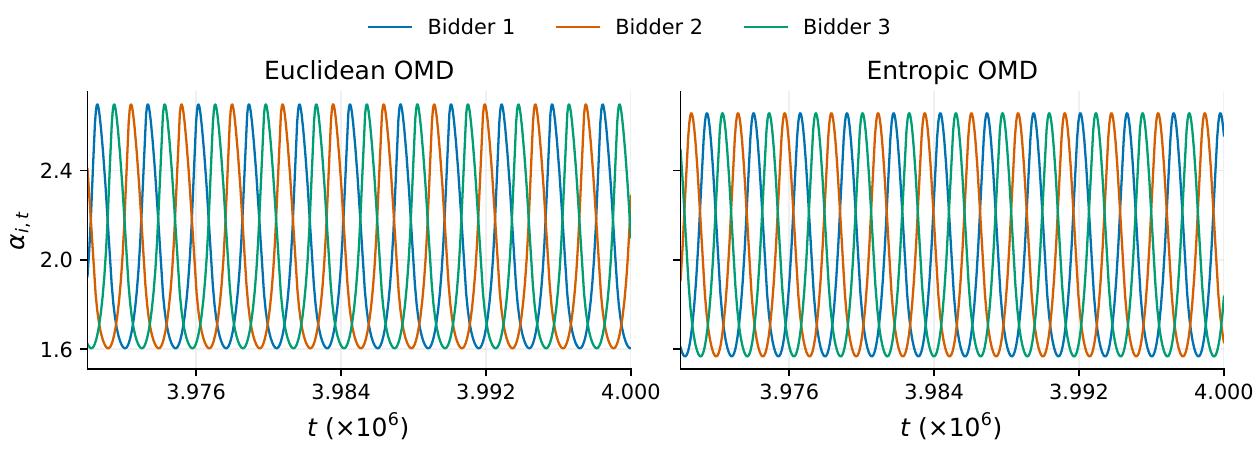}
\caption{Pacing multipliers of value maximizers under Euclidean OMD (left) and entropic OMD (right) with a fixed learning rate.}
\label{fig: roi-omd-cycles}
\end{figure}

\begin{figure}[!ht]
\centering
\includegraphics[width=\linewidth]{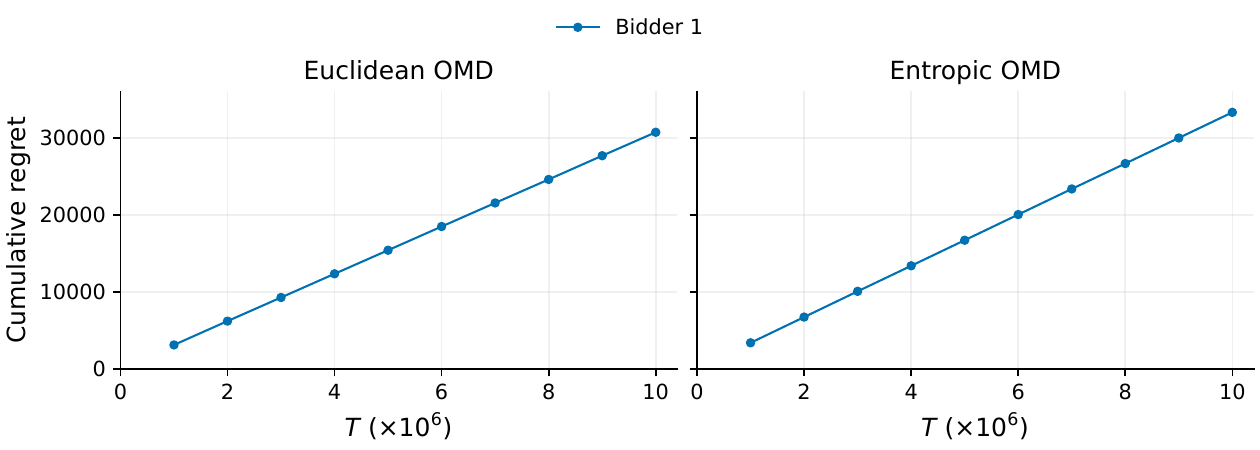}
\caption{Fluid regret of bidder~1 under Euclidean OMD (left) and entropic OMD (right), with $\eta$ proportional to $1/\sqrt{T}$; the other two value maximizers have nearly identical regret.}
\label{fig: roi-omd-regret}
\end{figure}

\FloatBarrier
\subsection*{Statement on the use of generative AI}
All results, proofs, and exposition in this paper are the authors' own work. We used Claude Opus 5.5 to improve various constants in the proofs, simplify intermediate computations, and make minor edits in the appendix. This use was not essential to obtaining the results. The authors take full responsibility for the content and correctness of the results.

\bibliographystyle{plainnat}
\bibliography{ref}
\appendix

\crefalias{section}{appendix}
\crefalias{subsection}{subappendix}
\crefalias{subsubsection}{subsubappendix}
\newpage
\section{Missing Results and Proofs of \Cref{section: preliminaries} and \Cref{section: regret-lower-bound}}
\label{app: general-proofs}

\subsection{Derivation of the Dual for Utility Maximizers with Budgets}
\label{app: budget-duality}

To get the hindsight benchmark \eqref{eq: primal}, consider $\mathcal X=\{0,1\}$, $f_t(x)=(v_{i,t}-d_{i,t})x$, $a_t(x)=d_{i,t}x$, and $\rho=\rho_i$. For \eqref{eq: fluid-hindsight-budget}, consider $\mathcal X=\{x:\valuespace\to\{0,1\}\}$, $\rho=\rho_i$, and
\begin{equation*}
f_t(x)=\int x(\vec v)\Bigl(v_i-\max_{j\ne i}\alpha_{j,t}v_j\Bigr)\dd F(\vec v),\qquad
a_t(x)=\int x(\vec v)\max_{j\ne i}\alpha_{j,t}v_j\,\dd F(\vec v).
\end{equation*}
By the definition of $D_t$ in \eqref{eq: resource-allocation-dual}, the dual losses are
\begin{align*}
D_{i,t}(\mu;\vec v_t,\vec\alpha_{-i,t})&=\rho_i\mu+\bigl(v_{i,t}-(1+\mu)d_{i,t}\bigr)_+,\\
D_{i,t}^F(\mu;\vec\alpha_{-i,t})&=\rho_i\mu+\int\Bigl(v_i-(1+\mu)\max_{j\ne i}\alpha_{j,t}v_j\Bigr)_+\dd F(\vec v).
\end{align*}

\paragraph{Bounds on optimal dual solutions.}
The right derivative of $D_{i,t}$ at $\muupper_i$ satisfies
\begin{equation*}
\rho_i-d_{i,t}\indi{v_{i,t}>(1+\muupper_i)d_{i,t}}
\geq\rho_i-\frac{\vmax}{1+\muupper_i}=0,
\end{equation*}
and similarly for $D_{i,t}^F$. By convexity, both dual objectives attain their minima in $[0,\muupper_i]$.

\paragraph{Duality gaps.}
The reward--cost sets are finite for \eqref{eq: primal}, and compact for \eqref{eq: fluid-hindsight-budget} by Lyapunov's theorem \citep[Section~1]{dai2011extension}. Thus by \Cref{prop: one-item-duality-gap}, both duality gaps are in $[0,\vmax]$.

\subsection{Derivation of the Dual for Value Maximizers with ROI Targets}
\label{app: ROI-duality}

To get the hindsight benchmark \eqref{eq: primal-roi}, consider $\mathcal X=\{0,1\}$, $f_t(x)=v_{i,t}x$, $a_t(x)=(r_id_{i,t}-v_{i,t})x$, and $\rho=0$. For \eqref{eq: fluid-hindsight-roi}, consider $\mathcal X=\{x:\valuespace\to\{0,1\}\}$, $\rho=0$, and
\begin{equation*}
f_t(x)=\int x(\vec v)v_i\,\dd F(\vec v),\qquad
a_t(x)=\int x(\vec v)\Bigl(r_i\max_{j\ne i}\alpha_{j,t}v_j-v_i\Bigr)\dd F(\vec v).
\end{equation*}
By the definition of $D_t$ in \eqref{eq: resource-allocation-dual}, the dual losses are
\begin{align*}
D_{i,t}(\mu;\vec v_t,\vec\alpha_{-i,t})&=\bigl((1+\mu)v_{i,t}-\mu r_id_{i,t}\bigr)_+,\\
D_{i,t}^F(\mu;\vec\alpha_{-i,t})&=\int\Bigl((1+\mu)v_i-\mu r_i\max_{j\ne i}\alpha_{j,t}v_j\Bigr)_+\dd F(\vec v).
\end{align*}

\paragraph{Existence of optimal dual solutions.}
The dual objective of \eqref{eq: primal-roi} is piecewise linear and bounded below, so it attains its minimum. The dual objective of \eqref{eq: fluid-hindsight-roi} attains its minimum when it is coercive. In particular, this holds for $F=(1-\eta-\eta')F_1+\eta F_0+\eta'F_0'$ with $\eta'>0$, since then $D_{i,t}^F(\mu;\vec\alpha_{-i,t})\geq\eta'\vmax(1+\mu)/n$.

\paragraph{Duality gaps.}
The reward--cost sets are finite for \eqref{eq: primal-roi}, and compact for \eqref{eq: fluid-hindsight-roi} by Lyapunov's theorem \citep[Section~1]{dai2011extension}. Thus by \Cref{prop: one-item-duality-gap}, both duality gaps are in $[0,\vmax]$.

\subsection{Proof of \Cref{prop: one-item-duality-gap}}
\label{app: proof-duality-gap}

\dualitygapbound*

\begin{proof}
The lower bound $\Delta\geq0$ is weak duality. For the upper bound, let $S_t=\{(f_t(x),a_t(x)):x\in\mathcal X\}$, and consider the convex relaxation of \eqref{eq: hindsight-resource-allocation} in which round $t$ chooses a point $(y_t,z_t)\in\operatorname{conv}S_t$, with objective $\sumt y_t$ and constraint $\sumt z_t\leq T\rho$. Since $\operatorname{conv}S_t$ is compact and a linear function attains its maximum over $\operatorname{conv}S_t$ at a point of $S_t$, Sion's minimax theorem \citep{sion1958general} shows that the relaxation has optimal value $D^*$.

We write an optimal solution of the relaxation as convex combinations of finitely many points of $S_t$, and consider the linear program over these convex weights. It has value $D^*$, one equality constraint for each round, and one budget constraint. An optimal vertex therefore has at most $T+1$ positive weights, so all rounds but at most one round $t_0$ use a single point of $S_t$. In round $t_0$, we choose the point with the smaller cost. This gives a feasible solution of \eqref{eq: hindsight-resource-allocation}, and loses reward at most $\vmax$ since $f_{t_0}\in[0,\vmax]$. Hence $\Delta\leq\vmax$.
\end{proof}

\subsection{Proof of \Cref{prop: multipliers}}
\label{app: proof-multipliers}

\multiplierduals*

\begin{proof}
For utility maximizers with budgets, $\mu_{i,t}=1/\alpha_{i,t}-1$ gives $f_t(x)-\mu_{i,t}a_t(x)=x(\alpha_{i,t}v_{i,t}-d_{i,t})/\alpha_{i,t}$. For value maximizers with ROI targets, $\mu_{i,t}=1/(\alpha_{i,t}r_i-1)$ gives $f_t(x)-\mu_{i,t}a_t(x)=x\mu_{i,t}r_i(\alpha_{i,t}v_{i,t}-d_{i,t})$. In both cases, the maximizers over $x\in[0,1]$ are exactly the second-price allocations under the bid $\alpha_{i,t}v_{i,t}$, namely $x=1$ if $\alpha_{i,t}v_{i,t}>d_{i,t}$, $x=0$ if $\alpha_{i,t}v_{i,t}<d_{i,t}$, and any $x\in[0,1]$ under a tie. Since $\alpha\mapsto1/\alpha-1$ and $\alpha\mapsto1/(\alpha r_i-1)$ are bijections onto the dual domains, choosing $\alpha_{i,t}$ is equivalent to choosing $\mu_{i,t}$.
\end{proof}

\subsection{Proof of \Cref{prop: linearization-gap-nonnegativity}}
\label{app: proof-linearization-gap}

\linearizationgap*

\begin{proof}
Fix $t\in[T]$. \Cref{alg: primal-dual} chooses $x_t$ as a maximizer of
$f_t(x)-\mu_ta_t(x)$ over $\mathcal X$, so the definitions of $D_t$ and
$f_t^*$ give
\begin{equation*}
D_t(\mu_t)=\rho\mu_t+f_t(x_t)-\mu_ta_t(x_t).
\end{equation*}
For any comparator $\mu\geq0$, we evaluate the maximum in $f_t^*(\mu)$ at
the same decision $x_t$ and obtain
\begin{align*}
D_t(\mu)-D_t(\mu_t)
&\geq\rho\mu+f_t(x_t)-\mu a_t(x_t)
-\bigl(\rho\mu_t+f_t(x_t)-\mu_ta_t(x_t)\bigr)\\
&=(\rho-a_t(x_t))(\mu-\mu_t)
=\widetilde D_t(\mu)-\widetilde D_t(\mu_t).
\end{align*}
We rearrange and obtain
$D_t(\mu_t)-D_t(\mu)\leq\widetilde D_t(\mu_t)-\widetilde D_t(\mu)$.
We sum over $t\in[T]$ and obtain $R_D(\mu)\leq\widetilde R_D(\mu)$.
\end{proof}

\subsection{Proof of \Cref{lem: regret-to-dual-regret}}
\label{app: proof-regret-decomposition}

\regretdecomposition*

\begin{proof}
Fix $t\in[T]$. Since $x_t$ maximizes $f_t(x)-\mu_ta_t(x)$ over
$\mathcal X$, we have $f_t^*(\mu_t)=f_t(x_t)-\mu_ta_t(x_t)$. By the
definitions of $D_t$ and $\widetilde D_t$,
\begin{align*}
D_t(\mu_t)-\widetilde D_t(\mu_t)
&=\rho\mu_t+f_t^*(\mu_t)-(\rho-a_t(x_t))\mu_t\\
&=f_t(x_t)-\mu_ta_t(x_t)+\mu_ta_t(x_t)
=f_t(x_t).
\end{align*}
Since $\mu^*$ is optimal for \eqref{eq: resource-allocation-dual}, we have
$D^*=\sumt D_t(\mu^*)$, and $\mathrm{OPT}=D^*-\Delta$ by the
definition of $\Delta$. Consequently,
\begin{align*}
\mathrm{OPT}-\sumt f_t(x_t)
&=\sumt D_t(\mu^*)-\sumt D_t(\mu_t)
+\sumt\widetilde D_t(\mu_t)-\Delta\\
&=\left(\sumt\widetilde D_t(\mu_t)-\sumt\widetilde D_t(\mu^*)\right)
-\left(\sumt D_t(\mu_t)-\sumt D_t(\mu^*)\right)
+\sumt\widetilde D_t(\mu^*)-\Delta\\
&\overset{\text{(a)}}{=}\widetilde R_D(\mu^*)-R_D(\mu^*)
+\mu^*\left(\rho T-\sumt a_t(x_t)\right)-\Delta,
\end{align*}
where step (a) uses the definitions of $\widetilde R_D$ and $R_D$, and
$\sumt\widetilde D_t(\mu^*)=\mu^*\sumt(\rho-a_t(x_t))
=\mu^*(\rho T-\sumt a_t(x_t))$.
\end{proof}

\subsection{Proof of \Cref{thm: strongly-convex-regret-decomposition}}
\label{app: proof-strongly-convex-decomposition}

\stronglyconvexdecomposition*

\begin{proof}
Let $I$ be the interval in the statement. Fix $t\in[T]$. By the strong
convexity of $D_t$ on $I$, the function $h_t(\mu)=D_t(\mu)-\sigma\mu^2/2$ is
convex on $I$. For $\lambda\in(0,1]$, the point $\mu_t+\lambda(\mu^*-\mu_t)$
lies in $I$, and
\begin{align*}
h_t(\mu^*)-h_t(\mu_t)
&\overset{\text{(a)}}{\geq}
\frac{h_t(\mu_t+\lambda(\mu^*-\mu_t))-h_t(\mu_t)}{\lambda}\\
&=\frac{D_t(\mu_t+\lambda(\mu^*-\mu_t))-D_t(\mu_t)}{\lambda}
-\frac{\sigma}{2}(\mu^*-\mu_t)\bigl(2\mu_t+\lambda(\mu^*-\mu_t)\bigr)\\
&\overset{\text{(b)}}{\geq}(\rho-a_t(x_t))(\mu^*-\mu_t)
-\frac{\sigma}{2}(\mu^*-\mu_t)\bigl(2\mu_t+\lambda(\mu^*-\mu_t)\bigr),
\end{align*}
where step (a) is by the convexity of $h_t$ on the segment between $\mu_t$
and $\mu^*$, and step (b) is by the inequality
$D_t(\mu)-D_t(\mu_t)\geq(\rho-a_t(x_t))(\mu-\mu_t)$, $\mu\geq0$, from the
proof of \Cref{prop: linearization-gap-nonnegativity}, at
$\mu=\mu_t+\lambda(\mu^*-\mu_t)$.
We let $\lambda\to0$ and add $\sigma({\mu^*}^2-\mu_t^2)/2$ to both sides:
\begin{align*}
D_t(\mu^*)-D_t(\mu_t)
&\geq(\rho-a_t(x_t))(\mu^*-\mu_t)-\sigma\mu_t(\mu^*-\mu_t)
+\frac{\sigma}{2}(\mu^*-\mu_t)(\mu^*+\mu_t)\\
&=(\rho-a_t(x_t))(\mu^*-\mu_t)+\frac{\sigma}{2}(\mu^*-\mu_t)^2.
\end{align*}
We sum over $t\in[T]$. Since
$\widetilde R_D(\mu^*)=\sumt(\rho-a_t(x_t))(\mu_t-\mu^*)$ and
$R_D(\mu^*)=\sumt(D_t(\mu_t)-D_t(\mu^*))$, we obtain
\begin{align*}
\widetilde R_D(\mu^*)-R_D(\mu^*)
&\geq\frac{\sigma}{2}\sumt(\mu_t-\mu^*)^2\\
&\overset{\text{(c)}}{=}\frac{\sigma}{2}\sumt\left(\mu_t-\frac1T\sum_{\tau=1}^T\mu_\tau\right)^2
+\frac{\sigma T}{2}\left(\frac1T\sum_{\tau=1}^T\mu_\tau-\mu^*\right)^2\\
&\geq\frac{\sigma}{2}\sumt\left(\mu_t-\frac1T\sum_{\tau=1}^T\mu_\tau\right)^2,
\end{align*}
where step (c) expands the square and uses
$\sumt(\mu_t-T^{-1}\sum_{\tau=1}^T\mu_\tau)=0$.
We substitute this bound into the identity in
\Cref{lem: regret-to-dual-regret} and obtain the claimed inequality.
\end{proof}

\section{Missing Results and Proofs of \Cref{section: main-section-on-SPPE}}
\label{app: SPPE-proofs}

\subsection{Proof of \Cref{lem: strong-convexity-F0}}

\strongconvexityFzero*

\begin{proof}
In the construction of $F_0$ in \Cref{def: F0-distribution}, we uniformly select a pair of bidders with nonzero utility. Therefore, bidder $i$ belongs to the selected pair with probability
$2/n$. Conditional on bidder $i$ being selected, we have the following,
\begin{itemize}
    \item $v_i$ is uniform on $[0, \vmax]$.
    \item $\max_{k\neq i} \alpha_{k, t} v_k$ is uniform on $[0, \alpha_{j, t} \vmax]$ (where $j$ is the other selected bidder) and is independent of $v_i$. 
\end{itemize}

Conditional on the complement event where bidder $i$ is not selected, we simply have $D_{i,t}^F (\mu; \vec \alpha_{-i, t}) = \rho_i \mu$. 

Notice that the integral term in each fluid dual loss is upper-bounded by $\vmax$. 
By the bounded convergence theorem, $D_{i,t}^F(\cdot; \vec \alpha_{-i, t})$ is differentiable. We next calculate the difference of derivatives for
$0\leq\mu_1<\mu_2\leq\muupper_i$. By a conditional probability argument, 
\begin{align*}
&\frac{\partial}{\partial\mu}D_{i,t}^{F_0}
       (\mu_2;\vec\alpha_{-i,t})
-\frac{\partial}{\partial\mu}D_{i,t}^{F_0}
       (\mu_1;\vec\alpha_{-i,t})\\
&\quad=
\frac{2}{n(n-1)}\sum_{j\ne i}
\int_0^{\alpha_{j,t}\vmax}\int_0^{\vmax}
\frac{d}{\alpha_{j,t}\vmax^2}
\indi{(1+\mu_1)d<v\leq(1+\mu_2)d}
\,\dd v\,\dd d\\
&\quad\overset{\text{(a)}}{\geq}
\frac{2}{n(n-1)}\sum_{j\ne i}
\int_0^{\vmax/(1+\max_k\muupper_k)}
\int_{(1+\mu_1)d}^{(1+\mu_2)d}
\frac{d}{\alpha_{j,t}\vmax^2}\,\dd v\,\dd d\\
&\quad=
\frac{2(\mu_2-\mu_1)}{n(n-1)\vmax^2}
\left(\sum_{j\ne i}\frac{1}{\alpha_{j,t}}\right)
\int_0^{\vmax/(1+\max_k\muupper_k)}d^2\,\dd d\\
&\quad\overset{\text{(b)}}{\geq}
\frac{2(\mu_2-\mu_1)}{n\vmax^2}
\int_0^{\vmax/(1+\max_k\muupper_k)}d^2\,\dd d\\
&\quad=
\frac{2(\mu_2-\mu_1)}{3n\vmax^2}
\left(\frac{\vmax}{1+\max_k\muupper_k}\right)^3\\
&\quad=
\frac{2\vmax}{3n(1+\max_k\muupper_k)^3}(\mu_2-\mu_1)
=\sigma_0(\mu_2-\mu_1),
\end{align*}
where in step (a) we restrict the outer integral to
$0\leq d\leq\vmax/(1+\max_k\muupper_k)$, and in step (b) we use $\alpha_{j,t}\leq1$ to obtain
$\sum_{j\ne i}1/\alpha_{j,t}\geq n-1$.
This proves $\sigma_0$-strong convexity of the dual loss 
on $[0,\muupper_i]$.
\end{proof}

\subsection{Proof of \Cref{lem: low-regret-implies-approximate-SPPE}}

\begingroup
\sppefullstatementtrue
\sppeextraction*
\endgroup

\begin{proof}
We construct the pacing multipliers $(\alpha_i')_{i\in[n]}$ and allocations $(\hat x_{i}'(\vec v))$ of the 
approximate fluid SPPE, and then show the four approximate SPPE properties in \Cref{def: approximate-fluid-SPPE}.

We construct the approximate SPPE pacing multipliers using the average of online dual variables, $\mu_i'=T^{-1}\sumt\mu_{i,t}$, and
set $\alpha_i'=1/(1+\mu_i')$. Since $\mu_i'\in[0,\muupper_i]$, we have
$\alpha_i'\in[\alphalower_i,1]$.

We next construct the approximate SPPE allocations. For a parameter $\iota>0$ to be chosen later, we consider the set of timesteps where
the $\ell_{\infty}$ distance deviation of the dual is greater than $\iota$, i.e., 
$\mathcal T=\{t\in[T]:\max_i|\mu_{i,t}-\mu_i'|>\iota\}$.
We define the allocation
\begin{equation*}
\hat x_{i,t}(\vec v)=
\begin{cases}
x_{i,t}(\vec v), & t\notin\mathcal T, \vec v \in \supp(F)\\
\indi{i=\min\operatorname*{arg\,max}_{j\in[n]}\alpha_j'v_j},
& \text{otherwise}.
\end{cases}
\end{equation*}
Intuitively, for steps with large dual deviation, we force the allocation to be consistent with the pacing multipliers $(\alpha_i')_{i\in[n]}$.
We then construct the offline fluid allocation by taking the average over the adjusted online allocations: $$\hat x_i'(\vec v)=T^{-1}\sumt\hat x_{i,t}(\vec v).$$

It is useful to bound the deviations from $\mu_i'$. Since $\sumt(\mu_{i,t}-\mu_i')=0$, 
by \eqref{eq: extraction-variance-hypothesis} we have
\begin{align*}
\frac1T\sum_i\sumt(\mu_{i,t}-\mu_i')^2
+\sum_i(\mu_i'-\mu_i)^2
&=\frac1T\sum_i\sumt(\mu_{i,t}-\mu_i)^2 \leq\delta_2.
\end{align*}
This implies $|\mu_i'-\mu_i|\leq\sqrt{\delta_2}$, and
\begin{align}
\label{eq: extraction-exceptional-rounds}
\frac{|\mathcal T|}{T}
&\leq\frac{1}{T\iota^2}
\sum_{t\in\mathcal T}\max_i|\mu_{i,t}-\mu_i'|^2
\leq\frac{1}{T\iota^2}\sumt\sum_i(\mu_{i,t}-\mu_i')^2
\leq\frac{\delta_2}{\iota^2}.
\end{align}
For $t\notin\mathcal T$, we have $|\mu_{i,t}-\mu_i'|\leq\iota$,
so
\begin{align*}
\frac{\alpha_{i,t}}{\alpha_i'}
&=1+\frac{\mu_i'-\mu_{i,t}}{1+\mu_{i,t}}
\leq1+\iota,\\
\frac{\alpha_i'}{\alpha_{i,t}}
&=1+\frac{\mu_{i,t}-\mu_i'}{1+\mu_i'}
\leq1+\iota.
\end{align*}

Both the highest and second-highest bids are coordinatewise monotone
and positively homogeneous in the bid vector $(\alpha_jv_j)_{j\in[n]}$.
We therefore obtain, for $t\notin\mathcal T$,
\begin{equation}
\label{eq: relative-price-comparison}
\begin{aligned}
\frac{h(\vec v;\vec\alpha')}{1+\iota}
&\leq h(\vec v;\vec\alpha_t)
\leq(1+\iota)h(\vec v;\vec\alpha'),\\
\frac{p(\vec v;\vec\alpha')}{1+\iota}
&\leq p(\vec v;\vec\alpha_t)
\leq(1+\iota)p(\vec v;\vec\alpha').
\end{aligned}
\end{equation}

Now we show that $( (\alpha'_i), (\hat x'_i(\vec v)))$ satisfy conditions (a)--(d) in
the definition of approximate fluid SPPE (\Cref{def: approximate-fluid-SPPE}).

\begin{itemize}
\item[(a)] \textit{Allocation to the near-highest bidder.}
For $\vec v\in\supp(F)$, $t\notin\mathcal T$, and $x_{i,t}(\vec v)>0$,
we have
\begin{align*}
\alpha_i'v_i
&\geq\frac{\alpha_{i,t}v_i}{1+\iota}
=\frac{h(\vec v;\vec\alpha_t)}{1+\iota}
\geq\frac{h(\vec v;\vec\alpha')}{(1+\iota)^2}.
\end{align*}
For $t\in\mathcal T$ and for $\vec v\notin\supp(F)$,
we allocate only to highest bidders under
$\vec\alpha'$.
Thus $\hat x_i'(\vec v)>0$ implies
$\alpha_i'v_i\geq h(\vec v;\vec\alpha')/(1+\iota)^2$. We then have condition (a) whenever we have $(1+\iota)^{-2}\geq1-\gamma$.
\item[(b)] \textit{Market clearance.} This is obviously satisfied by the construction itself.

\item[(c)] \textit{Approximate budget feasibility.}
To upper-bound the budget expense, we compare the spending under $\vec\alpha'$
with the online spending.
When $x_{i,t}(\vec v)>0$, we have
$p(\vec v;\vec\alpha_t)\leq\alpha_{i,t}v_i\leq v_i$,
and thus $c_{i,t}^F\leq\int v_i\,\dd F(\vec v)$.
We can then upper-bound the spending as follows, 
\begin{align*}
\int\hat x_i'(\vec v)p(\vec v;\vec\alpha')\,\dd F(\vec v)
&\leq\frac1T\sum_{t\notin\mathcal T}
\int x_{i,t}(\vec v)p(\vec v;\vec\alpha')\,\dd F(\vec v)
+\frac{|\mathcal T|}{T}\int v_i\,\dd F(\vec v)\\
&\overset{\text{(a)}}{\leq}
\frac{1+\iota}{T}\sum_{t\notin\mathcal T}c_{i,t}^F
+\frac{|\mathcal T|}{T}\int v_i\,\dd F(\vec v)\\
&\overset{\text{(b)}}{\leq}
(1+\iota)(1+\delta_1)\rho_i
+\frac{\delta_2}{\iota^2}\int v_i\,\dd F(\vec v),
\end{align*}
where step (a) is by \eqref{eq: relative-price-comparison}, and step (b) invokes condition \eqref{eq: extraction-budget-hypothesis} and
the bound in \eqref{eq: extraction-exceptional-rounds}.
Therefore, we know that condition (c) holds if, for all $i \in [n]$,
\begin{equation}
\label{eq: extraction-budget-condition}
(1+\iota)(1+\delta_1)
+\frac{\delta_2}{\iota^2}
\frac{\int v_i\,\dd F(\vec v)}{\rho_i}
\leq1+\gamma.
\end{equation}

\item[(d)] \textit{Approximate no unnecessary pacing.}
We consider a bidder with $\alpha_i'<1-\gamma$. Since $\sqrt{\delta_2} \leq \gamma/2$, we have
\begin{align*}
\mu_i
\geq\mu_i'-\sqrt{\delta_2} >\frac{\gamma}{1-\gamma}-\sqrt{\delta_2}
\geq\frac{\gamma}{2}.
\end{align*}
By \eqref{eq: extraction-slack-hypothesis}, we therefore have
$T^{-1}\sumt c_{i,t}^F\geq\rho_i-2\delta_3/\gamma$.
To lower bound a bidder's spending, we have
\begin{align*}
\int\hat x_i'(\vec v)p(\vec v;\vec\alpha')\,\dd F(\vec v)
&\geq\frac1T\sum_{t\notin\mathcal T}
\int x_{i,t}(\vec v)p(\vec v;\vec\alpha')\,\dd F(\vec v)\\
&\overset{\text{(a)}}{\geq}
\frac{1}{T(1+\iota)}\sum_{t\notin\mathcal T}c_{i,t}^F\\
&=\frac{T^{-1}\sumt c_{i,t}^F
-T^{-1}\sum_{t\in\mathcal T}c_{i,t}^F}{1+\iota}\\
&\overset{\text{(b)}}{\geq}
\frac{\rho_i-2\delta_3/\gamma
-(\delta_2/\iota^2)\int v_i\,\dd F(\vec v)}{1+\iota},
\end{align*}
where in (a) we use \eqref{eq: relative-price-comparison}, and in (b)
we use \eqref{eq: extraction-slack-hypothesis},
\eqref{eq: extraction-exceptional-rounds}, and the upper bound on $c_{i,t}^F$.
To lower bound the right-hand side by $(1-\gamma)\rho_i$, we will require
\begin{equation}
\label{eq: extraction-slack-condition}
\frac{2\delta_3}{\gamma\rho_i}
+\frac{\delta_2}{\iota^2}
\frac{\int v_i\,\dd F(\vec v)}{\rho_i}
+\iota(1-\gamma)\leq\gamma.
\end{equation}
Therefore,  \eqref{eq: extraction-slack-condition} will imply (d) by contradiction.
\end{itemize}

\paragraph{Fixing the parameters.}

We now take $\iota=\gamma/2$ and use the bounds on
$\delta_1,\delta_2,\delta_3$ in the lemma statement. We show that this satisfies all 
inequality requirements for the above conditions (a)--(d).
For condition (a), we apply Bernoulli's inequality to obtain
$(1+\iota)^{-2}\geq1-2\iota=1-\gamma$.

From the bound on $\delta_2$, we obtain, for all $i$,
\begin{align*}
\frac{\delta_2}{\iota^2}
\frac{\int v_i\,\dd F(\vec v)}{\rho_i}
&\leq\frac{4\gamma}{9}
\frac{\int v_i\,\dd F(\vec v)/\rho_i}
{1+\max_j\int v_j\,\dd F(\vec v)/\rho_j}\leq\frac{4\gamma}{9}.
\end{align*}
We can then check \eqref{eq: extraction-budget-condition} as follows,
\begin{align*}
(1+\iota)(1+\delta_1)
+\frac{\delta_2}{\iota^2}
\frac{\int v_i\,\dd F(\vec v)}{\rho_i}
&\leq\left(1+\frac{\gamma}{2}\right)
\left(1+\frac{\gamma^3}{36}\right)+\frac{4\gamma}{9}\\
&=1+\frac{17\gamma}{18}+\frac{\gamma^3}{36}
+\frac{\gamma^4}{72}\\
&\leq1+\frac{71\gamma}{72}\leq1+\gamma,
\end{align*}
Finally, for \eqref{eq: extraction-slack-condition}, we have
\begin{align*}
\frac{2\delta_3}{\gamma\rho_i}
+\frac{\delta_2}{\iota^2}
\frac{\int v_i\,\dd F(\vec v)}{\rho_i}
+\iota(1-\gamma)
&\leq\frac{\gamma^2}{18}+\frac{4\gamma}{9}
+\frac{\gamma}{2}(1-\gamma)\\
&=\frac{17\gamma}{18}-\frac{4\gamma^2}{9}
\leq\gamma,
\end{align*}
We conclude that $(\vec\alpha',(\hat x_i')_{i\in[n]})$ satisfies all four
conditions of a $\gamma$-approximate fluid SPPE.

\end{proof}

\subsection{Proof of \Cref{lemma: fluid-hardness}}

\fluidhardness*

\begin{proof}
We use the following definition of approximate (discrete) SPPE from
\citet{liu2026tight}.

\begin{definition}
\label{def: Liu-approximate-SPPE}
For $\lambda_1,\lambda_2,\lambda_3\in[0,1)$, a
$(\lambda_1,\lambda_2,\lambda_3)$-approximate SPPE for $m$ items with value vectors
$\vec v_j=(v_{1j},\ldots,v_{nj})$, $j\in[m]$, and budgets
$B_i>0$, $i\in[n]$, is a tuple
$(\vec\alpha,(x_{ij})_{i\in[n],j\in[m]})$, with
$\vec\alpha\in[0,1]^n$ and $x_{ij}\geq0$ that satisfy the following, 
\begin{enumerate}
\item[(a)] Allocation to near highest bidders: $x_{ij}>0\implies
\alpha_iv_{ij}\geq(1-\lambda_1)h(\vec v_j;\vec\alpha)$ for all
$i\in[n]$ and $j\in[m]$.
\item[(b)] Market clearance: $\sumi x_{ij}\leq1$ and
$h(\vec v_j;\vec\alpha)>0\implies\sumi x_{ij}=1$ for all $j\in[m]$.
\item[(c)] Budget feasibility:
$\sum_{j=1}^m x_{ij}p(\vec v_j;\vec\alpha)\leq B_i$
for all $i\in[n]$.
\item[(d)] Approximate no unnecessary pacing:
$\sum_{j=1}^m x_{ij}p(\vec v_j;\vec\alpha)
<(1-\lambda_2)B_i\implies\alpha_i\geq1-\lambda_3$
for all $i\in[n]$.
\end{enumerate}
\end{definition}

We claim that hardness is preserved after allowing approximations on budget feasibility.
\begin{claim}
\label{clm: relaxed-discrete-SPPE-hardness}
For any fixed $\lambda\in(0,1)$, computing a
$(\lambda,\lambda,\lambda)$-approximate (discrete) SPPE with (c) in
\Cref{def: Liu-approximate-SPPE} relaxed to
$\sum_{j=1}^m x_{ij}p(\vec v_j;\vec\alpha)\leq(1+\lambda)B_i$
is still $\mathsf{PPAD}$-hard.
\end{claim}

\begin{proof}[Proof of the claim]
We reduce from the $\mathsf{PPAD}$-hard problem in \citet[Theorem A.2]{liu2026tight} with
$(\lambda_1,\lambda_2,\lambda_3)=(\lambda,2\lambda/(1+\lambda),\lambda)$.
Given an instance of this problem with budgets $B_i'$, we construct
a new instance with the same bidders, items, and valuations, and
budgets $B_i=B_i'/(1+\lambda)$.
For a solution to the budget-relaxed problem, we have
\begin{align*}
\sum_{j=1}^m x_{ij}p(\vec v_j;\vec\alpha)
&\leq(1+\lambda)B_i=B_i',\\
(1-\lambda)B_i
&=\frac{1-\lambda}{1+\lambda}B_i'
=\left(1-\frac{2\lambda}{1+\lambda}\right)B_i'.
\end{align*}
It is then easy to check the conditions (a)--(d) for a $(\lambda_1, \lambda_2, \lambda_3)$-approximate equilibrium for the unrelaxed notion. The reduction can clearly be done in polynomial time.
\end{proof}

In the rest of the proof, we will call an equilibrium in
\Cref{clm: relaxed-discrete-SPPE-hardness}
a $\lambda$-approximate SPPE.
We fix $\lambda\in(0,1)$ and consider an instance of this discrete problem, with $m$ items with value vectors $(\vec v_1, \cdots, \vec v_m)$. We construct the fluid by setting the budget as $\rho_i = B_i/m$, and the fluid distribution as
$F=(1-\eta)F_1+\eta F_0$, where $F_1=\frac1m\sum_{j=1}^m\delta_{\vec v_j}$ is the empirical distribution on the $m$ items; $\delta_{\vec v_j}$ is the Dirac probability measure at $\vec v_j$.

We next show that, given a $\gamma$-approximate fluid SPPE for $F$, by restricting 
to the $m$ items one can extract a $\lambda$-approximate (discrete) SPPE
whenever $\eta,\gamma\in(0,1)$ satisfy
\begin{equation}
\label{eq: mixture-hardness-margins}
\gamma+(1+\lambda)\eta\leq\lambda,\qquad
\gamma+\frac{\eta m\vmax}{n\min_i B_i}\leq\lambda.
\end{equation}

Concretely, given a $\gamma$-approximate fluid SPPE
$(\vec\alpha',(x_i'(\vec v))_{i\in[n]})$,
we construct the discrete solution as
$\alpha_i=\alpha_i'$ and $x_{ij}=x_i'(\vec v_j)$.
On each value point $\vec v_j, j\in [m]$, both solutions share the same bids from all bidders and the same allocation. We then have for each bidder $i\in[n]$,
\begin{equation}
\label{eq: mixture-spending-identity}
\begin{aligned}
\underbrace{\int x_i'(\vec v)p(\vec v;\vec\alpha')\,\dd F(\vec v)}_{\text{spending in the fluid market}}
=&\frac{1-\eta}{m}\cdot \underbrace{\sum_{j=1}^m x_{ij}p(\vec v_j;\vec\alpha)}_{\text{spending in the discrete market}}+\eta\int x_i'(\vec v)p(\vec v;\vec\alpha')\,\dd F_0(\vec v).
\end{aligned}
\end{equation}

We bound the last term in \eqref{eq: mixture-spending-identity} as follows. Notice that in $F_0$ bidder $i$ has nonzero value with probability $2/n$ and has conditional mean value $\vmax / 2$. Since the price is at most $v_i$ when bidder $i$ receives positive allocation, we have
\begin{equation}
\label{eq: regularizer-payment-bound}
\int x_i'(\vec v)p(\vec v;\vec\alpha')\,\dd F_0(\vec v)
\leq \int v_i \dd F_0(\vec v) = \frac{2}{n} \cdot \frac{\vmax}{2} = \frac{\vmax}{n}.
\end{equation}

We next show that $(\vec \alpha, (x_{ij})_{i\in [n], j\in [m]})$ indeed constitutes a $\lambda$-approximate SPPE for the (discrete) market. 

\begin{itemize}
    \item [(a)] \textit{Allocation to the near-highest bidder.} By \eqref{eq: mixture-hardness-margins}, we have $\gamma<\lambda$. For $x_{ij}>0$, by our construction, we have
$
\alpha_i v_{ij}
=\alpha_i' v_{ij}
\geq(1-\gamma)h(\vec v_j;\vec\alpha')\geq(1-\lambda)h(\vec v_j;\vec\alpha).$
    \item [(b)] \textit{Market clearance.} Obvious.
    \item [(c)] \textit{Approximate budget feasibility.} By the decomposition in \eqref{eq: mixture-spending-identity},
\begin{align*}
\sum_{j=1}^m x_{ij}p(\vec v_j;\vec\alpha)
&\leq\frac{m}{1-\eta}
\int x_i'(\vec v)p(\vec v;\vec\alpha')\,\dd F(\vec v)\\
&\leq\frac{1+\gamma}{1-\eta}B_i
\leq(1+\lambda)B_i,
\end{align*}
where in the last inequality we use the first bound in
\eqref{eq: mixture-hardness-margins}.
    \item [(d)] \textit{Approximate no unnecessary pacing.} For a bidder $i\in[n]$ with $\sum_{j=1}^m x_{ij}p(\vec v_j;\vec\alpha)<(1-\lambda)B_i$, by \eqref{eq: mixture-spending-identity} and
\eqref{eq: regularizer-payment-bound},
\begin{align*}
\int x_i'(\vec v)p(\vec v;\vec\alpha')\,\dd F(\vec v)
&\leq\frac{1-\eta}{m}\sum_{j=1}^m x_{ij}p(\vec v_j;\vec\alpha)
+\frac{\eta \vmax}{n}\\
&<\frac{1-\lambda}{m} \cdot B_i+\frac{\eta \vmax}{n}\\
&\leq\left(\frac{1-\lambda}{m}+
\frac{\eta \vmax}{n\min_k B_k}\right)B_i
\leq \frac{1-\gamma}{m}B_i,
\end{align*}
where in the last step we use the second inequality in \eqref{eq: mixture-hardness-margins}. By the approximate no unnecessary pacing property of the fluid SPPE solution, we know that $\alpha_i = \alpha_i' \geq 1-\gamma \geq 1-\lambda$.
\end{itemize}

Thus the constructed solution is a $\lambda$-approximate discrete SPPE.
The reduction construction can clearly be done in polynomially many operations. 

It only remains for us to argue that the factor $m\vmax / (n \min_{i\in [n]} B_i)$ in \eqref{eq: mixture-hardness-margins} is a constant for a hard instance of the discrete problem. Indeed, for the hard instance in \citet[Appendix A.1]{liu2026tight}, after dividing the budgets as in \Cref{clm: relaxed-discrete-SPPE-hardness},
the ratio $m\vmax/(n\min_{i\in [n]} B_i)$ is bounded by a constant that only depends on $\lambda$.
We can therefore fix $\gamma=\lambda/4$ and choose a sufficiently
small $\eta>0$ independent of the instance, so that
\eqref{eq: mixture-hardness-margins} holds for the hard instance family.

\paragraph{Useful bounds for the hard family.} We provide parameters about the 
(discrete) hard market instances, which will be useful in
the proof of \Cref{thm: fixed-polynomial-regret}. We refer the readers to \citet[Appendix A.1]{liu2026tight} for details of the hardness construction.

In the hard instances, $k = \lceil8192(1+\lambda)^2/(1-\lambda)^3\rceil+1$ is the number of items created for each directed edge of the circuit used in the reduction. In \citet{liu2026tight}, 
all valuations are bounded by $\vmax' = \frac{1024(1+\lambda)^2}{k(1-\lambda)^4}\left(\frac{4}{1-\lambda}\right)^{k-1}$, and all budgets are lower-bounded by $b' = \frac{(1-\lambda)^2}{32(1+\lambda)}$.
We scale all values and budgets uniformly by the same multiplicative factor so that they all lie in $[0, 1]$. After also accounting for the budget rescaling in \Cref{clm: relaxed-discrete-SPPE-hardness}, the budgets of all bidders are bounded below by $b = \frac{b'}{(1+\lambda)\vmax'}$, which depends only on $\lambda$. We therefore have $\vmax = 1, B_i \in [b, 1]$ in our discrete hard market instances. 

Another useful property is that the number of items is linear in the number of bidders. By the structure of the construction in \citet{liu2026tight}, each bidder $i\in [n]$ has positive values on at most $3k$ items, and each item $j\in [m]$ receives positive valuation from exactly $2$ bidders. Therefore, we have $2m \leq 3kn$.

\end{proof}

\subsection{Proof of \Cref{lem: uniform-fluid-stochastic-transfer}}

\fluidstochastictransfer*

\begin{algorithm}[H]
\caption{Simulation of $\mathcal A^F$ from $\mathcal A$}
\label{alg: stochastic-to-fluid-simulation}
\KwIn{The inputs to the fluid bidding process, the stochastic dynamic $\mathcal A$.}
Initialize $\mathcal A$ with the same auction parameters\;
\For{$t=1,\ldots,T$}{
  Choose the fluid multiplier $\vec\alpha_t$ to be the same as in $\mathcal{A}$.
  
  \For{each $\vec v$}{
    Set the allocation $x_{i,t}(\vec v)$ equal to what $\mathcal{A}$ would choose at step $t$ with value $\vec v$ and pacing multiplier $\vec \alpha_t$.
  } 
  Sample $\vec v_t\sim F$ independently and simulate the $t$'th step of $\mathcal{A}$. 
}
\end{algorithm}

\begin{proof}
We construct $\mathcal{A}^F$ by simulating a run of the stochastic bidding dynamic $\mathcal{A}$, and 
use its pacing multipliers as the fluid pacing multipliers; see \Cref{alg: stochastic-to-fluid-simulation} for details.

Let $\mathcal H_t$ be the history up to step $t$ before
$\vec v_t$ is generated.
Conditional on $\mathcal H_t$, the fluid allocation functions are fixed. Since $\vec v_t\sim F$, we have
\begin{align*}
\E[c_{i,t}\mid\mathcal H_t]
&=\int x_{i,t}(\vec v)p(\vec v;\vec\alpha_t)\,\dd F(\vec v)
=c_{i,t}^F,\\
\E[f_{i,t}(x_{i,t})\mid\mathcal H_t]
&=f_{i,t}^F(x_{i,t}).
\end{align*}
Therefore, we have
\begin{align*}
\E\left[\left(\sumt(c_{i,t}-c_{i,t}^F)\right)^2\right]
&=\sumt\E[(c_{i,t}-c_{i,t}^F)^2]
\leq T\vmax^2,\\
\E\left[\left(\sumt
\bigl(f_{i,t}(x_{i,t})-f_{i,t}^F(x_{i,t})\bigr)\right)^2\right]
&=\sumt\E\left[
\bigl(f_{i,t}(x_{i,t})-f_{i,t}^F(x_{i,t})\bigr)^2\right]
\leq T\vmax^2.
\end{align*}
Applying the Cauchy-Schwarz inequality, 
\begin{equation}
\label{eq: utility-payment-martingale}
\begin{aligned}
\E\left[\left|\sumt(c_{i,t}-c_{i,t}^F)\right|\right]
\leq\vmax\sqrt T, \quad
\E\left[\left|\sumt
\bigl(f_{i,t}(x_{i,t})-f_{i,t}^F(x_{i,t})\bigr)\right|\right]
\leq\vmax\sqrt T.
\end{aligned}
\end{equation}
By the inequality $(a+b)_+\leq a_++|b|$ for all $a, b \in \mathbb{R}$, we get
\begin{align*}
\E\left[\left(\sumt c_{i,t}^F-T\rho_i\right)_+\right]
&\leq\E\left[\left(\sumt c_{i,t}-T\rho_i\right)_+\right]
+\E\left[\left|\sumt(c_{i,t}^F-c_{i,t})\right|\right]\\
&\leq\E\left[\left(\sumt c_{i,t}-T\rho_i\right)_+\right]
+\vmax\sqrt T,
\end{align*}
which proves the second inequality in \Cref{lem: uniform-fluid-stochastic-transfer}.

We next bound the difference between the two hindsight benchmarks through their duals. 
Define the sum of cumulative dual differences as follows,
\begin{equation*}
Z_{i,t}(\mu)=\sum_{s=1}^t
\bigl(D_{i,s}(\mu;\vec v_s,\vec\alpha_{-i,s})
-D_{i,s}^F(\mu;\vec\alpha_{-i,s})\bigr),
\qquad Z_{i,0}(\mu)=0.
\end{equation*}
By a similar conditional probability argument, we have
$\E[Z_{i,t}(\mu)-Z_{i,t-1}(\mu)\mid\mathcal H_t]=0$, i.e., $Z_{i,t}(\mu)$ is a martingale. 
Since
$|Z_{i,t}(\mu)-Z_{i,t-1}(\mu)|\leq\vmax$, 
we apply the conditional Hoeffding bound. For any $s\in \mathbb{R}$, 
\begin{align*}
\E[e^{sZ_{i,t}(\mu)}]
&=\E\left[e^{sZ_{i,t-1}(\mu)}
\E\left[e^{s(Z_{i,t}(\mu)-Z_{i,t-1}(\mu))}
\mid\mathcal H_t\right]\right]\\
&\leq e^{s^2\vmax^2/2}\E[e^{sZ_{i,t-1}(\mu)}]
\leq e^{s^2t\vmax^2/2}.
\end{align*}

We take the grid discretization 
$\mathcal G=\{k/T:k=0,\ldots,\lfloor\muupper_i T\rfloor\}
\cup\{\muupper_i\}$.
Then $|\mathcal G|\leq\lceil\muupper_iT\rceil+1$ and any point in 
$[0,\muupper_i]$ is within $1/T$ distance of a grid point.

Letting $s=\sqrt{2\log(2|\mathcal G|)/(T\vmax^2)}$, we have
\begin{align*}
\E\left[\max_{\mu\in\mathcal G}|Z_{i,T}(\mu)|\right]
&\leq\frac1s\E\left[\log\sum_{\mu\in\mathcal G}
\left(e^{sZ_{i,T}(\mu)}+e^{-sZ_{i,T}(\mu)}\right)\right]\\
&\overset{\text{(a)}}{\leq}\frac1s\log\sum_{\mu\in\mathcal G}
\left(\E[e^{sZ_{i,T}(\mu)}]+\E[e^{-sZ_{i,T}(\mu)}]\right)\\
&\leq\frac{\log(2|\mathcal G|)}s+\frac{sT\vmax^2}{2}
=\vmax\sqrt{2T\log(2|\mathcal G|)},
\end{align*}
where step (a) is by Jensen's inequality.
Using the inequality $|a_+-b_+|\leq|a-b|$ for $a, b \in \mathbb{R}$, we get
\begin{align*}
|Z_{i,T}(\mu)-Z_{i,T}(\nu)|
&\leq\sumt\left(d_{i,t}
+\int\max_{j\ne i}\alpha_{j,t}v_j\,\dd F(\vec v)\right)|\mu-\nu| \leq2T\vmax|\mu-\nu|.
\end{align*}
For any $\mu\in[0,\muupper_i]$, choose $\nu\in\mathcal G$ with
$|\mu-\nu|\leq1/T$. Then
\begin{align*}
|Z_{i,T}(\mu)|
&\leq|Z_{i,T}(\nu)|+2T\vmax|\mu-\nu|\leq\max_{\nu\in\mathcal G}|Z_{i,T}(\nu)|+2\vmax.
\end{align*}
Taking the supremum and expectations, we obtain
\begin{equation}
\label{eq: explicit-uniform-dual-transfer}
\E\left[\sup_{\mu\in[0,\muupper_i]}|Z_{i,T}(\mu)|\right]
\leq\vmax\sqrt{2T\log\left(2(\lceil\muupper_i T\rceil+1)\right)}
+2\vmax.
\end{equation}

For $\mu\in[0,\muupper_i]$, we have
\begin{align*}
\sumt D_{i,t}(\mu;\vec v_t,\vec\alpha_{-i,t})
&=\sumt D_{i,t}^F(\mu;\vec\alpha_{-i,t})+Z_{i,T}(\mu)\\
&\leq\sumt D_{i,t}^F(\mu;\vec\alpha_{-i,t})
+\sup_{\nu\in[0,\muupper_i]}|Z_{i,T}(\nu)|.
\end{align*}
We minimize over $\mu$ and then interchange the sampled and fluid
objectives to bound the difference of their minima.
Using the $\vmax$ bound on both duality gaps, we obtain
\begin{equation*}
\left|\opt_i^F(\bm\alpha_{-i})-\opt_i(\bm v;\bm\alpha_{-i})\right|
\leq\sup_{\mu\in[0,\muupper_i]}|Z_{i,T}(\mu)|+2\vmax.
\end{equation*}
Combining \eqref{eq: utility-payment-martingale}
and \eqref{eq: explicit-uniform-dual-transfer},
\begin{align*}
&\E\left[\left|\regret_{i,\mathcal A^F}^F-\regret_{i,\mathcal A}\right|\right]\\
&\quad\leq
\E\left[\left|\opt_i^F(\bm\alpha_{-i})-\opt_i(\bm v;\bm\alpha_{-i})\right|\right]
+\E\left[\left|\sumt
\bigl(f_{i,t}^F(x_{i,t})-f_{i,t}(x_{i,t})\bigr)\right|\right]\\
&\quad\leq\E\left[\sup_{\mu\in[0,\muupper_i]}|Z_{i,T}(\mu)|\right]
+2\vmax+\vmax\sqrt T\\
&\quad\leq\vmax\left[
\sqrt{2T\log\left(2(\lceil\muupper_iT\rceil+1)\right)}
+\sqrt T+4\right].
\end{align*}
This proves the first inequality in \Cref{lem: uniform-fluid-stochastic-transfer}.
\end{proof}

\subsection{Discretization}
\label{app: SPPE-discretization}

In the proof of the main theorem, we will use a discretized version of the regularizer $F_0$, which 
can be efficiently encoded; see \Cref{def: F0-distribution-discretized}. 
We bound the discretization error (\Cref{lem: regularizer-grid-error}), and show that 
the discretized regularizer also provides strong-convexity-like properties (\Cref{lem: finite-dual-curvature}).
We will consider $\widetilde F=(1-\eta)F_1+\eta F_{0,K}$ where $F_{0, K}$ is the discrete regularizer.
We remark that \Cref{lemma: fluid-hardness} also applies when the regularizer
is $F_{0, K}$ in place of $F_0$. 

\begin{definition}
\label{def: F0-distribution-discretized}
For any integer $K\geq1$, we define $F_{0,K}$ by the construction in
\Cref{def: F0-distribution}, replacing the uniform distribution in
Step 2 with $\mathrm{Uniform}\{0,\vmax/K,\ldots,(K-1)\vmax/K\}$.
\end{definition}

\begin{lemma}
\label{lem: regularizer-grid-error}
For any $\mu\in[0,\muupper_i]$, it holds that
\begin{equation*}
\left|D_{i,t}^{F_0}(\mu;\vec\alpha_{-i,t})
      -D_{i,t}^{F_{0,K}}(\mu;\vec\alpha_{-i,t})\right|
\leq\xi_K:= 2\vmax(2+\max_j\muupper_j)/(nK).
\end{equation*}
\end{lemma}

\begin{proof}
We sample $\vec v\sim F_0$ and set
$v_k'=(\vmax/K)\min\{\lfloor Kv_k/\vmax\rfloor,K-1\}$ for all $k$.
Then $\vec v'\sim F_{0,K}$ and $|v_k-v_k'|\leq\vmax/K$, so
\begin{align*}
&\left|D_{i,t}^{F_0}(\mu;\vec\alpha_{-i,t})
      -D_{i,t}^{F_{0,K}}(\mu;\vec\alpha_{-i,t})\right|\\
&\quad\overset{\text{(a)}}{\leq}\frac2n
\E\left[
|v_i-v_i'|+(1+\mu)
\left|\max_{j\ne i}\alpha_{j,t}v_j
      -\max_{j\ne i}\alpha_{j,t}v_j'\right|
\,\middle|\,i\text{ is selected}\right]\\
&\quad\leq\frac2n
\E\left[
|v_i-v_i'|+(1+\mu)\max_{j\ne i}\alpha_{j,t}|v_j-v_j'|
\,\middle|\,i\text{ is selected}\right]\\
&\quad\leq\frac2n\left[\frac{\vmax}{K}
+(1+\mu)\frac{\vmax}{K}\right]\\
&\quad=\frac{2\vmax(2+\mu)}{nK}
\leq\frac{2\vmax(2+\max_j\muupper_j)}{nK}=\xi_K,
\end{align*}
where in step (a) the affine terms cancel out, and we apply the inequality $|a_+-b_+|\leq|a-b|$ for any $a, b \in \mathbb{R}$.
\end{proof}

\begin{samepage}
\begin{lemma}
\label{lem: finite-dual-curvature}
For any $\mu_1,\mu_2\in[0,\muupper_i]$ with
$\mu_2\leq\max_k\int v_k\,\dd\widetilde F(\vec v)/\rho_k$,
and $g\in\partial D_{i,t}^{\widetilde F}(\mu_1;\vec\alpha_{-i,t})$, it holds that
\begin{equation*}
D_{i,t}^{\widetilde F}(\mu_2;\vec\alpha_{-i,t})
-D_{i,t}^{\widetilde F}(\mu_1;\vec\alpha_{-i,t})
-g(\mu_2-\mu_1)
\geq\frac{\widetilde\sigma}{4}(\mu_2-\mu_1)^2-4\eta\xi_K,
\end{equation*}
where $\widetilde\sigma=
\dfrac{\eta\vmax}
{6n\left(2+\max_k\frac{\int v_k\,\dd\widetilde F(\vec v)}{\rho_k}\right)
(1+\max_k\muupper_k)^2}$.
\end{lemma}
\end{samepage}

\begin{proof}
As in the proof of \Cref{lem: strong-convexity-F0}, we condition on
the event that $\{i,j\}$ is the pair of bidders with nonzero values, and integrate over bidder $i$'s value. We have
\begin{align*}
\frac{\dd^2}{\dd\mu^2}D_{i,t}^{F_0}(\mu;\vec\alpha_{-i,t})
&=\frac{2}{n(n-1)}\sum_{j\ne i}
\frac1{\alpha_{j,t}\vmax^2}
\int_0^{\min\{\alpha_{j,t}\vmax,\vmax/(1+\mu)\}}d^2\,\dd d\\
&=\frac{2\vmax}{3n(n-1)}\sum_{j\ne i}
\min\left\{\alpha_{j,t}^2,
\frac1{\alpha_{j,t}(1+\mu)^3}\right\}\\
&\geq\frac{2\vmax}{3n(n-1)}\sum_{j\ne i}
\frac{\alpha_{j,t}^2}{(1+\alpha_{j,t}(1+\mu))^3}\\
&=\frac{\dd^2}{\dd\mu^2}
\left[\frac{\vmax}{3n(n-1)}
\sum_{j\ne i}\frac1{1+\alpha_{j,t}(1+\mu)}\right].
\end{align*}
Multiplying both sides by $\eta$, and applying the convexity of $D_{i,t}^{F_1}$,
we know that the following function is strongly convex:

\begin{equation*}
D_{i,t}^F(\mu;\vec\alpha_{-i,t})
-\frac{\eta\vmax}{3n(n-1)}
\sum_{j\ne i}\frac1{1+\alpha_{j,t}(1+\mu)}
\end{equation*}
By the basic property of convex functions, 
\begin{align*}
&\frac{D_{i,t}^F(\mu_1;\vec\alpha_{-i,t})
       +D_{i,t}^F(\mu_2;\vec\alpha_{-i,t})}{2}
-D_{i,t}^F\left(\frac{\mu_1+\mu_2}{2};\vec\alpha_{-i,t}\right)\\
&\quad\geq\frac{\eta\vmax}{6n(n-1)}\sum_{j\ne i}
\left[
\frac1{1+\alpha_{j,t}(1+\mu_1)}
+\frac1{1+\alpha_{j,t}(1+\mu_2)}
-\frac4{2+\alpha_{j,t}(2+\mu_1+\mu_2)}
\right]\\
&\quad=\frac{\eta\vmax(\mu_2-\mu_1)^2}{6n(n-1)}
\sum_{j\ne i}
\frac{\alpha_{j,t}^2}
{(1+\alpha_{j,t}(1+\mu_1))(1+\alpha_{j,t}(1+\mu_2))
 (2+\alpha_{j,t}(2+\mu_1+\mu_2))}\\
&\quad\overset{\text{(a)}}{\geq}
\frac{\eta\vmax(\mu_2-\mu_1)^2}{12n(n-1)(2+\mu_2)}
\sum_{j\ne i}
\frac{\alpha_{j,t}^2}{(1+\alpha_{j,t}(1+\max_k\muupper_k))^2}\\
&\quad\overset{\text{(b)}}{\geq}
\frac{\eta\vmax(\mu_2-\mu_1)^2}
{48n(2+\mu_2)(1+\max_k\muupper_k)^2}\\
&\quad\overset{\text{(c)}}{\geq}
\frac{\widetilde\sigma}{8}(\mu_2-\mu_1)^2,
\end{align*}
where in steps (a) and (b) we apply various bounded conditions of the multiplier
and the dual variable, and in step (c) we use the fact that
$\mu_2\leq\max_k\int v_k\,\dd\widetilde F(\vec v)/\rho_k$.

By \Cref{lem: regularizer-grid-error}, we have
$|D_{i,t}^{\widetilde F}(\mu;\vec\alpha_{-i,t})
-D_{i,t}^F(\mu;\vec\alpha_{-i,t})|\leq\eta\xi_K$
on $[0,\muupper_i]$.
Therefore, 
\begin{align*}
&D_{i,t}^{\widetilde F}(\mu_2;\vec\alpha_{-i,t})
-D_{i,t}^{\widetilde F}(\mu_1;\vec\alpha_{-i,t})
-g(\mu_2-\mu_1)\\
&\quad\geq
D_{i,t}^{\widetilde F}(\mu_1;\vec\alpha_{-i,t})
+D_{i,t}^{\widetilde F}(\mu_2;\vec\alpha_{-i,t})
-2D_{i,t}^{\widetilde F}
\left(\frac{\mu_1+\mu_2}{2};\vec\alpha_{-i,t}\right)\\
&\quad\geq
D_{i,t}^F(\mu_1;\vec\alpha_{-i,t})
+D_{i,t}^F(\mu_2;\vec\alpha_{-i,t})
-2D_{i,t}^F\left(\frac{\mu_1+\mu_2}{2};\vec\alpha_{-i,t}\right)
-4\eta\xi_K\\
&\quad\geq\frac{\widetilde\sigma}{4}(\mu_2-\mu_1)^2-4\eta\xi_K.
\end{align*}
\end{proof}

\subsection{Proof of \Cref{thm: fixed-polynomial-regret}}

\budgetimpossibility*

\begin{proof}
In this proof, we use the SPPE hard instances and constants $\lambda,\eta,\gamma,k,b>0$
from the proof of \Cref{lemma: fluid-hardness}. 
For notational simplicity, we denote $H=12\eta^{-1}(2+3k/b)(3k/b)^2$ and 
$C\geq864H(1+H)^2/\gamma^3$; these parameters are constants independent of $n, T,$ and $\varepsilon$.
For the sake of contradiction, suppose both expectations in
the theorem statement are at most $T/(Cn^\varepsilon)$
for all bidders on every instance with $n\geq N_0$ and $T=\lceil q(n)\rceil$.

\paragraph{Construction of the fluid bidding instance.}
Consider the hard instance of SPPE from the proof of \Cref{lemma: fluid-hardness} with $s$ bidders. For $r=\max\{9,\lceil4/\varepsilon\rceil\}$,
we construct a fluid instance with $n=s^r$ bidders by adding $n-s$ auxiliary bidders who form a disjoint market component.
Specifically, let $\mathcal I=\{1,\ldots,s\}$ be the set of original bidders from the hard instance, and $\widetilde{F}$ be defined as in \Cref{def: F0-distribution-discretized} with $K=\lceil Cs^4\rceil$. To construct the joint value distribution $\widehat F$ of $n$ bidders,
we draw $(\widetilde v_i)_{i\in\mathcal I}\sim\widetilde F$ and
independently choose one of the following two cases:
\begin{equation}
v_i=\begin{cases}
\widetilde v_i, & i\in\mathcal I,\\
0, & i\notin\mathcal I,
\end{cases}
\quad\text{w.p. }1/2,
\qquad
v_i=\begin{cases}
0, & i\in\mathcal I,\\
1, & i\notin\mathcal I,
\end{cases}
\quad\text{w.p. }1/2.
\end{equation}
Intuitively, each of the two disjoint components is active with probability $1/2$. To construct the budgets, we set $\rho_i = 1/8$
for $i \notin \mathcal{I}$, and $\rho_i = B_i /(2m)$ for $i\in\mathcal{I}$, where we halve the per-round budgets used in \Cref{lemma: fluid-hardness} since
the corresponding arrivals now occur with probability $1/2$.

Given a fluid bidding process under $\widehat{F}$, we can construct a fluid bidding process for the $s$ bidders in $\mathcal{I}$ by restriction. Concretely, we keep all per-round budgets, pacing multipliers, and allocations except at value $\vec 0$ for bidders in $\mathcal{I}$. The resulting value distribution on $\mathcal{I}$ is then $F'=\frac12\widetilde F+\frac12\delta_{\vec0}$. Since bidders in $\mathcal{I}$ do not have positive valuations on the other market component, we have $u_{i,t}^{F'}=u_{i,t}^{\widehat F}$ and
$c_{i,t}^{F'}=c_{i,t}^{\widehat F}$. Consequently, 
\begin{equation*}
\regret_i^{F'}=\regret_i^{\widehat F},\qquad
\left(\sumt c_{i,t}^{F'}-T\rho_i\right)_+
=\left(\sumt c_{i,t}^{\widehat F}-T\rho_i\right)_+
\quad(i\in\mathcal I).
\end{equation*}

\paragraph{Low regret of the $s$-bidder fluid process.}
We show that if an $n$-bidder stochastic bidding process has low regret (and budget violation) in expectation, then the $s$-bidder fluid bidding process will also inherit these properties. We first bound the parameters using the constants from the proof of \Cref{lemma: fluid-hardness},
\begin{align*}
\rho_i
&=\frac{B_i}{2m}\geq\frac{b}{3ks}, \quad 
1+\max_{i\in\mathcal I}\muupper_i
=\frac1{\min_{i\in\mathcal I}\rho_i}\leq\frac{3k}{b}s,\\
\frac{\int v_i\,\dd F'(\vec v)}{\rho_i}
&=(1-\eta)\frac{\sum_{j=1}^m v_{ij}}{B_i}
+\eta\frac{m(K-1)}{sKB_i}
\leq(1-\eta)\frac{3k}{b}+\eta\frac{3k}{b}
=\frac{3k}{b}.
\end{align*}
To apply \Cref{lem: regularizer-grid-error,lem: finite-dual-curvature},
we express $F'$ as a mixture of distributions with regularizer weight $\eta/2$:
\begin{equation*}
F'=\frac{1-\eta}{2}F_1+\frac12\delta_{\vec0}+\frac\eta2F_{0,K}
=\left(1-\frac\eta2\right)
\frac{(1-\eta)F_1+\delta_{\vec0}}{2-\eta}+\frac\eta2F_{0,K}.
\end{equation*}
Substituting $s$ for the number of bidders and $\eta/2$ for the
regularizer weight in \Cref{lem: regularizer-grid-error,lem: finite-dual-curvature}, we have
\begin{equation*}
\widetilde\sigma\geq\frac{\eta}{12(2+3k/b)(3k/b)^2s^3},
\qquad
\xi_K\leq\frac{2\cdot (1+3k/b)}{K}.
\end{equation*}
Since $3k/b\geq1$ and $\eta\in(0,1)$, we have
$H\geq\max\{3k/b,2\eta(1+3k/b)\}$. We can then check that the following parameter bounds are all constants,
\begin{equation}
\label{eq: padded-instance-scales}
\begin{aligned}
\min_{i\in\mathcal I}\rho_i\geq\frac1{Hs}, \ \ 
\max_{i\in\mathcal I}\frac{\int v_i\, \ \ \dd F'(\vec v)}{\rho_i}\leq H, \  \ 
\max_{i\in\mathcal I}\muupper_i\leq Hs, \  \ 
\widetilde\sigma\geq\frac1{Hs^3}, \ \ 
\xi_K\leq\frac{2H}{\eta K}.
\end{aligned}
\end{equation}

We now show bounds on the expected regret and budget violation of
the fluid process under $F'$.
Since $n^\varepsilon=s^{r\varepsilon}\geq s^4$, by our assumption for contradiction, we have for all $i\in \mathcal{I}$, 
\begin{equation*}
    \mathbb{E}_{\vec v_t \sim \widehat F}[\regret_{i, \mathcal{A}} (\bm v) ] \leq \frac{T}{Cs^4} , \ \ \text{and} \ \ 
    \mathbb{E}_{\vec v_t \sim \widehat F} \left[\left(\sumt c_{i,t} - T\rho_i\right)_+\right] \leq  \frac{T}{Cs^4}. 
\end{equation*}

In the stochastic setting, we first apply
\Cref{lem: uniform-fluid-stochastic-transfer} to the $n$-bidder
process under $\widehat F$ and then restrict to $\mathcal I$ similarly.
For a bidder $i\in\mathcal I$, either fluid expectation is at most its
stochastic counterpart plus
$\sqrt{2T\log(2(\lceil\muupper_iT\rceil+1))}+\sqrt T+4$.
We have $T=\lceil q(s^r)\rceil=\Omega(s^9)$ and $\log T=O(\log s)$,
because $r$ and $q$ are fixed. Therefore, 
by \eqref{eq: padded-instance-scales},
\begin{align*}
&\frac{\sqrt{2T\log(2(\lceil\muupper_iT\rceil+1))}
+\sqrt T+4}{T/(Cs^4)}\\
&\quad\leq Cs^4\left[
\sqrt{\frac{2\log(2(HsT+2))}{T}}+\frac1{\sqrt T}+\frac4T
\right]\\
&\quad=O\left(s^{-1/2}\sqrt{\log s}+s^{-5}\right)
\longrightarrow0.
\end{align*}
We take $s$ sufficiently large that $n\geq N_0$, $T\geq Cs^4$.
The fluid process under $F'$ therefore satisfies
\begin{equation}
\label{eq: padded-fluid-premise}
\E[\regret_i^{F'}]\leq\frac{2T}{Cs^4},\qquad
\E\left[\left(\sumt c_{i,t}^{F'}-T\rho_i\right)_+\right]
\leq\frac{2T}{Cs^4}
\quad(i\in\mathcal I).
\end{equation}

\paragraph{Extracting SPPE from the fluid process.}
We next extract an $s$-bidder fluid SPPE from the fluid bidding process by checking the conditions of \eqref{eq: extraction-budget-hypothesis}--\eqref{eq: extraction-slack-hypothesis}, with probability at least $3/4$.
The following bound on the optimal fluid dual solution
$\mu_i^{*,F'}\in[0,\muupper_i]$ will be useful. By its optimality, 
\begin{align*}
T\rho_i\mu_i^{*,F'}
\leq\sumt D_{i,t}^{F'}(\mu_i^{*,F'};\vec\alpha_{-i,t})
\leq\sumt D_{i,t}^{F'}(0;\vec\alpha_{-i,t})
\leq T\int v_i\,\dd F'(\vec v)
\leq T\rho_i H.
\end{align*}
Thus $\mu_i^{*,F'}\leq\int v_i\,\dd F'(\vec v)/\rho_i\leq H$.

By the deterministic regret decomposition in \Cref{lem: regret-to-dual-regret},
\begin{equation}
\label{eq: padded-fluid-regret-identity}
\begin{aligned}
\regret_i^{F'}
={}&\sumt\bigl[
D_{i,t}^{F'}(\mu_i^{*,F'};\vec\alpha_{-i,t})
-D_{i,t}^{F'}(\mu_{i,t};\vec\alpha_{-i,t})\\
&\hspace{27mm}
-(\rho_i-c_{i,t}^{F'})(\mu_i^{*,F'}-\mu_{i,t})
\bigr]+\mu_i^{*,F'}
\left(T\rho_i-\sumt c_{i,t}^{F'}\right)
-\Delta_i^{F'},
\end{aligned}
\end{equation}
where $0\leq\Delta_i^{F'}\leq1$.
To bound the right-hand side as a whole, for notational simplicity we define 
\begin{equation*}
Y_i=\regret_i^{F'}
+(1+H)\left(\sumt c_{i,t}^{F'}-T\rho_i\right)_+
+1+2\eta\xi_KT.
\end{equation*}
Using $\mu_i^{*,F'}\leq H$ and $\Delta_i^{F'}\leq1$, we obtain
\begin{align*}
Y_i-\left(\sumt c_{i,t}^{F'}-T\rho_i\right)_+
&=\regret_i^{F'}
+H\left(\sumt c_{i,t}^{F'}-T\rho_i\right)_+
+1+2\eta\xi_KT\\
&\geq\regret_i^{F'}+\mu_i^{*,F'}
\left(\sumt c_{i,t}^{F'}-T\rho_i\right)
+\Delta_i^{F'}+2\eta\xi_KT\\
&\overset{\text{(a)}}{\geq}\frac{\widetilde\sigma}{4}
\sumt(\mu_{i,t}-\mu_i^{*,F'})^2\geq0.
\end{align*}
In step (a), we use \eqref{eq: padded-fluid-regret-identity} and apply
\Cref{lem: finite-dual-curvature} with $\mu_1=\mu_{i,t}$,
$\mu_2=\mu_i^{*,F'}$, and
$g=\rho_i-c_{i,t}^{F'}\in
\partial D_{i,t}^{F'}(\mu_{i,t};\vec\alpha_{-i,t})$.
In particular, $Y_i\geq0$.
By convexity, the bracketed terms in
\eqref{eq: padded-fluid-regret-identity} are nonnegative, so
\begin{align*}
\mu_i^{*,F'}\left(T\rho_i-\sumt c_{i,t}^{F'}\right)
&\leq\regret_i^{F'}+\Delta_i^{F'}
\leq\regret_i^{F'}+1
\leq Y_i.
\end{align*}

We choose $\delta_3=24(1+H)/(Cs^3)$,
$\delta_1=\delta_3/\min_{i\in\mathcal I}\rho_i$, and
$\delta_2=4\delta_3/\widetilde\sigma$.
We bound the expectation of $\sum_{i\in\mathcal I}Y_i$ as follows:
\begin{align*}
\frac1T\E\left[\sum_{i\in\mathcal I}Y_i\right]
&\overset{\text{(a)}}{\leq}s\left[(2+H)\frac{2}{Cs^4}
+\frac1T+2\eta\xi_K\right]\\
&\overset{\text{(b)}}{\leq}s\left[\frac{2(2+H)}{Cs^4}
+\frac1{Cs^4}+\frac{4H}{Cs^4}\right]=\frac{6H+5}{Cs^3}\leq\frac{\delta_3}{4},
\end{align*}
where step (a) is by \eqref{eq: padded-fluid-premise}, and step (b) is by \eqref{eq: padded-instance-scales}.
Since $Y_i\geq0$, we have by Markov's inequality that
$\Pr(\sum_{i\in\mathcal I}Y_i>T\delta_3)\leq1/4$.
On the event $\sum_{i\in\mathcal I}Y_i\leq T\delta_3$, we have
\begin{align*}
\frac1T\sumt(c_{i,t}^{F'}-\rho_i)
&\leq\frac1T\left(\sumt c_{i,t}^{F'}-T\rho_i\right)_+
\leq\frac{Y_i}{T}\leq\delta_3\leq\delta_1\rho_i,\\
\frac1T\sum_{i\in\mathcal I}\sumt(\mu_{i,t}-\mu_i^{*,F'})^2
&\leq\frac{4}{T\widetilde\sigma}\sum_{i\in\mathcal I}Y_i
\leq\frac{4\delta_3}{\widetilde\sigma}=\delta_2,\\
\mu_i^{*,F'}\left(\rho_i-\frac1T\sumt c_{i,t}^{F'}\right)
&\leq\frac{Y_i}{T}\leq\delta_3.
\end{align*}
Thus \eqref{eq: extraction-budget-hypothesis}--%
\eqref{eq: extraction-slack-hypothesis} hold with $\mu_i=\mu_i^{*,F'}$.
By \eqref{eq: padded-instance-scales} and
$C\geq864H(1+H)^2/\gamma^3$, we check \Cref{lem: low-regret-implies-approximate-SPPE}'s requirement on $\delta_1, \delta_2, \delta_3$ as follows,
\begin{align*}
\delta_1
&\leq Hs\delta_3=\frac{24H(1+H)}{Cs^2}
\leq\frac{\gamma^3}{36(1+H)s^2}
\leq\frac{\gamma^3}{36},\\
\delta_2
&\leq4Hs^3\delta_3=\frac{96H(1+H)}{C}
\leq\frac{\gamma^3}{9(1+H)}
\leq\frac{\gamma^3}
{9\left(1+\max_{i\in\mathcal I}\frac{\int v_i\,\dd F'(\vec v)}{\rho_i}\right)},\\
\delta_3
&=\frac{24(1+H)}{Cs^3}
\leq\frac{\gamma^3}{36H(1+H)s^3}
\leq\frac{\gamma^3}{36Hs}
\leq\frac{\gamma^3\min_{i\in\mathcal I}\rho_i}{36}.
\end{align*}

By \Cref{lem: low-regret-implies-approximate-SPPE}, we obtain a
$\gamma$-approximate fluid SPPE for $F'$ with probability at least $3/4$,
and hence for $\widetilde F$.
Therefore, we get a polynomial-time randomized algorithm for the
$\mathsf{PPAD}$-hard problem, contradicting
the assumption of \Cref{thm: fixed-polynomial-regret}. For details of the computation models and encoding, see \Cref{app: complexity-model}.
\end{proof}

\begin{samepage}
\subsection{Computational Details}
\label{app: complexity-model}

We show that the reduction in the proof of \Cref{thm: fixed-polynomial-regret} can be implemented as a polynomial-time randomized algorithm that succeeds with probability at least $2/3$. By the item bound in \Cref{lemma: fluid-hardness} and $K=O(s^4)$, we have $|\supp(\widehat F)|\leq m+\binom s2K^2+1=O(s^{10})$, and all values, probability masses, and per-round budgets are rational with $O(\log s)$ bits. Since $r$ is fixed, the instance with $n=s^r$ bidders has polynomial encoding length.

We compute fluid utilities and payments by summing over the finite support, where in each round of \Cref{alg: stochastic-to-fluid-simulation}, the allocations at all support points use the same random bits of $\mathcal A$. Since $\mathcal A$ runs in polynomial time and $T=\lceil q(s^r)\rceil$ is polynomial in $s$, the simulation and the extraction in \Cref{lem: low-regret-implies-approximate-SPPE} take polynomial time. Since the probability masses are rational, each $\vec v_t$ can be sampled in polynomial time from a distribution within total variation distance $1/(24T)$ of $\widehat F$. By a coupling argument and the union bound over the $T$ rounds, this decreases the success probability by at most $1/24$, so the algorithm succeeds with probability at least $3/4-1/24>2/3$. The finitely many small source instances are handled by a hardcoded table.
\end{samepage}

\section{Missing Results and Proofs of \Cref{section: main-section-on-ABE}}
\label{app: ROI-proofs}

\subsection{Proof of \Cref{lem: ROI-strong-convexity}}

\roicurvature*

\begin{proof}
Similar to the proof of \Cref{lem: strong-convexity-F0}, we condition on the pair of bidders with nonzero values in the construction of $F_0$. We choose $j\in\operatorname*{arg\,min}_{\ell\ne i}\alpha_{\ell,t}$, and consider the event where $\{i,j\}$ is selected, which has probability $2/[n(n-1)]$. Conditional on this event, $v_i$ and $v_j$ are independent and uniform on $[0,\vmax]$, and $\max_{k\ne i}\alpha_{k,t}v_k=\alpha_{j,t}v_j$. Therefore, the conditional expectation of the integrand in $D_{i,t}^{F_0}$ is
\begin{align*}
&\frac1{\vmax^2}\int_0^{\vmax}\int_0^{\vmax}
\bigl((1+\mu)y-\mu\alpha_{j,t}z\bigr)_+\,\dd y\,\dd z
=
\begin{cases}
\displaystyle\vmax\left[\frac{1+\mu}{2}-\frac{\mu\alpha_{j,t}}2
+\frac{\mu^2\alpha_{j,t}^2}{6(1+\mu)}\right],
&\mu\alpha_{j,t}\leq1+\mu,\\[3mm]
\displaystyle\frac{\vmax(1+\mu)^2}{6\mu\alpha_{j,t}},
&\mu\alpha_{j,t}\geq1+\mu,
\end{cases}
\end{align*}
where we first integrate over $y\geq\mu\alpha_{j,t}z/(1+\mu)$, and then over $0\leq z\leq\min\{\vmax,\vmax(1+\mu)/(\mu\alpha_{j,t})\}$. The two expressions and their first two derivatives agree at the boundary $\mu\alpha_{j,t}=1+\mu$, and their second derivatives are $\vmax\alpha_{j,t}^2/[3(1+\mu)^3]$ and $\vmax/(3\alpha_{j,t}\mu^3)$, respectively. Since $\alpha_{j,t}>\alphalower_j=1$, both second derivatives are lower-bounded by $\vmax/[3\alpha_{j,t}(H+1)^3]$ for $\mu\in[0,H]$. This proves $\sigma_0$-strong convexity of the dual loss on $[0,H]$.
\end{proof}

\subsection{Proof of \Cref{lem: low-regret-implies-approximate-ABE}}

\begingroup
\roifullstatementtrue
\roiextraction*
\endgroup

\begin{proof}
We construct the pacing multipliers $(\alpha_i')_{i\in[n]}$ and allocations $(\hat x_i'(\vec v))_{i\in[n]}$ of the approximate fluid autobidding equilibrium, and then show the four properties in \Cref{def: approximate-fluid-ABE}. 

Different from the previous budget-constrained setting, we construct the pacing multipliers directly by $\alpha_i'=1+1/\mu_i$; we will discuss later that $\mu_i$ can be computed efficiently. 
Since $\mu_i\geq1/(M-1)$, we have $\alpha_i'\in(1,M]$.

We next construct the allocations. For a parameter $\iota>0$ to be chosen later, we consider the set of timesteps where the $\ell_{\infty}$ distance deviation of the dual is greater than $\iota$, i.e., $\mathcal T=\{t\in[T]:\max_i|\mu_{i,t}-\mu_i|>\iota\}$. Similar to \eqref{eq: extraction-exceptional-rounds}, by \eqref{eq: ROI-extraction-variance-hypothesis} we have
\begin{equation*}
\frac{|\mathcal T|}{T}
\leq\frac1{T\iota^2}\sum_{t\in\mathcal T}\max_i(\mu_{i,t}-\mu_i)^2
\leq\frac1{T\iota^2}\sumi\sumt(\mu_{i,t}-\mu_i)^2
\leq\frac{\delta_2}{\iota^2}.
\end{equation*}
As in the proof of \Cref{lem: low-regret-implies-approximate-SPPE}, we define the allocation
\begin{equation*}
\hat x_{i,t}(\vec v)=
\begin{cases}
x_{i,t}(\vec v), & t\notin\mathcal T, \vec v \in \supp(F)\\
\indi{i=\min\operatorname*{arg\,max}_{j\in[n]}\alpha_j'v_j},
& \text{otherwise},
\end{cases}
\end{equation*}
and construct the offline fluid allocation by taking the average over the adjusted online allocations, $$\hat x_i'(\vec v)=T^{-1}\sumt\hat x_{i,t}(\vec v).$$

For $t\notin\mathcal T$, we have $|\mu_{i,t}-\mu_i|\leq\iota$. Since $\mu_i,\mu_{i,t}\geq1/(M-1)$, we have
\begin{align*}
\frac{\alpha_{i,t}}{\alpha_i'}
&=1+\frac{\mu_i-\mu_{i,t}}{\mu_{i,t}(1+\mu_i)}
\leq1+\iota(M-1),\\
\frac{\alpha_i'}{\alpha_{i,t}}
&=1+\frac{\mu_{i,t}-\mu_i}{\mu_i(1+\mu_{i,t})}
\leq1+\iota(M-1).
\end{align*}
Similar to \eqref{eq: relative-price-comparison}, we have for $t\notin\mathcal T$,
\begin{equation}
\label{eq: ROI-relative-price-comparison}
\begin{aligned}
\frac{h(\vec v;\vec\alpha')}{1+\iota(M-1)}
&\leq h(\vec v;\vec\alpha_t)
\leq[1+\iota(M-1)]h(\vec v;\vec\alpha'),\\
\frac{p(\vec v;\vec\alpha')}{1+\iota(M-1)}
&\leq p(\vec v;\vec\alpha_t)
\leq[1+\iota(M-1)]p(\vec v;\vec\alpha').
\end{aligned}
\end{equation}

It is useful to bound the value in the timesteps in $\mathcal T$. Notice that \eqref{eq: ROI-extraction-payment-hypothesis} implicitly requires $\sumt u_{i,t}^F>0$ for all $i\in[n]$. For notational simplicity, we denote $\varepsilon_i=\delta_2T\int v_i\,\dd F(\vec v)/(\iota^2\sumt u_{i,t}^F)$. Then
\begin{equation}
\label{eq: ROI-extraction-exceptional-value}
\sum_{t\in\mathcal T}u_{i,t}^F
\leq|\mathcal T|\int v_i\,\dd F(\vec v)
\leq\varepsilon_i\sumt u_{i,t}^F,
\end{equation}
and the same bound holds for the value $R_i=\sum_{t\in\mathcal T}\int\hat x_{i,t}(\vec v)v_i\,\dd F(\vec v)$ that bidder $i$ receives from the adjusted allocations in $\mathcal T$. Notice that a winning bidder pays at most its bid, so $c_{i,t}^F\leq Mu_{i,t}^F$, and the payment of bidder $i$ under $\vec\alpha'$ for the adjusted allocations in $\mathcal T$ is at most $MR_i$.

Now we show that $(\vec\alpha',(\hat x_i'(\vec v))_{i\in[n]})$ satisfies conditions (a)--(d) in the definition of approximate fluid autobidding equilibrium (\Cref{def: approximate-fluid-ABE}).

\begin{itemize}
\item[(a)] \textit{Allocation to the near-highest bidder.}
For $\vec v\in\supp(F)$, $t\notin\mathcal T$, and $x_{i,t}(\vec v)>0$, we have
\begin{align*}
\alpha_i'v_i
&\geq\frac{\alpha_{i,t}v_i}{1+\iota(M-1)}
=\frac{h(\vec v;\vec\alpha_t)}{1+\iota(M-1)}
\geq\frac{h(\vec v;\vec\alpha')}{[1+\iota(M-1)]^2}.
\end{align*}
For $t\in\mathcal T$ and for $\vec v\notin\supp(F)$, we allocate only to highest bidders under $\vec\alpha'$. Thus $\hat x_i'(\vec v)>0$ implies $\alpha_i'v_i\geq h(\vec v;\vec\alpha')/[1+\iota(M-1)]^2$. We then have condition (a) whenever we have $[1+\iota(M-1)]^{-2}\geq1-\gamma$.

\item[(b)] \textit{Market clearance.} This is obviously satisfied by the construction itself.

\item[(c)] \textit{Approximate ROI feasibility.}
To upper-bound the payment-to-value ratio, we compare the payment under $\vec\alpha'$ with the online payment as follows,
\begin{equation}
\label{eq: ROI-extracted-feasibility-bound}
\begin{aligned}
\frac{\int\hat x_i'(\vec v)p(\vec v;\vec\alpha')\,\dd F(\vec v)}
{\int\hat x_i'(\vec v)v_i\,\dd F(\vec v)}
&\overset{\text{(a)}}{\leq}\frac{[1+\iota(M-1)]\sum_{t\notin\mathcal T}c_{i,t}^F+MR_i}
{\sum_{t\notin\mathcal T}u_{i,t}^F+R_i}\\
&\overset{\text{(b)}}{\leq}\frac{[1+\iota(M-1)](1+\delta_1)\sumt u_{i,t}^F+MR_i}
{(1-\varepsilon_i)\sumt u_{i,t}^F+R_i}\\
&\overset{\text{(c)}}{\leq}\max\left\{\frac{[1+\iota(M-1)](1+\delta_1)}{1-\varepsilon_i},\ [1+\iota(M-1)](1+\delta_1)+M\varepsilon_i\right\},
\end{aligned}
\end{equation}
where step (a) is by \eqref{eq: ROI-relative-price-comparison} and the payment bound for the adjusted allocations, step (b) invokes condition \eqref{eq: ROI-extraction-payment-hypothesis} with $\sum_{t\notin\mathcal T}c_{i,t}^F\leq\sumt c_{i,t}^F$ and \eqref{eq: ROI-extraction-exceptional-value}, and step (c) holds since the middle expression is monotone in $R_i\in[0,\varepsilon_i\sumt u_{i,t}^F]$; we will choose the parameters such that $\varepsilon_i<1$. Therefore, condition (c) holds if the right-hand side of \eqref{eq: ROI-extracted-feasibility-bound} is at most $1+\gamma$.

\item[(d)] \textit{Approximate complementarity.}
We consider a bidder with $\alpha_i'<M$. Then $\mu_i>1/(M-1)$ and $(1+\mu_i)/\mu_i=\alpha_i'<M$, so \eqref{eq: ROI-extraction-slack-hypothesis} implies $\sumt c_{i,t}^F\geq(1-M\delta_3)\sumt u_{i,t}^F$. To lower-bound the payment-to-value ratio, we have
\begin{equation}
\label{eq: ROI-extracted-complementarity-bound}
\begin{aligned}
\frac{\int\hat x_i'(\vec v)p(\vec v;\vec\alpha')\,\dd F(\vec v)}
{\int\hat x_i'(\vec v)v_i\,\dd F(\vec v)}
&\overset{\text{(a)}}{\geq}\frac{\sum_{t\notin\mathcal T}c_{i,t}^F}
{[1+\iota(M-1)](\sumt u_{i,t}^F+R_i)}
\overset{\text{(b)}}{\geq}\frac{\sumt c_{i,t}^F-M\sum_{t\in\mathcal T}u_{i,t}^F}
{[1+\iota(M-1)](1+\varepsilon_i)\sumt u_{i,t}^F}\\
&\overset{\text{(c)}}{\geq}\frac{1-M\delta_3-M\varepsilon_i}
{[1+\iota(M-1)](1+\varepsilon_i)},
\end{aligned}
\end{equation}
where in (a) we use \eqref{eq: ROI-relative-price-comparison} and $\sum_{t\notin\mathcal T}u_{i,t}^F\leq\sumt u_{i,t}^F$, in (b) we use $c_{i,t}^F\leq Mu_{i,t}^F$ and $R_i\leq\varepsilon_i\sumt u_{i,t}^F$, and in (c) we use the above lower bound on $\sumt c_{i,t}^F$ and \eqref{eq: ROI-extraction-exceptional-value}. Therefore, condition (d) holds by contradiction if the right-hand side of \eqref{eq: ROI-extracted-complementarity-bound} is at least $1-\gamma$.
\end{itemize}

\paragraph{Fixing the parameters.}

We now take $\iota=\gamma/[2(M-1)]$ and use the bounds on $\delta_1,\delta_2,\delta_3$ in the lemma statement, and show that this satisfies all inequality requirements for the above conditions (a)--(d). The bound on $\delta_2$ gives
\begin{equation*}
\varepsilon_i
\leq\frac{\gamma^3}{32M(M-1)^2\iota^2}
=\frac{\gamma}{8M}<1.
\end{equation*}
For condition (a), we apply Bernoulli's inequality to obtain $[1+\iota(M-1)]^{-2}\geq1-2\iota(M-1)=1-\gamma$. For condition (c), we bound the two terms in \eqref{eq: ROI-extracted-feasibility-bound} as follows,
\begin{align*}
\frac{[1+\iota(M-1)](1+\delta_1)}{1-\varepsilon_i}
&\leq\frac{(1+\gamma/2)(1+\gamma/8)}{1-\gamma/8}
=1+\gamma-\frac{\gamma/4-3\gamma^2/16}{1-\gamma/8}
\leq1+\gamma,\\
[1+\iota(M-1)](1+\delta_1)+M\varepsilon_i
&\leq\left(1+\frac\gamma2\right)\left(1+\frac\gamma8\right)+\frac\gamma8
=1+\frac{3\gamma}4+\frac{\gamma^2}{16}
\leq1+\gamma,
\end{align*}
where the first line uses $\gamma/4-3\gamma^2/16>0$. Finally, for condition (d), we have
\begin{align*}
\frac{1-M\delta_3-M\varepsilon_i}
{[1+\iota(M-1)](1+\varepsilon_i)}
&\geq\frac{1-\gamma/4}{(1+\gamma/2)(1+\gamma/8)}
\geq1-\gamma,
\end{align*}
where the last inequality uses $(1-\gamma)(1+\gamma/2)(1+\gamma/8)=1-3\gamma/8-9\gamma^2/16-\gamma^3/16\leq1-\gamma/4$.
We conclude that $(\vec\alpha',(\hat x_i')_{i\in[n]})$ satisfies all four conditions of a $\gamma$-approximate fluid autobidding equilibrium with cap $M$.
\end{proof}

\subsection{\Cref{lemma: roi-fluid-hardness} and Its Proof}
\label{app: ROI-hardness}

\begin{lemma}
\label{lemma: roi-fluid-hardness}
There exist constants $\eta_0,\gamma\in(0,1)$ such that, for any fixed $M\geq10$ and $\eta,\eta'>0$ with $\eta+\eta'\leq\eta_0$ and $6\eta\leq(1-\gamma)\eta'$, computing a $\gamma$-approximate fluid autobidding equilibrium with cap $M$ is \ppad-hard for the family of distributions $\{(1-\eta-\eta')F_1+\eta F_0+\eta'F_0'\}$.
\end{lemma}

\begin{proof}
We use the approximate (discrete) autobidding equilibrium from \citet[Theorem~2 and Appendix~G]{li2024vulnerabilities}. For $\lambda\in(0,1)$, a $\lambda$-approximate (discrete) autobidding equilibrium for $m$ items with value vectors $\vec v_j=(v_{1j},\ldots,v_{nj})$, $j\in[m]$, is a tuple $(\vec\alpha,(x_{ij})_{i\in[n],j\in[m]})$ that satisfies conditions (a)--(d) in \Cref{def: approximate-fluid-ABE}, with $\gamma$ replaced by $\lambda$ and the integrals over $F$ replaced by sums over the $m$ items. Computing such an equilibrium is \ppad-hard for any sufficiently small constant $\lambda>0$.

We fix such a $\lambda$ and consider an instance of this discrete problem, with $m$ items with value vectors $(\vec v_1,\ldots,\vec v_m)$. We construct the fluid distribution as $F=(1-\eta-\eta')F_1+\eta F_0+\eta'F_0'$, where $F_1=\frac1m\sum_{j=1}^m\delta_{\vec v_j}$ is the empirical distribution on the $m$ items.

We next show that, given a $\gamma$-approximate fluid autobidding equilibrium for $F$, by restricting to the $m$ items one can extract a $\lambda$-approximate (discrete) autobidding equilibrium, whenever each bidder $i$ receives value $\sum_{j=1}^m x_{ij}v_{ij}\geq b$ and has multiplier $\alpha_i\leq A$ in the restriction, for some $b>0$ and $A\in[1,M)$ such that
\begin{equation}
\label{eq: roi-mixture-hardness-margins}
\gamma+\frac{m(\eta+\eta')(1+\gamma)\vmax}{(1-\eta-\eta')nb}\leq\lambda,
\qquad A\eta\leq(1-\gamma)\eta'.
\end{equation}
Concretely, given a $\gamma$-approximate fluid autobidding equilibrium $(\vec\alpha',(x_i'(\vec v))_{i\in[n]})$, we construct the discrete solution as $\alpha_i=\alpha_i'$ and $x_{ij}=x_i'(\vec v_j)$. On each value point $\vec v_j$, $j\in[m]$, both solutions share the same bids from all bidders and the same allocation.

We bound the contribution of the two regularizers as follows. Notice that under $F_0$, at most two bidders have positive bids, so a bidder who receives the item pays at most its own bid $\alpha_i'v_i$. Under $F_0'$, bidder $i$ is the only bidder with positive value with probability $1/n$, in which case it receives the item at price zero. We then have
\begin{equation}
\label{eq: roi-regularizer-value-payment}
\begin{aligned}
\int x_i'(\vec v)v_i\,\dd F_0(\vec v)
&\leq\frac{\vmax}{n},
&\int x_i'(\vec v)p(\vec v;\vec\alpha')\,\dd F_0(\vec v)
&\leq\frac{\alpha_i'\vmax}{n},\\
\int x_i'(\vec v)v_i\,\dd F_0'(\vec v)
&=\frac{\vmax}{n},
&\int x_i'(\vec v)p(\vec v;\vec\alpha')\,\dd F_0'(\vec v)
&=0.
\end{aligned}
\end{equation}
Similar to \eqref{eq: mixture-spending-identity}, we decompose the spending and the value of each bidder $i\in[n]$ in the fluid market as follows,
\begin{equation}
\label{eq: roi-mixture-identities}
\begin{aligned}
\int x_i'(\vec v)p(\vec v;\vec\alpha')\,\dd F(\vec v)
&=\frac{1-\eta-\eta'}{m}\sum_{j=1}^m x_{ij}p(\vec v_j;\vec\alpha)
+\eta\int x_i'(\vec v)p(\vec v;\vec\alpha')\,\dd F_0(\vec v),\\
\int x_i'(\vec v)v_i\,\dd F(\vec v)
&=\frac{1-\eta-\eta'}{m}\sum_{j=1}^m x_{ij}v_{ij}
+\eta\int x_i'(\vec v)v_i\,\dd F_0(\vec v)
+\frac{\eta'\vmax}{n}.
\end{aligned}
\end{equation}

We next show that $(\vec\alpha,(x_{ij})_{i\in[n],j\in[m]})$ is indeed a $\lambda$-approximate autobidding equilibrium for the (discrete) market.

\begin{itemize}
\item[(a)] \textit{Allocation to the near-highest bidder.} By \eqref{eq: roi-mixture-hardness-margins}, we have $\gamma\leq\lambda$, and the rest is the same as in the proof of \Cref{lemma: fluid-hardness}.
\item[(b)] \textit{Market clearance.} Obvious.
\item[(c)] \textit{Approximate ROI feasibility.} By the approximate ROI feasibility of the fluid solution, \eqref{eq: roi-regularizer-value-payment}, and \eqref{eq: roi-mixture-identities}, we have for all bidders,
\begin{equation}
\label{eq: roi-discrete-upper-transfer}
\sum_{j=1}^m x_{ij}p(\vec v_j;\vec\alpha)
\leq(1+\gamma)\sum_{j=1}^m x_{ij}v_{ij}
+\frac{m(\eta+\eta')(1+\gamma)\vmax}{(1-\eta-\eta')n}.
\end{equation}
Since $\sum_jx_{ij}v_{ij}\geq b$, the right-hand side is at most $(1+\lambda)\sum_jx_{ij}v_{ij}$ by the first inequality in \eqref{eq: roi-mixture-hardness-margins}.
\item[(d)] \textit{Approximate complementarity.} We show that $\sum_jx_{ij}p(\vec v_j;\vec\alpha)\geq(1-\lambda)\sum_jx_{ij}v_{ij}$ for all bidders, so condition (d) holds. For a bidder with $\alpha_i'<M$, the approximate complementarity of the fluid solution gives $\int x_i'(\vec v)p(\vec v;\vec\alpha')\,\dd F(\vec v)\geq(1-\gamma)\int x_i'(\vec v)v_i\,\dd F(\vec v)$. If we further have $\alpha_i'\eta\leq(1-\gamma)\eta'$, by \eqref{eq: roi-regularizer-value-payment}, the spending on $F_0$ satisfies
\begin{equation*}
\eta\int x_i'(\vec v)p(\vec v;\vec\alpha')\,\dd F_0(\vec v)
\leq\frac{\alpha_i'\eta\vmax}{n}
\leq(1-\gamma)\frac{\eta'\vmax}{n}
\leq(1-\gamma)\left(\eta\int x_i'(\vec v)v_i\,\dd F_0(\vec v)+\frac{\eta'\vmax}{n}\right).
\end{equation*}
Subtracting this from the fluid inequality by \eqref{eq: roi-mixture-identities}, we obtain
\begin{equation}
\label{eq: roi-discrete-lower-transfer}
\sum_{j=1}^m x_{ij}p(\vec v_j;\vec\alpha)
\geq(1-\gamma)\sum_{j=1}^m x_{ij}v_{ij}.
\end{equation}
Since $\alpha_i\leq A<M$ for all bidders, the second inequality in \eqref{eq: roi-mixture-hardness-margins} ensures \eqref{eq: roi-discrete-lower-transfer} for all bidders, and condition (d) follows from $\gamma\leq\lambda$.
\end{itemize}

Thus the constructed solution is a $\lambda$-approximate discrete autobidding equilibrium. The reduction construction can clearly be done in polynomially many operations.

It only remains for us to argue that $b$ and $A$ can be chosen as constants for the hard instances of \citet{li2024vulnerabilities}. Different from the budget setting, the received values and the multipliers are not directly bounded by the source instance, and we derive their bounds from the conditions of the fluid equilibrium.

\paragraph{Useful bounds for the hard family.} We provide parameters about the (discrete) hard market instances, which will be useful in the proof of \Cref{thm: fixed-polynomial-ROI-regret}. We refer the readers to \citet[Appendix~G]{li2024vulnerabilities} for details of the hardness construction.

We follow the notation of \citet[Appendix~G]{li2024vulnerabilities}, except that we write $\theta$ for their construction parameter $\eta$, to avoid confusion with our regularizer weight. This parameter can be taken sufficiently small, and we take $\theta\leq1/100$. In the hard instances, each vertex bidder values three incoming items at $1/3$ each, its lower-bound item at $(1-\theta)/2$, and its upper-bound item at $1/[4(1-\theta)]$. It also values at most three outgoing items at $(1-\theta)/14$ each, and a missing incoming edge leaves its item valued only by the head bidder. The incoming items have reserve price $(1-\theta)/7$, and the two bound items have reserve price $1$. Each reserve price $r$ on an item $j_0$ is replaced by two auxiliary bidders $i_1,i_2$ and two auxiliary items $j_1,j_2$, with the following values (all unspecified values are zero),
\begin{center}
\begin{tabular}{c|ccc}
 & $j_1$ & $j_2$ & $j_0$\\ \hline
$i_1$ & $0$ & $2r$ & $r$\\
$i_2$ & $r/2$ & $r$ & $r/2$
\end{tabular}
\end{center}
Therefore, each vertex corresponds to $11$ bidders and $15$ items, so $m/n=15/11$, and $\vmax=2$. Moreover, all reserve prices are at least $1/8$, each bidder has total value $\sum_jv_{ij}\leq3$, and each item has at most four positive values, which are all at least $(1-\theta)/14\geq1/15$. For a vertex bidder, we have the sharper bound
\begin{equation*}
\sum_jv_{ij}
\leq1+\frac{3(1-\theta)}{14}+\frac{1-\theta}{2}+\frac1{4(1-\theta)}
=\frac{55}{28}-\theta\left(\frac57-\frac1{4(1-\theta)}\right)
\leq\frac{55}{28}.
\end{equation*}

\paragraph{Bounding the received values and multipliers.}
We take $\lambda\leq1/100$, fix $\gamma=\eta_0=\lambda/100$, and consider $\eta,\eta'>0$ with $\eta+\eta'\leq\eta_0$ and $6\eta\leq(1-\gamma)\eta'$. We will show that every bidder receives value at least $b=1/16$ and has multiplier at most $A=6$ in the restriction.

We first simplify \eqref{eq: roi-discrete-upper-transfer} for the hard family. Since $m/n=15/11$, $\vmax=2$, $\sum_jv_{ij}\leq3$, and $\eta+\eta'\leq\gamma\leq10^{-4}$, we have for all bidders,
\begin{equation}
\label{eq: roi-hard-family-upper-transfer}
\sum_jx_{ij}p(\vec v_j;\vec\alpha)
\leq(1+\gamma)\sum_jx_{ij}v_{ij}
+\frac{30(\eta+\eta')(1+\gamma)}{11(1-\eta-\eta')}
\leq(1+\gamma)\sum_jx_{ij}v_{ij}+3(\eta+\eta')
\leq\sum_jx_{ij}v_{ij}+6\gamma.
\end{equation}
Moreover, for a bidder with $\alpha_i\leq6$, we have $\alpha_i<M$ and $\alpha_i\eta\leq(1-\gamma)\eta'$, so \eqref{eq: roi-discrete-lower-transfer} also applies.

We next consider the auxiliary bidders $i_1,i_2$ for a reserve price $r$ on item $j_0$. Since bidder $i_2$ is the only bidder that values $j_1$, it receives all of $j_1$ at price zero, and thus receives value at least $r/2$. Since all multipliers are at least one, we have $p(\vec v_{j_0};\vec\alpha)\geq r/2$, and if $x_{i_2j_2}>0$, condition (a) gives $p(\vec v_{j_2};\vec\alpha)\geq2r(1-\gamma)$. Applying \eqref{eq: roi-hard-family-upper-transfer} to bidder $i_2$, we get
\begin{equation*}
(1-2\gamma)rx_{i_2j_2}-\frac r2
\leq\sum_jx_{i_2j}\bigl(p(\vec v_j;\vec\alpha)-v_{i_2j}\bigr)
\leq6\gamma,
\end{equation*}
and thus $x_{i_1j_2}=1-x_{i_2j_2}\geq1-(1/2+6\gamma/r)/(1-2\gamma)\geq1/4$. Therefore, both auxiliary bidders receive value at least $r/2\geq1/16$, and \eqref{eq: roi-hard-family-upper-transfer} gives for $i\in\{i_1,i_2\}$,
\begin{equation}
\label{eq: roi-auxiliary-relative-feasibility}
\sum_jx_{ij}p(\vec v_j;\vec\alpha)
\leq(1+\gamma+48(\eta+\eta'))\sum_jx_{ij}v_{ij}
\leq(1+\lambda)\sum_jx_{ij}v_{ij}.
\end{equation}
We use \eqref{eq: roi-auxiliary-relative-feasibility} to bound the auxiliary multipliers. On $j_0$, both auxiliary bids are half of their bids on $j_2$, so $p(\vec v_{j_0};\vec\alpha)\geq p(\vec v_{j_2};\vec\alpha)/2$. If $p(\vec v_{j_2};\vec\alpha)>2r(1+\lambda)$, the prices of both items valued by $i_1$ exceed $1+\lambda$ times its values, which contradicts \eqref{eq: roi-auxiliary-relative-feasibility} since $x_{i_1j_2}>0$. Hence $p(\vec v_{j_2};\vec\alpha)\leq2r(1+\lambda)$. If $\alpha_{i_2}<2(1-\gamma)\alpha_{i_1}$, then by condition (a), bidder $i_2$ only receives $j_1$ at price zero, and $\alpha_{i_2}=p(\vec v_{j_2};\vec\alpha)/r\leq2(1+\lambda)<6$, which contradicts \eqref{eq: roi-discrete-lower-transfer}. Together with condition (a) for $x_{i_1j_2}>0$, the two bids on $j_2$ differ by a factor of at most $1/(1-\gamma)$, so
\begin{equation*}
\max\{\alpha_{i_1},\alpha_{i_2}/2\}
\leq\frac{p(\vec v_{j_2};\vec\alpha)}{2r(1-\gamma)}
\leq\frac{1+\lambda}{1-\gamma}
\leq\frac{21}{20}.
\end{equation*}
In particular, both auxiliary multipliers are below $6$.

Finally, we consider a vertex bidder $i$. If $\alpha_i>9/2$, its bids on both bound items are at least $\alpha_i/4>9/8>21/[20(1-\gamma)]$, while the auxiliary bids on these items are at most $21/20$. By condition (a), it receives both bound items in full, and pays at least $1$ for each of them since $\alpha_{i_1}\geq1$. Since its value is at most $55/28$ and $2>55/28+6\gamma$, this contradicts \eqref{eq: roi-hard-family-upper-transfer}. Thus all vertex multipliers are at most $9/2$. On an incoming item, the head then bids at least $1/3$, the tail bids at most $9/28$, and the auxiliary bids are at most $3/20$, where the latter two bounds are below $(1-\gamma)/3$. By conditions (a) and (b), the head receives all three of its incoming items, with total value $1$.

We have shown that every bidder receives value at least $b=1/16$ and has multiplier at most $A=6<M$. We can then check \eqref{eq: roi-mixture-hardness-margins} for the whole family as follows,
\begin{equation*}
\gamma+\frac{m(\eta+\eta')(1+\gamma)\vmax}{(1-\eta-\eta')nb}
\leq\gamma+48(\eta+\eta')
\leq\frac{49\lambda}{100}<\lambda,
\qquad
A\eta=6\eta\leq(1-\gamma)\eta'.
\end{equation*}
\end{proof}

\subsection{\Cref{lem: ROI-fluid-stochastic-transfer} and Its Proof}
\label{app: roi-stochastic-transfer}
Since their pacing multipliers are unbounded, the payments of value maximizers can be arbitrarily large. We therefore truncate the prices at a level $P\geq\vmax$, and define the scaled value and payment as
\begin{equation*}
\widehat u_{i,t}=u_{i,t}\min\{1,P/p_t\},\qquad
\widehat c_{i,t}=c_{i,t}\min\{1,P/p_t\},
\end{equation*}
where the scaling factor is one when $p_t=0$. In other words, the scaled outcome allocates $x_{i,t}\min\{1,P/p_t\}\in[0,1]$ of the item to bidder $i$ at the same price. We define their fluid counterparts $\widehat u_{i,t}^F$ and $\widehat c_{i,t}^F$ by integrating the scaled outcomes over $F$, in the same way as $u_{i,t}^F$ and $c_{i,t}^F$.

\begin{lemma}
\label{lem: ROI-fluid-stochastic-transfer}
For any stochastic multi-agent bidding dynamic $\mathcal A$, we can efficiently obtain $\mathcal A^F$ for the fluid setting by simulating $\mathcal A$, and achieve
\begin{align*}
\E\left[\left|\regret_{i,\mathcal A^F}^F
                 -\regret_{i,\mathcal A}\right|\right]
&\leq\vmax\left[(1+H)\sqrt{2T\log\left(2(\lceil HT^2\rceil+1)\right)}
+\sqrt T+6+T\zeta\right],\\
\E\left[\left(\sumt(\widehat c_{i,t}^F-\widehat u_{i,t}^F)\right)_+\right]
&\leq\E\left[\left(\sumt(c_{i,t}-u_{i,t})\right)_+\right]
+(P+\vmax)\sqrt T,
\end{align*}
where $H\geq1$ and $\zeta\in[0,1]$ are such that the hindsight benchmarks $\opt_i(\bm v;\bm\alpha_{-i})$ and $\opt_i^F(\bm\alpha_{-i})$ both have optimal dual solutions in $[0,H]$ with probability at least $1-\zeta$.
\end{lemma}

\begin{proof}
We construct $\mathcal A^F$ by simulating $\mathcal A$ as in \Cref{alg: stochastic-to-fluid-simulation}, and let $\mathcal H_t$ be the history up to step $t$ before $\vec v_t$ is generated. Similar to the proof of \Cref{lem: uniform-fluid-stochastic-transfer}, we have
\begin{equation*}
\E[u_{i,t}\mid\mathcal H_t]=u_{i,t}^F,\qquad
\E[\widehat c_{i,t}-\widehat u_{i,t}\mid\mathcal H_t]=\widehat c_{i,t}^F-\widehat u_{i,t}^F.
\end{equation*}
Since $u_{i,t}\in[0,\vmax]$ and $\widehat c_{i,t}-\widehat u_{i,t}\in[-\vmax,P]$, the same argument as \eqref{eq: utility-payment-martingale} gives
\begin{equation}
\label{eq: ROI-utility-payment-martingale}
\E\left[\left|\sumt(u_{i,t}-u_{i,t}^F)\right|\right]
\leq\vmax\sqrt T,\quad
\E\left[\left|\sumt
\bigl(\widehat c_{i,t}-\widehat u_{i,t}-\widehat c_{i,t}^F+\widehat u_{i,t}^F\bigr)\right|\right]
\leq(P+\vmax)\sqrt T.
\end{equation}
Since $P\geq\vmax$, we have $c_{i,t}-u_{i,t}\geq0$ whenever $p_t>P$, so $\widehat c_{i,t}-\widehat u_{i,t}=(c_{i,t}-u_{i,t})\min\{1,P/p_t\}\leq c_{i,t}-u_{i,t}$. We then have
\begin{align*}
\E\left[\left(\sumt(\widehat c_{i,t}^F-\widehat u_{i,t}^F)\right)_+\right]
&\leq\E\left[\left(\sumt(\widehat c_{i,t}-\widehat u_{i,t})\right)_+\right]
+(P+\vmax)\sqrt T\\
&\leq\E\left[\left(\sumt(c_{i,t}-u_{i,t})\right)_+\right]
+(P+\vmax)\sqrt T,
\end{align*}
which proves the second inequality in \Cref{lem: ROI-fluid-stochastic-transfer}.

We next bound the difference between the two hindsight benchmarks through their duals. Since the dual losses $D_{i,t}(\mu;\vec v_t,\vec\alpha_{-i,t})=((1+\mu)v_{i,t}-\mu d_{i,t})_+$ are not Lipschitz near $\mu=0$, here we compare the two dual objectives only on $[1/T,H]$. 
On the event where both dual objectives have minimizers in $[0,H]$, rounding a minimizer in $[0,1/T)$ up to $1/T$ increases either cumulative dual objective by at most $\vmax$, i.e.,
\begin{equation*}
0\leq\min_{\mu\in[1/T,H]}\sumt D_{i,t}(\mu;\vec v_t,\vec\alpha_{-i,t})
-\min_{\mu\geq0}\sumt D_{i,t}(\mu;\vec v_t,\vec\alpha_{-i,t})
\leq\vmax,
\end{equation*}
and the same bound holds for the fluid dual objective.

We define $Z_{i,t}(\mu)$ as in the proof of \Cref{lem: uniform-fluid-stochastic-transfer}. For $\mu\in[1/T,H]$, both dual losses lie in $[0,(1+H)\vmax]$, so $|Z_{i,t}(\mu)-Z_{i,t-1}(\mu)|\leq(1+H)\vmax$. Moreover, both dual losses are $T\vmax$-Lipschitz on $[1/T,H]$.

We take the grid discretization $\mathcal G=\{k/T^2:k=T,\ldots,\lfloor HT^2\rfloor\}\cup\{H\}$. 
Similar to \eqref{eq: explicit-uniform-dual-transfer}, applying the conditional Hoeffding bound with $s=\sqrt{2\log(2|\mathcal G|)/[T(1+H)^2\vmax^2]}$, and using the fact that $Z_{i,T}$ is $2T^2\vmax$-Lipschitz on $[1/T,H]$, we obtain
\begin{equation*}
\E\left[\sup_{\mu\in[1/T,H]}|Z_{i,T}(\mu)|\right]
\leq(1+H)\vmax\sqrt{2T\log\left(2(\lceil HT^2\rceil+1)\right)}
+2\vmax.
\end{equation*}

On the event where both dual objectives have minimizers in $[0,H]$, we compare the minima of the two dual objectives over $[1/T,H]$ using $\sup_{\mu\in[1/T,H]}|Z_{i,T}(\mu)|$, and then combine with the rounding error and duality gap. 
Outside this event, both benchmarks are in $[0,T\vmax]$, so their difference is at most $T\vmax$. Therefore,
\begin{equation*}
\E\left[\left|\opt_i^F(\bm\alpha_{-i})-\opt_i(\bm v;\bm\alpha_{-i})\right|\right]
\leq\E\left[\sup_{\mu\in[1/T,H]}|Z_{i,T}(\mu)|\right]
+4\vmax+T\vmax\zeta.
\end{equation*}
Combining with \eqref{eq: ROI-utility-payment-martingale},
\begin{align*}
\E\left[\left|\regret_{i,\mathcal A^F}^F-\regret_{i,\mathcal A}\right|\right]
&\leq\E\left[\left|\opt_i^F(\bm\alpha_{-i})-\opt_i(\bm v;\bm\alpha_{-i})\right|\right]
+\E\left[\left|\sumt(u_{i,t}^F-u_{i,t})\right|\right]\\
&\leq\vmax\left[(1+H)\sqrt{2T\log\left(2(\lceil HT^2\rceil+1)\right)}
+\sqrt T+6+T\zeta\right].
\end{align*}
This proves the first inequality in \Cref{lem: ROI-fluid-stochastic-transfer}.
\end{proof}

\subsection{Discretization}
\label{app: ROI-discretization}

We will use the discretized regularizer $F_{0,K}$ in \Cref{def: F0-distribution-discretized}, and consider $\widetilde F=(1-\eta-\eta')F_1+\eta F_{0,K}+\eta'F_0'$. We remark that \Cref{lemma: roi-fluid-hardness} also applies when the regularizer is $F_{0,K}$ in place of $F_0$.

We show that the discretized regularizer also provides strong-convexity-like properties in the following \Cref{lem: ROI-discrete-curvature}. Since the dual iterates can be unbounded, we state it with the scaled outcomes in \Cref{app: roi-stochastic-transfer}, and project the dual iterate $\mu_{i,t}$ onto $[0,H+1]$. Without the scaling, the left-hand side is exactly $D_{i,t}^{\widetilde F}(\mu)-D_{i,t}^{\widetilde F}(\mu_{i,t})-g(\mu-\mu_{i,t})$ with the subgradient $g=u_{i,t}^{\widetilde F}-c_{i,t}^{\widetilde F}$, as in \Cref{lem: finite-dual-curvature}.

\begin{lemma}
\label{lem: ROI-discrete-curvature}
For any $H>0$, $L>1$, $P\geq L\vmax$, and $\mu\in[0,H]$, if $\alpha_{j,t}\leq L$ for some $j\neq i$, it holds that
\begin{equation*}
D_{i,t}^{\widetilde F}(\mu;\vec\alpha_{-i,t})
-(1+\mu)\widehat u_{i,t}^{\widetilde F}
+\mu\widehat c_{i,t}^{\widetilde F}
\geq\frac{\widetilde\sigma}{4}\bigl(\min\{\mu_{i,t},H+1\}-\mu\bigr)^2
-\frac{8\eta\vmax[H+2+(H+1)L]}{n(n-1)K},
\end{equation*}
where $\widetilde\sigma=\dfrac{2\eta\vmax}{3n(n-1)L(H+2)^3}$.
\end{lemma}

\begin{proof}
We condition on the event that $\{i,j\}$ is the pair of bidders with nonzero values under the component $F_{0,K}$, which has probability $2\eta/[n(n-1)]$ under $\widetilde F$. Let $f$ and $f_K$ be the conditional dual losses under the continuous and the grid distributions, respectively. By a rounding argument, we have for $\mu'\in[0,H+1]$,
\begin{equation*}
|f(\mu')-f_K(\mu')|
\leq\frac{\vmax}{K}\bigl[(1+\mu')+\mu'\alpha_{j,t}\bigr]
\leq\frac{\vmax}{K}\bigl[H+2+(H+1)L\bigr].
\end{equation*}
By the computation in the proof of \Cref{lem: ROI-strong-convexity}, $f$ is $\vmax/[3L(H+2)^3]$-strongly convex on $[0,H+1]$. Meanwhile, conditional on this event, the scaled outcomes of bidder $i$ equal the original outcomes, since its payment is at most $\alpha_{j,t}\vmax\leq L\vmax\leq P$.

We first consider the case $\mu_{i,t}\leq H+1$. Let $\phi(\mu')$ be the conditional expectation of $((1+\mu')v_i-\mu'\alpha_{j,t}v_j)x_{i,t}(\vec v)$ under the grid distribution, which is affine in $\mu'$. By \Cref{prop: multipliers}, the online allocation attains $\phi(\mu_{i,t})=f_K(\mu_{i,t})$, and we always have $\phi\leq f_K$. Therefore,
\begin{align*}
f_K(\mu)-\phi(\mu)
&=f_K(\mu)+f_K(\mu_{i,t})-2\phi\left(\frac{\mu+\mu_{i,t}}2\right)
\geq f_K(\mu)+f_K(\mu_{i,t})-2f_K\left(\frac{\mu+\mu_{i,t}}2\right)\\
&\geq f(\mu)+f(\mu_{i,t})-2f\left(\frac{\mu+\mu_{i,t}}2\right)
-\frac{4\vmax[H+2+(H+1)L]}{K}\\
&\geq\frac{\vmax}{12L(H+2)^3}(\mu_{i,t}-\mu)^2
-\frac{4\vmax[H+2+(H+1)L]}{K}.
\end{align*}
If $\mu_{i,t}>H+1$, bidder $i$ bids lower than under the dual variable $H+1$, so the online allocation is dominated by the optimal allocation at $H+1$. Since $\mu\leq H$, the coefficient $(1+\mu)v_i-\mu\alpha_{j,t}v_j$ is nonnegative at every value vector where the optimal allocation at $H+1$ gives the item to bidder $i$, because there $(H+2)v_i\geq(H+1)\alpha_{j,t}v_j$. Thus replacing the online allocation with the optimal allocation at $H+1$ only decreases $f_K(\mu)-\phi(\mu)$, and the above calculation applies with $H+1$ in place of $\mu_{i,t}$.

\end{proof}

\subsection{Proof of \Cref{thm: fixed-polynomial-ROI-regret}}

\roiimpossibility*

\begin{proof}
In this proof, we use the hard instances and constants $\lambda,\eta_0,\gamma$ from the proof of \Cref{lemma: roi-fluid-hardness}, where $\vmax=2$. We will choose constants $\tau>0$, $M\geq10$, $\eta,\eta'>0$ with $\eta+\eta'\leq\eta_0$, and $c\in(0,1]$ in this order, where $\tau$, $\eta/\eta'$, and $c$ are sufficiently small and $M$ is sufficiently large; we specify these requirements where they are used. For notational simplicity, we denote $H=20/\eta'$ and $C=32(H+3)/c$; these parameters are constants independent of $n$, $T$, and $\varepsilon$. For the sake of contradiction, suppose both expectations in the theorem statement are at most $T/(Cn^\varepsilon)$ for all bidders on every instance with $n\geq N_0$ and $T=\lceil q(n)\rceil$.

\paragraph{Construction of the fluid bidding instance.}
Consider the hard instance from the proof of \Cref{lemma: roi-fluid-hardness} with $s$ bidders and $m=15s/11$ items. For $r=\max\{14,\lceil4/\varepsilon\rceil\}$, we construct a fluid instance with $n=s^r$ bidders by adding $n-s$ auxiliary bidders who form a disjoint market component. Specifically, let $\mathcal I=\{1,\ldots,s\}$ be the set of original bidders from the hard instance, and $\widetilde F$ be defined as in \Cref{app: ROI-discretization} with $K=2\lceil s/c\rceil$. To construct the joint value distribution $\widehat F$ of $n$ bidders, we draw $(\widetilde v_i)_{i\in\mathcal I}\sim\widetilde F$ and independently choose one of the following two cases,
\begin{equation*}
v_i=\begin{cases}
\widetilde v_i, & i\in\mathcal I,\\
0, & i\notin\mathcal I,
\end{cases}
\quad\text{w.p. }1/2,
\qquad
v_i=\begin{cases}
0, & i\in\mathcal I,\\
1, & i\notin\mathcal I,
\end{cases}
\quad\text{w.p. }1/2.
\end{equation*}

As in the proof of \Cref{thm: fixed-polynomial-regret}, restricting to the $s$ bidders in $\mathcal I$ gives a fluid bidding process under $F'=\frac12\widetilde F+\frac12\delta_{\vec0}$, with the same values, payments, and hindsight benchmarks for bidders in $\mathcal I$. Since $\vec0$ contributes zero value and payment, using the same pacing multipliers and allocations under $\widetilde F$ exactly doubles all these quantities. Consequently, for $i\in\mathcal I$,
\begin{equation*}
\regret_i^{\widetilde F}=2\regret_i^{\widehat F},\qquad
\left(\sumt(\widehat c_{i,t}^{\widetilde F}-\widehat u_{i,t}^{\widetilde F})\right)_+
=2\left(\sumt(\widehat c_{i,t}^{\widehat F}-\widehat u_{i,t}^{\widehat F})\right)_+.
\end{equation*}

\paragraph{Low regret of the $s$-bidder fluid process.}
We show that if the $n$-bidder bidding process has low regret and ROI violation in expectation, then the $s$-bidder fluid bidding process under $\widetilde F$ will also inherit these properties. We first bound the parameters of $\widetilde F$. By the useful bounds for the hard family in the proof of \Cref{lemma: roi-fluid-hardness}, we have for all $i\in\mathcal I$,
\begin{equation*}
\int v_i\,\dd\widetilde F(\vec v)
\leq\frac3m+\frac{(\eta+\eta')\vmax}{s}
\leq\frac5s.
\end{equation*}
The following bound on the optimal fluid dual solution $\mu_i^{*,\widetilde F}$ will be useful. Since $D_{i,t}^{\widetilde F}(\mu;\vec\alpha_{-i,t})\geq\eta'\vmax\mu/s$ from the component $F_0'$, by its optimality,
\begin{equation*}
\frac{T\eta'\vmax\mu_i^{*,\widetilde F}}{s}
\leq\sumt D_{i,t}^{\widetilde F}(\mu_i^{*,\widetilde F};\vec\alpha_{-i,t})
\leq\sumt D_{i,t}^{\widetilde F}(0;\vec\alpha_{-i,t})
=T\int v_i\,\dd\widetilde F(\vec v)
\leq\frac{5T}{s}.
\end{equation*}
Thus $\mu_i^{*,\widetilde F}\leq5/(\eta'\vmax)<H$. We use the scaled outcomes in \Cref{app: roi-stochastic-transfer}, with $P=\lceil80\vmax s^2/c\rceil$.

We now show bounds on the expected regret and scaled ROI violation of the fluid process under $\widetilde F$. Since $n^\varepsilon=s^{r\varepsilon}\geq s^4$, by our assumption for contradiction, we have for all $i\in\mathcal I$,
\begin{equation*}
\E[\regret_{i,\mathcal A}(\bm v)]\leq\frac{T}{Cs^4},\ \ \text{and}\ \
\E\left[\left(\sumt(c_{i,t}-u_{i,t})\right)_+\right]\leq\frac{T}{Cs^4}.
\end{equation*}
In the fluid setting, the same bounds hold for the $n$-bidder fluid process under $\widehat F$, and the scaled ROI violation is also at most $T/(Cs^4)$ since $\widehat c_{i,t}-\widehat u_{i,t}\leq c_{i,t}-u_{i,t}$.

In the stochastic setting, we first apply \Cref{lem: ROI-fluid-stochastic-transfer} to the $n$-bidder process under $\widehat F$ and then restrict to $\mathcal I$ similarly. Different from the budget setting, we need the optimal dual solutions of both benchmarks to lie in $[0,H]$. For the fluid benchmark, this holds since the fluid dual objective under $\widehat F$ is half of that under $\widetilde F$ and $\mu_i^{*,\widetilde F}<H$. For the stochastic benchmark, bidder $i$ has value $\vmax$ and faces zero competing bids with probability $\eta'/(2s)$ in each round, from the component $F_0'$. The total value from these arrivals has mean $\eta'\vmax T/(2s)$, and $\E[\sumt v_{i,t}]\leq5T/(2s)$. By the Chernoff bounds,
\begin{align*}
\Pr\left(\text{total value of bidder $i$ from $F_0'$}<\frac{\eta'\vmax T}{8s}\right)
&\leq\exp\left(-\frac{\eta'T}{32s}\right),\\
\Pr\left(\sumt v_{i,t}>\frac{5T}{s}\right)
&\leq\exp\left(-\frac{5T}{6\vmax s}\right)
\leq\exp\left(-\frac{\eta'T}{32s}\right).
\end{align*}
Outside these two events, any optimal dual solution $\mu$ of the stochastic benchmark satisfies
\begin{equation*}
\frac{\mu\eta'\vmax T}{8s}
\leq\sumt D_{i,t}(\mu;\vec v_t,\vec\alpha_{-i,t})
\leq\sumt D_{i,t}(0;\vec v_t,\vec\alpha_{-i,t})
=\sumt v_{i,t}
\leq\frac{5T}{s},
\end{equation*}
so $\mu\leq40/(\eta'\vmax)=H$. We then apply \Cref{lem: ROI-fluid-stochastic-transfer} with $\zeta=\min\{1,2\exp(-\eta'T/(32s))\}$. For a bidder $i\in\mathcal I$, the expected fluid regret and scaled ROI violation are at most their stochastic counterparts plus $\vmax[(1+H)\sqrt{2T\log(2(\lceil HT^2\rceil+1))}+\sqrt T+6+T\zeta]$ and $(P+\vmax)\sqrt T$, respectively. We have $T=\lceil q(s^r)\rceil=\Omega(s^{14})$, $\log T=O(\log s)$, and $P=O(s^2)$, because $r$, $q$, and $c$ are fixed. Therefore,
\begin{align*}
&\frac{\vmax\left[(1+H)\sqrt{2T\log(2(\lceil HT^2\rceil+1))}+\sqrt T+6+T\zeta\right]}{T/(Cs^4)}
=O\left(s^{-3}\sqrt{\log s}+s^4e^{-\Omega(s^{13})}\right)\longrightarrow0,\\
&\frac{(P+\vmax)\sqrt T}{T/(Cs^4)}=O\left(s^{-1}\right)\longrightarrow0.
\end{align*}
We take $s$ sufficiently large that $n\geq N_0$, $T\geq Cs^4$, and both ratios are at most one. In both settings, bidders in $\mathcal I$ then have expected regret and scaled ROI violation at most $2T/(Cs^4)$ under $\widehat F$. Doubling these bounds, the fluid process under $\widetilde F$ satisfies
\begin{equation}
\label{eq: ROI-final-fluid-bounds}
\E[\regret_i^{\widetilde F}]\leq\frac{4T}{Cs^4},\qquad
\E\left[\left(\sumt
(\widehat c_{i,t}^{\widetilde F}-\widehat u_{i,t}^{\widetilde F})
\right)_+\right]\leq\frac{4T}{Cs^4}
\quad(i\in\mathcal I).
\end{equation}

\paragraph{Extracting autobidding equilibrium from the fluid process.}
We next extract an $s$-bidder approximate fluid autobidding equilibrium from the fluid bidding process, with probability at least $3/4$. Similar to \eqref{eq: padded-fluid-regret-identity}, we decompose the regret, but with the scaled outcomes since only the scaled ROI violation is controlled,
\begin{equation}
\label{eq: ROI-fluid-regret-identity}
\begin{aligned}
\regret_i^{\widetilde F}+\Delta_i^{\widetilde F}
+\sumt(u_{i,t}^{\widetilde F}-\widehat u_{i,t}^{\widetilde F})
={}&\sumt\left[D_{i,t}^{\widetilde F}(\mu_i^{*,\widetilde F};\vec\alpha_{-i,t})
-(1+\mu_i^{*,\widetilde F})\widehat u_{i,t}^{\widetilde F}
+\mu_i^{*,\widetilde F}\widehat c_{i,t}^{\widetilde F}\right]\\
&+\mu_i^{*,\widetilde F}\sumt
(\widehat u_{i,t}^{\widetilde F}-\widehat c_{i,t}^{\widetilde F}).
\end{aligned}
\end{equation}
where $0\leq\Delta_i^{\widetilde F}\leq\vmax$ is the duality gap, and both sides equal $\sumt D_{i,t}^{\widetilde F}(\mu_i^{*,\widetilde F};\vec\alpha_{-i,t})-\sumt\widehat u_{i,t}^{\widetilde F}$. The bracketed terms are nonnegative, since the scaled allocation belongs to $[0,1]$ and pays the competing price per unit. Moreover, $u_{i,t}=\widehat u_{i,t}$ unless $p_t>P$, in which case $\widehat c_{i,t}=Px_{i,t}\geq(P/\vmax)u_{i,t}$. Thus $u_{i,t}-\widehat u_{i,t}\leq(\vmax/P)\widehat c_{i,t}$, and integrating this inequality gives
\begin{equation*}
\sumt(u_{i,t}^{\widetilde F}-\widehat u_{i,t}^{\widetilde F})
\leq\frac{\vmax}{P}\sumt\widehat c_{i,t}^{\widetilde F}
\leq\frac{5\vmax T}{sP}
+\frac{\vmax}{P}\left(\sumt
(\widehat c_{i,t}^{\widetilde F}-\widehat u_{i,t}^{\widetilde F})\right)_+,
\end{equation*}
where we use $\sumt\widehat u_{i,t}^{\widetilde F}\leq T\int v_i\,\dd\widetilde F(\vec v)\leq5T/s$. To bound the right-hand side of \eqref{eq: ROI-fluid-regret-identity} as a whole, for notational simplicity we define
\begin{equation*}
Y_i=\regret_i^{\widetilde F}
+(H+2)\left(\sumt
(\widehat c_{i,t}^{\widetilde F}-\widehat u_{i,t}^{\widetilde F})\right)_+
+\vmax+\frac{5\vmax T}{sP}.
\end{equation*}
Using $\mu_i^{*,\widetilde F}\leq H$, $\Delta_i^{\widetilde F}\leq\vmax$, and $P\geq\vmax$, we obtain
\begin{align*}
Y_i-\left(\sumt(\widehat c_{i,t}^{\widetilde F}
-\widehat u_{i,t}^{\widetilde F})\right)_+
&\geq\regret_i^{\widetilde F}+\Delta_i^{\widetilde F}
+\sumt(u_{i,t}^{\widetilde F}-\widehat u_{i,t}^{\widetilde F})
+\mu_i^{*,\widetilde F}\sumt(\widehat c_{i,t}^{\widetilde F}-\widehat u_{i,t}^{\widetilde F})\\
&\overset{\text{(a)}}{=}\sumt\left[D_{i,t}^{\widetilde F}(\mu_i^{*,\widetilde F};\vec\alpha_{-i,t})
-(1+\mu_i^{*,\widetilde F})\widehat u_{i,t}^{\widetilde F}
+\mu_i^{*,\widetilde F}\widehat c_{i,t}^{\widetilde F}\right]
\geq0,
\end{align*}
where step (a) is by \eqref{eq: ROI-fluid-regret-identity}. In particular, $Y_i\geq0$. Since the bracketed terms in \eqref{eq: ROI-fluid-regret-identity} are nonnegative, we also have
\begin{align*}
\mu_i^{*,\widetilde F}\sumt(\widehat u_{i,t}^{\widetilde F}
-\widehat c_{i,t}^{\widetilde F})
&\leq\regret_i^{\widetilde F}+\Delta_i^{\widetilde F}
+\sumt(u_{i,t}^{\widetilde F}-\widehat u_{i,t}^{\widetilde F})
\leq Y_i.
\end{align*}

We choose $\delta=c/s^2$, and bound the expectation of $\sum_{i\in\mathcal I}Y_i$ as follows,
\begin{align*}
\frac1T\E\left[\sum_{i\in\mathcal I}Y_i\right]
&\overset{\text{(a)}}{\leq}s\left[(H+3)\frac{4}{Cs^4}+\frac{\vmax}{T}+\frac{5\vmax}{sP}\right]
\overset{\text{(b)}}{\leq}\frac{c}{8s^3}+\frac{c}{16s^3}+\frac{c}{16s^2}
\leq\frac{\delta}{4},
\end{align*}
where step (a) is by \eqref{eq: ROI-final-fluid-bounds}, and step (b) is by $C=32(H+3)/c$, $T\geq Cs^4$, $P\geq80\vmax s^2/c$, and $\vmax=2$. Since $Y_i\geq0$, we have by Markov's inequality that $\Pr(\sum_{i\in\mathcal I}Y_i>T\delta)\leq1/4$. On the event $\sum_{i\in\mathcal I}Y_i\leq T\delta$, we have
\begin{equation}
\label{eq: ROI-preparation-bounds}
\begin{aligned}
\frac1T\sumt(\widehat c_{i,t}^{\widetilde F}-\widehat u_{i,t}^{\widetilde F})
&\leq\delta,\quad i\in\mathcal I,\\
\frac1T\sum_{i\in\mathcal I}\sumt\left[D_{i,t}^{\widetilde F}(\mu_i^{*,\widetilde F};\vec\alpha_{-i,t})
-(1+\mu_i^{*,\widetilde F})\widehat u_{i,t}^{\widetilde F}
+\mu_i^{*,\widetilde F}\widehat c_{i,t}^{\widetilde F}\right]
&\leq\delta,\\
\frac{\mu_i^{*,\widetilde F}}{T}\sumt(\widehat u_{i,t}^{\widetilde F}-\widehat c_{i,t}^{\widetilde F})
&\leq\delta,\quad i\in\mathcal I.
\end{aligned}
\end{equation}

Since the pacing multipliers may exceed the cap $M$, \eqref{eq: ROI-preparation-bounds} does not directly give \eqref{eq: ROI-extraction-payment-hypothesis}--\eqref{eq: ROI-extraction-slack-hypothesis}, i.e., the conditions we require to apply \Cref{lem: low-regret-implies-approximate-ABE}. To this end, we first discard a small fraction of exceptional rounds, then cap the pacing multipliers at $M$ in the remaining rounds, and finally check the conditions of \Cref{lem: low-regret-implies-approximate-ABE}.

\paragraph{Exceptional rounds.}
For notational simplicity, we let $\beta=\gamma\eta'\vmax/(320M)$, $L=32(\vmax+1)/(\eta\vmax\beta)$, and $\kappa=\gamma/(M-1)$. We consider the set of exceptional rounds $\mathcal T=\mathcal T_1\cup\mathcal T_2$, where
\begin{equation*}
\mathcal T_1=\left\{t\in[T]:\bigl|\{j\in\mathcal I:\alpha_{j,t}\leq L\}\bigr|<2\right\},
\qquad
\mathcal T_2=\left\{t\notin\mathcal T_1:\max_{i\in\mathcal I}\bigl|\min\{\mu_{i,t},H+1\}-\mu_i^{*,\widetilde F}\bigr|>\kappa\right\}.
\end{equation*}

\begin{claim}
\label{clm: ROI-exceptional-rounds}
We have $|\mathcal T|\leq\beta T$, and $|\mu_{i,t}-\mu_i^{*,\widetilde F}|\leq\kappa$ for all $i\in\mathcal I$ and $t\notin\mathcal T$.
\end{claim}

\begin{proof}[Proof of the claim]
We first bound $|\mathcal T_1|$. For $t\in\mathcal T_1$, with probability at least $1/8$, the selected pair of $F_{0,K}$ has two multipliers exceeding $L$ and two grid values at least $\vmax/2$. Since $P\geq L\vmax$ for sufficiently large $s$, the total scaled payment in such a round is at least $\eta L\vmax/16$. On the other hand, by the first inequality in \eqref{eq: ROI-preparation-bounds},
\begin{equation*}
\frac1T\sum_{i\in\mathcal I}\sumt\widehat c_{i,t}^{\widetilde F}
\leq\frac1T\sum_{i\in\mathcal I}\sumt\widehat u_{i,t}^{\widetilde F}+s\delta
\leq\vmax+1.
\end{equation*}
Thus $|\mathcal T_1|/T\leq16(\vmax+1)/(\eta L\vmax)=\beta/2$.

We next bound $|\mathcal T_2|$. For $t\notin\mathcal T_1$, every bidder has an opponent with multiplier at most $L$. We sum \Cref{lem: ROI-discrete-curvature} with $\mu=\mu_i^{*,\widetilde F}\in[0,H]$ over $i\in\mathcal I$ and $t\notin\mathcal T_1$, and use the second inequality in \eqref{eq: ROI-preparation-bounds} and $K\geq s/c$. Requiring
\begin{equation*}
c\leq\frac{\beta\eta\vmax\kappa^2}{24L(H+2)^3},\qquad
c\leq\frac{\beta\kappa^2}{192L(H+2)^3[H+2+(H+1)L]},
\end{equation*}
we obtain
\begin{align}
\frac1T\sum_{t\notin\mathcal T_1}\sum_{i\in\mathcal I}
\bigl(\min\{\mu_{i,t},H+1\}-\mu_i^{*,\widetilde F}\bigr)^2
&\leq\frac{4}{\widetilde\sigma}\left[\delta+\frac{8\eta\vmax[H+2+(H+1)L]}{(s-1)K}\right]\nonumber\\
&=\frac{6c(s-1)L(H+2)^3}{\eta\vmax s}
+\frac{48sL(H+2)^3[H+2+(H+1)L]}{K}\nonumber\\
&\leq\frac{\beta\kappa^2}{4}+\frac{\beta\kappa^2}{4}
=\frac{\beta\kappa^2}{2},
\label{eq: ROI-projected-dual-deviation}
\end{align}
where $\widetilde\sigma=2\eta\vmax/[3s(s-1)L(H+2)^3]$. By Markov's inequality, $|\mathcal T_2|/T\leq\beta/2$.

Finally, for $t\notin\mathcal T$, we have $\mu_{i,t}\leq H+1$, as $\kappa<1$ and $\mu_i^{*,\widetilde F}\leq H$. Thus $|\mu_{i,t}-\mu_i^{*,\widetilde F}|\leq\kappa$.
\end{proof}

\paragraph{Capping the pacing multipliers.}
We show that capping the pacing multipliers at $M$ does not change the outcomes on the atoms of $F_1$ in the rounds $t\notin\mathcal T$. We first claim the following.

\begin{claim}
\label{clm: ROI-no-two-small-duals}
No atom of $F_1$ has two positive-value bidders $i$ with $\mu_i^{*,\widetilde F}\leq\tau$.
\end{claim}

Intuitively, two bidders with very high multipliers cannot compete for the same item, since both would then pay too much.

\begin{proof}[Proof of the claim]
Suppose otherwise. For $t\notin\mathcal T$, both bidders have $\mu_{i,t}\leq\tau+\kappa\leq2\tau$ by \Cref{clm: ROI-exceptional-rounds}, where we take $M$ large enough that $\kappa=\gamma/(M-1)\leq\tau/2$. As positive values in the hard instances are at least $1/15$, both bid at least $1/(30\tau)$. For sufficiently large $s$, we have $P\geq1/(30\tau)$, so the scaled second price on this atom is at least $1/(30\tau)$. The atom has mass $(1-\eta-\eta')/m$ and at most four positive-value bidders. Using $|\mathcal T|\leq\beta T$, $(1-\eta-\eta')(1-\beta)\geq1/2$, and the first inequality in \eqref{eq: ROI-preparation-bounds}, we obtain
\begin{equation*}
\frac{11}{900\tau s}=\frac1{60\tau m}
\leq\frac{(1-\eta-\eta')(1-\beta)}{30\tau m}
\leq\frac1T\sum_{i:v_i>0}\sumt\widehat c_{i,t}^{\widetilde F}
\leq4\left(\frac5s+\delta\right)
\leq\frac{24}s,
\end{equation*}
where the sum is over the positive-value bidders of the atom. This is impossible if we choose a sufficiently small constant $\tau$.
\end{proof}

For $t\notin\mathcal T$, we define the capped pacing multipliers $\alpha_{i,t}'=\min\{\alpha_{i,t},M\}$. Equivalently, $\alpha_{i,t}'=1+1/\mu_{i,t}'$, where $\mu_{i,t}'=\max\{\mu_{i,t},1/(M-1)\}$ is the projected dual variable. We also let $\mu_i'=\max\{\mu_i^{*,\widetilde F},1/(M-1)\}$.

\begin{claim}
\label{clm: ROI-capping}
For $t\notin\mathcal T$, the winner and the price on each atom of $F_1$ are the same under $\vec\alpha_t$ and $\vec\alpha_t'$.
\end{claim}

\begin{proof}[Proof of the claim]
If no positive bid on the atom is capped, the two profiles coincide on the atom. Otherwise, let $i$ be a bidder whose positive bid is capped. Then $\mu_{i,t}<1/(M-1)$, so $\mu_i^{*,\widetilde F}<1/(M-1)+\kappa$ by \Cref{clm: ROI-exceptional-rounds}, which is less than $\tau$ for sufficiently large $M$. By \Cref{clm: ROI-no-two-small-duals}, every other positive-value bidder $k$ has $\mu_k^{*,\widetilde F}>\tau$. Thus $\mu_{k,t}\geq\tau-\kappa\geq\tau/2$, and bidder $k$ bids at most $(1+2/\tau)\vmax$. The capped bid of bidder $i$ is at least $M/15$, which exceeds $(1+2/\tau)\vmax$ for sufficiently large $M$. Hence bidder $i$ is the unique winner under both profiles with the same price.
\end{proof}

\paragraph{Checking the extraction conditions.}
In each round $t\notin\mathcal T$, we use the pacing multipliers $\vec\alpha_t'$, and keep the original allocations on the atoms of $F_1$ and on $F_0'$. On all other value vectors, we allocate the item to a highest bidder under $\vec\alpha_t'$. Let $u_{i,t}'$ and $c_{i,t}'$ be the resulting fluid values and payments. By \Cref{clm: ROI-capping}, the prices on the atoms of $F_1$ are at most $M\vmax$, which is at most $P$ for sufficiently large $s$. Thus the outcomes on the atoms of $F_1$ and on $F_0'$ equal the original scaled outcomes. From $F_0'$, each bidder receives value $\eta'\vmax/s$ in each round, so
\begin{equation*}
\frac1T\sum_{t\notin\mathcal T}u_{i,t}'
\geq\left(1-\frac{|\mathcal T|}{T}\right)\frac{\eta'\vmax}{s}
\geq\frac{\eta'\vmax}{2s}.
\end{equation*}

We now check \eqref{eq: ROI-extraction-payment-hypothesis}--\eqref{eq: ROI-extraction-slack-hypothesis} for the modified history, with $T-|\mathcal T|$ rounds and $\mu_i=\mu_i'$. Its pacing multipliers are in $(1,M]$, and $\mu_i'\geq1/(M-1)$.

\begin{itemize}
\item \textit{ROI violation.} Compared with $\sumt(\widehat c_{i,t}^{\widetilde F}-\widehat u_{i,t}^{\widetilde F})$, discarding the rounds in $\mathcal T$ adds at most $\sum_{t\in\mathcal T}\widehat u_{i,t}^{\widetilde F}\leq5\beta T/s$. Changing the outcomes on $F_{0,K}$ adds at most $\eta(M+1)\vmax/s$ per round, as a new payment is at most $Mv_i$ and a removed scaled value is at most $v_i$. Requiring $\eta/\eta'\leq\gamma/[64(M+1)]$ and $c\leq\gamma\eta'\vmax/(128M)$, the first inequality in \eqref{eq: ROI-preparation-bounds} gives
\begin{align*}
\frac{\sum_{t\notin\mathcal T}(c_{i,t}'-u_{i,t}')}
{\sum_{t\notin\mathcal T}u_{i,t}'}
&\leq\frac{2s}{\eta'\vmax}\left[\delta+\frac{5\beta}s
+\frac{\eta(M+1)\vmax}s\right]
=\frac{2c}{s\eta'\vmax}+\frac{10\beta}{\eta'\vmax}+\frac{2\eta(M+1)}{\eta'}\\
&\leq\frac{\gamma}{64M}+\frac{\gamma}{32M}+\frac{\gamma}{32}
\leq\frac\gamma8.
\end{align*}

\item \textit{Dual deviation.} By \Cref{clm: ROI-exceptional-rounds}, for $t\notin\mathcal T$ we have $|\mu_{i,t}'-\mu_i'|\leq|\mu_{i,t}-\mu_i^{*,\widetilde F}|=|\min\{\mu_{i,t},H+1\}-\mu_i^{*,\widetilde F}|$. Since $|\mathcal T|\leq\beta T\leq T/4$, \eqref{eq: ROI-projected-dual-deviation} gives
\begin{equation*}
\frac1{T-|\mathcal T|}
\sum_{i\in\mathcal I}\sum_{t\notin\mathcal T}(\mu_{i,t}'-\mu_i')^2
\leq\frac{T}{T-|\mathcal T|}\cdot\frac{\beta\kappa^2}{2}
\leq\beta\kappa^2.
\end{equation*}

\item \textit{ROI slack.} Consider a bidder with $\mu_i'>1/(M-1)$, so that $\mu_i'=\mu_i^{*,\widetilde F}$. Since the dual loss is nonnegative, at each value vector we have
\begin{equation*}
\frac{\mu_i'}{1+\mu_i'}\widehat c_{i,t}
\leq\widehat u_{i,t}
+\frac{D_{i,t}(\mu_i';\vec v,\vec\alpha_{-i,t})
-(1+\mu_i')\widehat u_{i,t}+\mu_i'\widehat c_{i,t}}{1+\mu_i'}.
\end{equation*}
We apply this inequality to the outcomes removed from the original process, i.e., those in the rounds of $\mathcal T$ and on $F_{0,K}$. These outcomes have total scaled value at most $5\beta T/s+\eta\vmax T/s$. Their total of $D_{i,t}-(1+\mu_i')\widehat u_{i,t}+\mu_i'\widehat c_{i,t}$ is at most $\delta T$ by the second inequality in \eqref{eq: ROI-preparation-bounds}. The new outcomes on $F_{0,K}$ have value at most $\eta\vmax T/s$. Together with the third inequality in \eqref{eq: ROI-preparation-bounds}, we obtain
\begin{align*}
\frac{\mu_i'}{1+\mu_i'}
\frac{\sum_{t\notin\mathcal T}(u_{i,t}'-c_{i,t}')}
{\sum_{t\notin\mathcal T}u_{i,t}'}
&\leq\frac{2s}{\eta'\vmax}\left[\frac{2\delta}{1+\mu_i'}+\frac{5\beta}s
+\frac{2\eta\vmax}s\right]
\leq\frac{4c}{s\eta'\vmax}
+\frac{10\beta}{\eta'\vmax}+\frac{4\eta}{\eta'}\\
&\leq\frac{\gamma}{32M}+\frac{\gamma}{32M}+\frac{\gamma}{16M}
=\frac{\gamma}{8M}.
\end{align*}
\end{itemize}

We take $\delta_1=\gamma/8$, $\delta_2=\beta\kappa^2$, and $\delta_3=\gamma/(8M)$, and check \Cref{lem: low-regret-implies-approximate-ABE}'s requirement on $\delta_2$. Since each bidder receives value $\eta'\vmax/s$ from $F_0'$ in each round and $\int v_i\,\dd\widetilde F(\vec v)\leq5/s$, we have for all $i\in\mathcal I$,
\begin{equation*}
\delta_2=\beta\kappa^2
=\frac{\gamma^3\eta'\vmax}{320M(M-1)^2}
\leq\frac{\gamma^3\sum_{t\notin\mathcal T}u_{i,t}'}
{32M(M-1)^2(T-|\mathcal T|)\int v_i\,\dd\widetilde F(\vec v)}.
\end{equation*}

By \Cref{lem: low-regret-implies-approximate-ABE}, we obtain a $\gamma$-approximate fluid autobidding equilibrium with cap $M$ for $\widetilde F$ with probability at least $3/4$. By our remark in \Cref{app: ROI-discretization} and $6\eta\leq(1-\gamma)\eta'$, \Cref{lemma: roi-fluid-hardness} also applies to $\widetilde F$, so this solves the source \ppad-hard problem. Therefore, we get a polynomial-time randomized algorithm for the $\mathsf{PPAD}$-hard problem, contradicting the assumption of \Cref{thm: fixed-polynomial-ROI-regret}. For details of the computation model, see \Cref{app: ROI-complexity-model}.
\end{proof}

\subsection{Computational Details}
\label{app: ROI-complexity-model}

Similar to \Cref{app: complexity-model}, we argue that the reduction in the proof of \Cref{thm: fixed-polynomial-ROI-regret} can be implemented as a polynomial-time randomized algorithm that succeeds with probability at least $2/3$. We can check the encoding length and running time similarly. 
 
 Different from the budget setting, the reduction additionally computes an optimal fluid dual solution $\mu_i^{*,\widetilde F}$ for each $i\in\mathcal I$, which gives the exceptional rounds $\mathcal T_2$ and the fixed multipliers $\mu_i'=\max\{\mu_i^{*,\widetilde F},1/(M-1)\}$ used in \Cref{lem: low-regret-implies-approximate-ABE}. Since $\widetilde F$ has finite support, the fluid dual objective is convex and piecewise linear with at most $T|\supp(\widetilde F)|$ rational breakpoints, so we find a minimizer by evaluating it at these breakpoints and at the endpoints of $[0,H]$. 
 The rest of the reduction is the same as in \Cref{app: complexity-model}.

\
\end{document}